\documentclass[11pt,letterpaper]{article}  
\usepackage[margin=1in]{geometry}
\usepackage[utf8]{inputenc}
\usepackage{xcolor}
\usepackage{hyperref}
\usepackage{amsmath,amsthm,amssymb}
\usepackage{physics}
\usepackage{mathtools}
\usepackage{enumitem}
\usepackage{comment}
\usepackage{graphicx}
\usepackage{authblk}
\usepackage{subcaption}
\usepackage{float}
\usepackage{thm-restate}
\usepackage{bm}    \usepackage[normalem]{ulem}    \usepackage{calc}
\usepackage{tocloft}
\usepackage[backend=biber, eprint=true, url=false, style=alphabetic, maxnames=5, minnames=3]{biblatex}
\numberwithin{equation}{section}

\usepackage{cleveref}
\Crefformat{equation}{(#2#1#3)}

\usepackage{tikz}
\usetikzlibrary{positioning,calc,decorations.pathreplacing,arrows.meta}

\definecolor{darkgreen}{rgb}{0,0.5,0}
\hypersetup{
	unicode=false,          colorlinks=true,        linkcolor=red,          citecolor=darkgreen,    filecolor=magenta,      urlcolor=cyan           }

\theoremstyle{plain}

\newtheorem{theorem}{Theorem}[section]
\newtheorem{proposition}[theorem]{Proposition}
\newtheorem{lemma}[theorem]{Lemma}

\newtheorem{claim}[theorem]{Claim}
\newtheorem{fact}[theorem]{Fact}
\theoremstyle{definition}
\newtheorem{definition}[theorem]{Definition}
\newtheorem*{remark}{Remark} 

\newcommand{\dee}{\mathrm{d}}
\newcommand{\rmextra}{\mathrm{extra}}
\newcommand{\rmmain}{\mathrm{main}}
\newcommand{\rmtarget}{\mathrm{target}}
\newcommand{\rmsim}{\mathrm{sim}}
\newcommand{\rmmax}{\mathrm{max}}

\newcommand{\calE}{\mathcal{E}}
\newcommand{\calK}{\mathcal{K}}
\newcommand{\calM}{\mathcal{M}}

\newcommand{\calH}{\mathcal{H}}

\newcommand{\calX}{\mathcal{X}}
\newcommand{\binary}{\{0,1\}}
\newcommand{\tensor}{\otimes}
\newcommand{\pathto}{\longrightarrow}
\newcommand{\dwfunc}{f_{\rm dw}}
\newcommand{\nGHV}{\mathfrak{n}}  \newcommand{\mGHV}{\mathfrak{m}}  \newcommand{\nA}{{n_{\sf A}}}
\newcommand{\nB}{{n_{\sf B}}}
\newcommand{\DeltaB}{{\Delta_{\sf B}}}
\newcommand{\XB}{{X_{\sf B}}}
\newcommand{\HL}{{H_{\sf L}}}
\newcommand{\QL}{{Q_{\sf L}}}
\newcommand{\DeltaL}{{\Delta_{\sf L}}}
\newcommand{\XL}{{X_{\sf L}}}
\newcommand{\ZL}{{Z_{\sf L}}}
\newcommand{\epsilonL}{{\epsilon_{\sf L}}}
\newcommand{\DeltaC}{{\Delta_{\sf C}}}
\newcommand{\HC}{{H_{\sf C}}}
\newcommand{\DC}{{D_{\sf C}}}
\newcommand{\HD}{{H_{\sf C'}}}
\newcommand{\epsilonD}{\eta}
\newcommand{\hatH}{\hat{H}}
\newcommand{\hatX}{\hat{X}}
\newcommand{\hatZ}{\hat{Z}}
\newcommand{\hatP}{\hat{P}}
\newcommand{\hatx}{\hat{x}}
\newcommand{\haty}{\hat{y}}
\newcommand{\hatz}{\hat{z}}
\newcommand{\hatD}{\hat{D}}
\newcommand{\hatDp}{\hat{D}_{\rm p}}
\newcommand{\hatDr}{\hat{D}_{\rm r}}
\newcommand{\calW}{\mathcal{W}}
\newcommand{\eigen}{\mu}
\newcommand{\hatzero}{\hat{0}}
\newcommand{\hatone}{\hat{1}}
\newcommand{\tildezero}{\tilde{0}}
\newcommand{\tildeone}{\tilde{1}}
\newcommand{\remzero}{\mathring{0}}
\newcommand{\remone}{\mathring{1}}
\newcommand{\tildex}{\tilde{x}}
\newcommand{\remx}{\mathring{x}}
\newcommand{\calI}{\mathcal{I}}
\newcommand{\calJ}{\mathcal{J}}
\newcommand{\mineng}{\eigen_0}
\newcommand{\baseeng}{\eta}
\newcommand{\frakS}{\mathfrak{S}}
\newcommand{\tend}{t_{\rm f}}
\renewcommand{\tilde}{\widetilde}
\renewcommand{\hat}{\widehat}
\renewcommand{\bar}{\overline}
\renewcommand{\epsilon}{\varepsilon}
\newcommand{\R}{\mathbb{R}}

\newlength{\equalwidth}
\newcommand{\xlongequal}[1]{\mathrel{\setlength{\equalwidth}{\widthof{\ensuremath{\scriptstyle #1}}}\overset{\scriptstyle #1}{\resizebox{\equalwidth}{\height}{$=$}}}}

\newcommand{\bs}[1]{\bm{#1}}

\DeclareMathOperator{\im}{Im}
\DeclareMathOperator{\dist}{dist}
\DeclareMathOperator{\poly}{poly}
\DeclareMathOperator{\sgn}{sgn}
\DeclareMathOperator{\spn}{span}
\DeclareMathOperator*{\argmin}{arg\,min}

\newcommand{\Crefitem}[2]{\Cref{#1}{~\ref{#2}}}
\newcommand{\email}[1]{\href{mailto:#1}{\texttt{\detokenize{#1}}}}

\makeatletter
\newcommand{\sharedfootnotemark}{\footnotemark
\protected@xdef\@savedfn{\number\numexpr\value{footnote}\relax}}
\newcommand{\sharedfootnoteref}{\hyperlink{fn\@savedfn}{\textsuperscript{\@savedfn}}}
\makeatother

\newif\ifblind

\newcommand{\apath}{Hamiltonian path}
\newcommand{\apaths}{\apath s}
\newcommand{\aPath}{Hamiltonian Path}
\newcommand{\sdg}{Schr\"odinger}

\allowdisplaybreaks

 \newcommand{\figdiamond}{
\begin{subfigure}[c]{0.45\textwidth}
        \centering
        \raisebox{1.2em}{
        \begin{tikzpicture}[scale=1.0]
\tikzstyle{vertex} = [circle, fill=black, inner sep=2pt, draw=black]

\node[vertex] (A) at (-1.5,0) {};
            \node[vertex] (B) at (1.5,0) {};

\draw (A) -- node[above] {$-J$} (B);
        \end{tikzpicture}
        }
\end{subfigure}
\begin{subfigure}[c]{0.45\textwidth}
        \centering
        \begin{tikzpicture}[scale=1.0]
\tikzstyle{vertex} = [circle, fill=black, inner sep=2pt, draw=black]
\tikzstyle{vertexEmpty} = [circle, fill=none, inner sep=2pt, draw=black]

\node[vertexEmpty] (top)    at ( 0,  0.5) {};
            \node[vertex]      (left)   at (-2,  0.0) {};
            \node[vertex]      (right)  at ( 2,  0.0) {};
            \node[vertexEmpty] (bottom) at ( 0, -0.5) {};

\draw[color=orange] (top)    -- node[above, xshift=-1em]  {\footnotesize $-\sqrt{\Delta/2}$} (left);
            \draw[color=orange] (top)    -- node[above] {\footnotesize $-\sqrt{\Delta/2}$} (right);
            \draw[color=orange] (bottom) -- node[below, xshift=-1em]  {\footnotesize $-\sqrt{\Delta/2}$} (left);
            \draw[color=orange] (bottom) -- node[below] {\footnotesize $-\sqrt{\Delta/2}$} (right);

\draw[color=blue] (top) to[out=130,in=50,looseness=16] 
                node[above] {\footnotesize $\Delta J^{-1}$} (top);
            \draw[color=blue] (bottom) to[out=-130,in=-50,looseness=16] 
                node[below] {\footnotesize $\Delta J^{-1}$} (bottom);

        \end{tikzpicture}
\end{subfigure}
}

\newcommand{\aloop}[1]{
    \draw[color=blue] (#1) to[out=50,   in=130,   looseness=16] node[above] {\tiny $2\Delta J^{-1}$} (#1);
}

\newcommand{\figtriangle}{
\begin{subfigure}[c]{0.35\textwidth}
        \centering
        \raisebox{1.5em}{
        \begin{tikzpicture}[scale=1.0]
\tikzstyle{vertex} = [circle, fill=black, inner sep=2pt, draw=black]

\node[vertex] (A) at (-1.5,0) {};
            \node[vertex] (B) at (1.5,0) {};

\draw (A) -- node[above] {$-J$} (B);
        \end{tikzpicture}
        }
\end{subfigure}
    \hfill
\begin{subfigure}[c]{0.55\textwidth}
        \centering
        \begin{tikzpicture}[scale=1.0]
\tikzstyle{blackVertex} = [circle, fill=black, draw=black, inner sep=2pt]
            \tikzstyle{whiteVertex} = [circle, fill=white, draw=black, inner sep=2pt]

\node[blackVertex] (L) at (-4, 0) {};
            \node[blackVertex] (R) at ( 4, 0) {};

            \node[whiteVertex] (B1) at ($(L)!2/7.5!(R)$) {};
            \node[whiteVertex] (B2) at ($(L)!3/7.5!(R)$) {};
            \node[whiteVertex] (B3) at ($(L)!4.5/7.5!(R)$) {};
            \node[whiteVertex] (B4) at ($(L)!5.5/7.5!(R)$) {};

            \draw[thick, orange] (L) -- (B1) node[midway, below] {\tiny $- \sqrt{\Delta (k +  1)}$};
            \draw[thick, blue]  (B1) -- (B2) node[midway, below] {\tiny $-\Delta J^{-1}$};
            \draw[thick, blue, dashed] (B2) -- (B3) node[midway, below] {\(\cdots\)};
            \draw[thick, blue]  (B3) -- (B4) node[midway, below] {\tiny $-\Delta J^{-1}$}; 
            \draw[thick, orange] (B4) -- (R) node[midway, below] {\tiny $- \sqrt{\Delta (k +  1)}$};

\foreach \n in {B1, B2, B3, B4} {
                \aloop{\n}
            }
            \draw [decorate,decoration={brace,amplitude=12pt,mirror}] 
            ([yshift=-6pt]B1.south west) -- ([yshift=-6pt]B4.south east) 
            node[midway, below=10pt] {\small $k$ nodes};

        \end{tikzpicture}
\end{subfigure}
}

\newcommand{\figspec}{
    \begin{subfigure}{0.45\textwidth}
        \includegraphics[trim=0 1.2em 0 0, clip, width=0.9\textwidth]{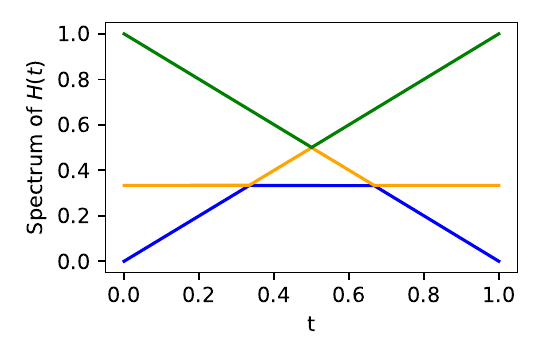}
        \caption{$\lambda = 0$}
        \label{fig:spec_0}
    \end{subfigure}
    \begin{subfigure}{0.45\textwidth}
        \includegraphics[trim=0 1.2em 0 0, clip, width=0.9\textwidth]{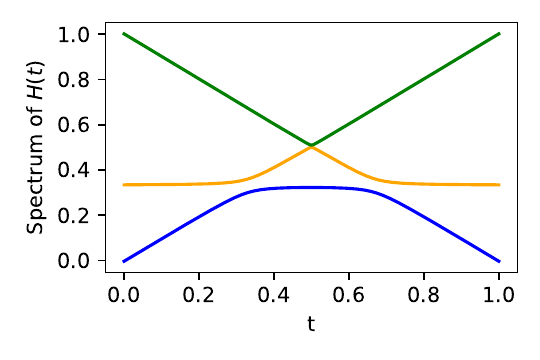}
        \caption{$\lambda = 0.03$}
        \label{fig:spec_1}
    \end{subfigure}
}

\newcommand{\figdwspec}{
    \begin{subfigure}{0.45\textwidth}
        \includegraphics[trim=0 1.2em 0 0, clip, width=0.9\textwidth]{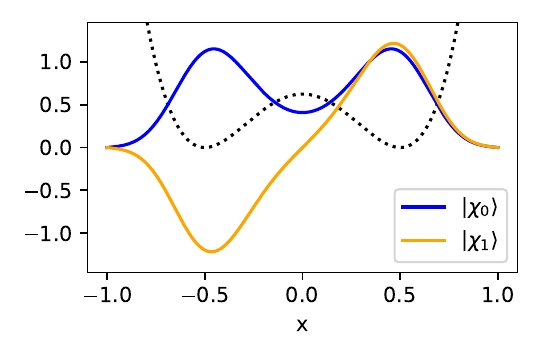}
        \caption{Low-energy eigenstates of $\hatX$}
        \label{fig:dwspec_0}
    \end{subfigure}
    \begin{subfigure}{0.45\textwidth}
        \hspace{0.002\textwidth}
        \includegraphics[trim=0 1.2em 0 0, clip, width=0.9\textwidth]{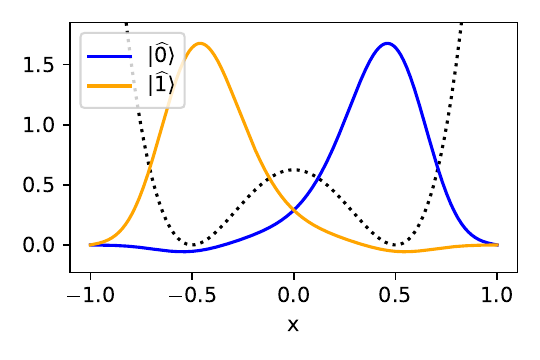}
        \caption{States $\ket{\hatzero}$ and $\ket{\hatone}$}
        \label{fig:dwspec_1}
    \end{subfigure}
}

\newcommand{\figheatmap}{
    \includegraphics[trim=0 1.2em 0 0, clip, width=0.5\textwidth]{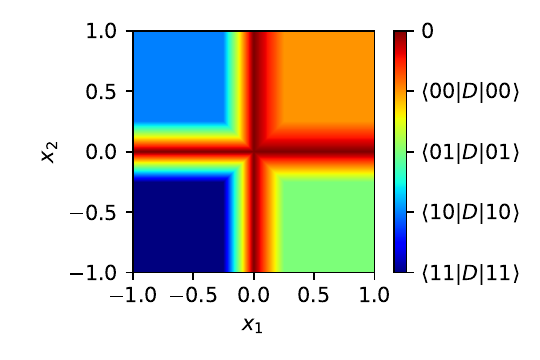}
}

\newcommand{\figdiagram}{
    \begin{tikzpicture}[stdnode/.style={draw, rectangle, fill=white, minimum height=1.8em, align=center},
        arrows={[scale=1]}]
    
\pgfdeclarelayer{background}
    \pgfdeclarelayer{main}
    \pgfsetlayers{background,main} 

\newlength{\nodehdist}
    \setlength{\nodehdist}{10em}
    \newlength{\nodevdist}
    \setlength{\nodevdist}{2em}
    \newlength{\dashwdist}
    \setlength{\dashwdist}{25em}
    \newlength{\dashedist}
    \setlength{\dashedist}{24em}

    \node[stdnode] (GHV) at (0,0) {\Cref{thm:GHV21}~\cite{gilyen2021sub}};
    \node[stdnode] (PA) at (-\nodehdist,-\nodevdist) {\Cref{lem:pathA}};
    \node[stdnode] (PB) at (0,-2\nodevdist) {\Cref{lem:pathB}};
    \node[stdnode] (PB2) at (-\nodehdist,-3\nodevdist) {\Cref{lem:pathB2}};
    \node[stdnode] (TOY) at (\nodehdist,-3\nodevdist) {\Cref{lem:linearTFI}};
    \node[stdnode] (PC) at (0,-4\nodevdist) {\Cref{lem:pathC}};
    \node[stdnode] (PD) at (-\nodehdist,-5\nodevdist) {\Cref{lem:pathD}};
    \node[stdnode] (T1) at (0,-6\nodevdist) {\Cref{thm:main_informal} (\ref{thm:main})};
    \node[stdnode] (SDG) at (\nodehdist,-7\nodevdist) {\Cref{lem:tosdg}};
    \node[stdnode] (ZLLW) at (2\nodehdist,-7\nodevdist) {\Cref{clm:properties-of-X}~\cite{zheng2024computational}};
    \node[stdnode] (T2) at (0,-8\nodevdist) {\Cref{thm:qhd_informal} (\ref{thm:qhd})};

\draw[-Latex] (GHV) -- (PA);
    \draw[-Latex] (PA) -- (PB);
    \draw[-Latex] (PB) -- (PB2);
    \draw[-Latex] (PB2) -- (PC);
    \draw[-Latex] (TOY) -- (PC);
    \draw[-Latex] (PC) -- (PD);
    \draw[-Latex] (PD) -- (T1);
    \draw[-Latex] (PB2) -- (T1);
    \draw[-Latex] (T1) -- (T2);
    \draw[-Latex] (SDG) -- (T2);
    \draw[-Latex] (ZLLW) -- (SDG);

\begin{pgfonlayer}{background}
        \draw[dashed] ($(PB)+(-\dashwdist,0)$) -- ($(PB)+(\dashedist,0)$);
        \draw[dashed] ($(PC)+(-\dashwdist,0)$) -- ($(PC)+(\dashedist,0)$);
        \draw[dashed] ($(T1)+(-\dashwdist,0)$) -- ($(T1)+(\dashedist,0)$);
    \end{pgfonlayer}
    
\node[align=center] (S1) at (-2\nodehdist,-\nodevdist) {\Cref{sec:overview-pert} \\ Reducing to TFD};
    \node[align=center] (S2) at (-2\nodehdist,-3\nodevdist) {\Cref{sec:overview-linearize} \\ Path Linearization};
    \node[align=center] (S3) at (-2\nodehdist,-5\nodevdist) {\Cref{sec:overview-tilt} and \ref{sec:pet} \\ Path Tilt and Extension};
    \node[align=center] (S4) at (-2\nodehdist,-7\nodevdist) {\Cref{sec:overview-qhd} \\ Generalizing to QHD};
    
    \end{tikzpicture}
} 

\ifblind
    \title{(Sub)Exponential Quantum Speedup for Optimization}
\else
    \title{(Sub)Exponential Quantum Speedup for Optimization}
\fi

\makeatletter
\begin{document}

\ifblind
	\author{Anonymous Authors}
\else
	\author[1,2]{Jiaqi Leng\thanks{\email{jiaqil@berkeley.edu}}}
	\author[3]{Kewen Wu\thanks{\email{shlw_kevin@hotmail.com}}}
	\author[4,5]{Xiaodi Wu\thanks{\email{xiaodiwu@umd.edu}}}
	\author[4,5]{Yufan Zheng\thanks{\email{yfzheng@umd.edu}}}

	\affil[1]{Simons Institute for the Theory of Computing, University of California, Berkeley}
	\affil[2]{Department of Mathematics, University of California, Berkeley}
	\affil[3]{Department of EECS, University of California, Berkeley}
	\affil[4]{Department of Computer Science, University of Maryland, College Park}
	\affil[5]{Joint Center for Quantum Information and Computer Science, University of Maryland}
\fi

\date{}
\maketitle

\thispagestyle{empty}

\begin{abstract}
	We demonstrate provable (sub)exponential quantum speedups in both discrete and continuous optimization, achieved through simple and natural quantum optimization algorithms, namely the quantum adiabatic algorithm for discrete optimization and quantum Hamiltonian descent for continuous optimization.

	Our result builds on the Gilyén--Hastings--Vazirani (sub)exponential oracle separation for adiabatic quantum computing.
	With a sequence of perturbative reductions, we compile their construction into two standalone objective functions, whose oracles can be directly leveraged by the plain adiabatic evolution and Schrödinger operator evolution for discrete and continuous optimization, respectively.

\end{abstract}

\newpage
\tableofcontents
\thispagestyle{empty}

\newpage
\setcounter{page}{1}
\section{Introduction}\label{sec:intro}

Identifying exponential quantum speedup has been a central focus in quantum computing. Prominent early examples, such as Shor’s factorization algorithm~\cite{shor-factorization-97} and simulation of quantum dynamics~\cite{Feynman-82,lloyd1996universal}, which demonstrate exponential speedups over the best-known classical methods, have played a crucial role in advancing and shaping the field.

Provable quantum and classical separations are known in more restricted settings, typically within the query model, where the problem input is provided to both quantum and classical algorithms through oracle access.
Exponential separations between quantum and classical are known for promise problems over Boolean functions such as Simon's problem~\cite{Simon1997}, Forrelation~\cite{Aaronson2018,Raz2019}, and the recent Yamakawa--Zhandry problem~\cite{Yamakawa2024}
against a random oracle,
as well as for graph problems such as Glued Trees~\cite{Childs2003} and for adiabatic quantum simulation tasks with no sign problem~\cite{gilyen2021sub}. We refer readers to~\cite{aaronson2022much} for a more comprehensive survey.

This paper focuses on examining the exponential separation between quantum and classical optimization, which has been a major open question\footnote{It is worth noting, as also pointed out by~\cite{jordan2025optimizationdecodedquantuminterferometry}, that one can convert quantum-classical separations over Boolean functions or search problems into trivial optimization problems in a generic way.
	It is also possible to design optimization problems that encode factoring and exhibit super-polynomial quantum advantages based on Shor's algorithm and the hardness of factorization~\cite{optimization-shor, szegedy2022quantumadvantagecombinatorialoptimization}. } in quantum algorithms and a promising venue for real-world quantum applications.
A seminal work in this pursuit is the \emph{adiabatic quantum optimization} introduced by
\cite{farhi2000quantum},
which has emerged as a distinct paradigm for solving \emph{discrete} and \emph{combinatorial} optimization problems using quantum mechanics.

Adiabatic algorithms operate by gradually evolving a quantum system from an initial Hamiltonian, whose ground state is easy to prepare, to a final Hamiltonian, whose ground state encodes the solution to the optimization problem. By ensuring that this evolution occurs slowly—following a continuous-time adiabatic process—the system remains in its ground state and eventually converges to the optimal solution by construction. This simple, yet natural, physics-inspired quantum algorithm has led to extensive empirical and experimental studies, as well as several important theoretical insights (notably~\cite{PhysRevA.65.042308} and \cite{Aharonov2007}); see \cite{Albash2018} for a comprehensive review. However, precisely characterizing the power of adiabatic quantum optimization remains a significant challenge~\cite{vanDam2001}. In particular, it is still unclear whether adiabatic quantum optimization can achieve exponential speedups over all classical methods for certain optimization problems.

The Quantum Approximate Optimization Algorithm (QAOA) \cite{farhi2014quantumapproximateoptimizationalgorithm} is a discrete-time, gate-based alternative to the adiabatic quantum optimization, offering the same near-term feasibility for experimental implementation (e.g.,~\cite{Google-QAOA-Experiment}) while being more accessible for in-depth analysis as first shown in~\cite{farhi2014quantumapproximateoptimizationalgorithm,farhi2015quantumapproximateoptimizationalgorithm}. Subsequent works (e.g.,~\cite{barak_et_al:LIPIcs.APPROX-RANDOM.2015.110,BGMZ-22,farhi2025lowerboundingmaxcuthigh}) have alternately advanced the frontiers of both QAOA and classical algorithms for certain constraint satisfaction problems, where the true potential of QAOA for these problems remains an active area of research. For more details, we refer interested readers to~\cite{BLEKOS20241, PhysRevX.10.021067}.

Very recently, an algorithm called Decoded Quantum Interferometry (DQI)~\cite{jordan2025optimizationdecodedquantuminterferometry} has introduced a novel approach to quantum optimization by providing a generic reduction from optimization problems to classical decoding tasks. Through both theoretical and empirical analyses, researchers have identified discrete optimization problems where the proposed DQI algorithm achieves a certain approximation ratio that no known efficient classical algorithm can match.
This leads to an apparent exponential quantum speedup while proving classical lower bounds for these optimization tasks still remains challenging.

Research on quantum algorithms for \emph{continuous} optimization has gained momentum over the past decade, with many known quantum speedups arising from replacing computational steps in classical algorithms with quantum-accelerated counterparts.
Successful examples include quantum speedups for generic convex optimization (e.g.,~\cite{vanApeldoorn2020convex,Chakrabarti2020convex}), semidefinite programming (e.g.,~\cite{brandao2017quantum,brandao2017quantum2,van2019improvements,van2019quantum}), interior-point methods (e.g.,~\cite{kerenidis2020quantum2,mohammadisiahroudi2022efficient,augustino2023quantum,wu2023inexact,apers2023quantum}), and so on,  where the speedup stems from faster quantum subroutines for, e.g., estimating the gradient~\cite{Jordan2005gradient}, sampling from a Gibbs distribution~\cite{PhysRevLett.103.220502}, or solving a linear equation system~\cite{PhysRevLett.103.150502}.
To our best knowledge, these developments have only achieved polynomial quantum speedups over their classical counterparts.

Quantum Hamiltonian Descent (QHD)~\cite{leng2023quantum} adopts a fundamentally different approach to continuous optimization by encoding the entire optimization process into continuous-time quantum dynamics. Derived through path-integral quantization of a (classical) Hamiltonian formulation of Nesterov’s accelerated gradient descent~\cite{su2016differential,wibisono2016variational}, QHD translates the optimization process into the evolution of a Schrödinger operator.
Like adiabatic algorithms, QHD is inherently simple and natural, allows near-term experimental implementation~\cite{leng2023quantum,qhd-opt}. However, unlike adiabatic methods—which rely on remaining in the ground state all the time—empirical studies suggest that QHD allows the quantum system to explore higher energy states and hence enables faster convergence.
Follow-up work~\cite{leng2023separation} proves that QHD can efficiently solve a class of high-dimensional, unconstrained non-convex optimization problems—with exponentially many local minima—in polynomial time, where empirical studies further reveal that leading classical optimizers (including Gurobi) require super-polynomial time to solve such problems.
Precise theoretical understanding of convergence rates of QHD and subsequent proposals (e.g.,~\cite{chen2023quantum,augustino2024quantumcentralpathalgorithm,catli2025exponentiallybetterboundsquantum,leng2025qhd,chakrabarti2025speedupsconvexoptimizationquantum}) that encode various optimization tasks into quantum dynamics remains limited to convex problems, where quantum speedup is limited.
No provable super-polynomial quantum speedup, as a result, has yet been established for continuous optimization.

Despite extensive research (see~\cite{abbas2024challenges} for a comprehensive survey) and considerable practical interest, the existence of \emph{provable exponential separation} between quantum and classical optimization remains open, to the best of our knowledge.

\paragraph{Contribution.} We demonstrate the possibility of (sub)exponential quantum-classical separation in optimization. Specifically, we construct oracle-based families of objective functions for both discrete and continuous optimization that require exponentially many queries to solve classically. Conversely, we demonstrate that optimal solutions to these problems can be efficiently found by the adiabatic quantum optimization and Quantum Hamiltonian Descent, respectively.

Without loss of generality, we study the \emph{minimization} problem over $\{0,1\}^n$ for an objective function $f:\{0,1\}^n \rightarrow [0,1]$ for discrete optimization, and over a bounded subset
$\calX \subset \mathbb{R}^n$ for an objective function $f:\calX \to [0,1]$ for continuous optimization.
We do not attempt to optimize any specific polynomial bound either given the focus is on exponential separation.

We consider the following \emph{plain} version of adiabatic quantum optimization for discrete optimization that follows a linear Hamiltonian path $H(t)$ from $H_{\rm M}$ to $D$:\begin{equation} \label{eqn:plain_path_adiabtic}
	H(t) := (1-t) H_{\rm M} + t D,\quad \forall t \in [0,1],
\end{equation}
where $H(t)$ is an $n$-qubit Hamiltonian for discrete optimization over $\{0,1\}^n$.
$H_{\rm M}$ is conventionally called the \emph{mixer} Hamiltonian and $D$ is a diagonal Hamiltonian that encodes the information of the objective function $f:\{0,1\}^n \rightarrow [0,1]$, namely, $D=\sum_{x\in\{0,1\}^n} f(x) \ketbra{x}{x} $.

The original proposal of adiabatic optimization~\cite{Farhi2001} makes use of
$H_{\rm M} \propto \sum_i X_i$ and $D$ is additionally $k$-local which refers to constraint satisfaction problems with clauses containing $k$ Boolean variables.
General forms of mixer Hamiltonians have since been introduced (e.g., see~\cite{QAOA-mixer-nasa}) and a general form of the objective function can be implemented with quantum oracles to $f(\cdot)$.
In particular, in our construction, we will use a \emph{transverse field Ising} (TFI) Hamiltonian as the mixer (see Section~\ref{sec:prelim} for a precise definition).

Adiabatic quantum optimization evolves from the ground state of $H_{\rm M}$, which needs to be efficiently preparable, along the Hamiltonian path $H(t)$ in \cref{eqn:plain_path_adiabtic} slowly enough according to the quantum adiabatic theorem (see Appendix~\ref{sec:qat}) so that the quantum system remains in the ground state along the Hamiltonian path and eventually arrives at the ground state of $D$ which refers to the optimal solution.
The evolution speed can be captured by the total evolution time $T$ which uniformly scale $H(t)$ over $[0,1]$ to $[0,T]$. (Or effectively, evolving with Hamiltonian $TH(t)$ over $[0,1]$.)
The evolution can be efficiently implemented by standard digital quantum simulation with oracle access to $D$, which leads to its complexity bound.

\begin{theorem}[Informal version of \Cref{thm:main}] \label{thm:main_informal}
	There exists a family of discrete optimization problems with objective functions $f:\binary^n \to [0,1]$ such that, for any $n$,  the following holds:
	\begin{itemize}
		\item Any classical algorithm requires $\exp(n^{\Omega(1)})$ queries to $f(\cdot)$ to find an $x$ such that $f(x) - f(x^*) \leq 1/\poly(n)$, where $x^*$ is the global minimizer of $f$.
		\item There exists a uniformly constructible TFI Hamiltonian $H_{\rm TFI}$ only depending on $n$, which has an efficiently preparable ground state.
		      The plain adiabatic evolution, i.e.,
		      \begin{equation}
			      i \frac{\dee}{\dee t} \ket{\psi(t)} = T H(t) \ket{\psi(t)}, \quad H(t) \coloneqq (1-t)H_{\rm TFI} + t\sum_{x \in \binary^n} f(x) \ketbra{x}{x}, \quad t \in [0,1],
		      \end{equation}
		      finds the global minimizer $x^*$ with $T = \poly(n)$. Moreover, its digital implementation takes $\poly(n)$ queries to $f(\cdot)$ and has $\poly(n)$ gate complexity.
	\end{itemize}
\end{theorem}

For continuous optimization, given any $n$-dimensional objective function $f:\calX \to [0,1]$
on a bounded subset $\calX \subset \mathbb{R}^n$,
with an appropriate time rescaling, Quantum Hamiltonian Descent solves the minimization of $f(\cdot)$ with the following quantum dynamics $H(t)$ for some finite time $T$,  where
\begin{equation} \label{eqn:qhd_plain}
	H(t) = - \Delta + s(t) f(x),\quad  \forall t\geq 0.
\end{equation}
Here $\Delta$ is the Laplace operator, $s(t)$ is an increasing function over $t$ where $s(0) \ll 1$ and $s(T) \gg 1$, and $f(x)$ is a diagonal operator induced by the objective function $f(\cdot)$.
Indeed, compared with \cref{eqn:plain_path_adiabtic}, $H(t)$ in \cref{eqn:qhd_plain} is an unbounded operator in $n$-dimensional real space that takes the shape of the \sdg ~operator.
Intuitively speaking,  QHD's dynamics starts with dominant kinetic energy (referring to $-\Delta$), and gradually shifts to dominant potential energy (referring to $f(x)$). This transition helps QHD explore the entire solution space in the early stage, and enables its quick convergence in the later stage.
(See an explanatory theory for QHD in~\cite{leng2023quantum}.)

QHD evolves from an efficiently preparable state according to the Hamiltonian dynamics in \cref{eqn:qhd_plain} for a specific $s(\cdot)$, an efficient digital simulation of which has only been explored for a few domains $\calX$~\cite{Childs2022} (and is fortunately known for our case).
It is worth mentioning that, due to a technique called Hamiltonian embedding~\cite{leng2024expanding}, a low-precision spatially discretized QHD dynamics has been experimentally implemented on an analog quantum Ising simulator with over 5000 spins~\cite{leng2023quantum}.
Empirical studies of QHD from both real-machine experiments and classical simulation suggest that QHD could leverage higher energy levels for faster convergence,
although the adiabatic theorem remains the only known tool in the theoretical analysis of QHD for the nonconvex setting, which also applies to this work.

To understand the computational power of the \sdg\ operator, a recent work~\cite{zheng2024computational} establishes a generic correspondence between transverse field Ising Hamiltonians and \sdg\ operators.
By generalizing their technique to our setting, we construct from \Cref{thm:main_informal} a family of continuous functions that exhibit (sub)exponential quantum speedup with dynamics in \cref{eqn:qhd_plain}.

\begin{theorem}[Informal version of \Cref{thm:qhd}] \label{thm:qhd_informal}
	There exists a family of continuous functions $f:\calX \to [0,1]$ with a bounded box domain $\calX \subset \mathbb{R}^n$ such that, for any $n$, the following holds:
	\begin{itemize}
		\item Any classical algorithm requires $\exp(n^{\Omega(1)})$ queries to $f(\cdot)$ to find an $x$ such that $f(x) - f(x^*) \leq 1/\poly(n)$, where $x^*$ is the global minimizer of $f$.
		\item There exists an auxiliary function $g(x)$, a time-dependent function $\nu(t) = e^{-\Theta(\sqrt{t})}$, both of which only depend on $n$. The QHD evolution, i.e.,
		      \begin{equation} \label{eqn:qhd_real}
			      i \frac{\dee}{\dee t} \ket{\psi(t)} = T H(t) \ket{\psi(t)}, \quad H(t) \coloneqq -\Delta + tf(x) + \nu(t)g(x),\quad t \in [0,\tend],
		      \end{equation}
		      finds the global minimizer $x^*$ with $T,\tend = \poly(n)$ from an efficiently preparable ground state of $H(0)$.
		      Its digital implementation takes $\poly(n)$ queries to $f(\cdot)$ and has $\poly(n)$ gate complexity.
	\end{itemize}
\end{theorem}
We note that the specific dynamics in \cref{eqn:qhd_real} deviates slightly from \cref{eqn:qhd_plain} due to a quick-vanishing term $\nu(t)g(x)$.
It can be interpreted as some auxiliary information to boost the performance of QHD's dynamics in the early stage.

\paragraph{Technical Contribution.} At a high level, our construction stems from the Gilyén--Hastings--Vazirani (GHV) (sub)exponential oracle separation for adiabatic quantum computing with no sign problem~\cite{gilyen2021sub}.
Despite of being an adiabatic evolution, GHV's construction concerns with a graph theoretical problem similar to the Glued Tree problem~\cite{Childs2003}, which is neither a standalone optimization problem nor solved by the plain adiabatic evolution for optimization.

Our main technical contribution is a series of \emph{efficient} reductions that map adiabatic paths (i.e., time-dependent Hamiltonians) from one to another where the latter adiabatic path can faithfully simulate the former one while being closer to the desired shape (i.e., \cref{eqn:plain_path_adiabtic} or \cref{eqn:qhd_plain}).
In particular, our reductions preserve the spectral gap of adiabatic paths, thereby guaranteeing the easiness for quantum adiabatic evolution.
On the other hand, the \emph{efficiency} of our reductions preserves the computational hardness for classical algorithms; that is, any algorithm that efficiently solves the reduced problem (induced by the reduced adiabatic path) would naturally give rise to an efficient algorithm to the original problem (induced by the original adiabatic path).

Similar Hamiltonian reductions have been seen in the literature of Hamiltonian complexity~\cite{Kempe2006,Oliveira2008,bravyi2011schrieffer,bravyi2017complexity} with perturbative gadgets.
Even though we also make extensive use of perturbation theory, we believe our reduction need is unique in the literature and our solutions are novel, and can be of independent interest.
Our final reduction to continuous optimization is a tour de force adaption and customization of the established correspondence between TFI Hamiltonians and \sdg\ operators in~\cite{zheng2024computational} to our setting.
We will elaborate on these technical ideas in Section~\ref{sec:overview}.

\paragraph{Organization.}
In \Cref{sec:overview}, we give a technical overview of our ideas.
In \Cref{sec:prelim}, we define formal notation and quote well-known results.
In \Cref{sec:main}, we prove \Cref{thm:main_informal} formally, by delegating proofs of technical lemmas to  \Cref{sec:reduction_I}, \ref{sec:reduction_II}, \ref{sec:toy_shift}, \ref{sec:reduction_III}, and \ref{sec:reduction_IV}.
In \Cref{sec:qhd}, we prove \Cref{thm:qhd_informal}, which extends our result to continuous optimization.
A short exposition on adiabatic quantum computing is provided in \Cref{sec:qat}.
In \Cref{sec:intro_merit}, we present clean statements of three technical building blocks implicitly used in our proof, which  may be of independent interest.

\paragraph{Acknowledgments.}
We thank Andrew Childs, David Gosset, and Tongyang Li for insightful discussions related to this work and valuable feedback on an early draft of the paper.
JL is partially supported by the Simons Quantum Postdoctoral Fellowship, DOE QSA grant \#FP00010905, and a Simons Investigator award through Grant No. 825053. KW is supported by a Sloan Research Fellowship and NSF CAREER Award CCF-2145474. XW and YZ are partially supported by NSF CAREER Award CCF-1942837, a Sloan Research Fellowship, and the U.S. Department of Energy, Office of Science, Accelerated Research in Quantum Computing, Fundamental Algorithmic Research toward Quantum Utility (FAR-Qu). Part of the research was conducted while XW was visiting Computer Science and Artificial Intelligence Laboratory, Massachusetts Institute of Technology and the Simons Institute for the Theory of Computing, UC Berkeley.
\section{Technical Overview}\label{sec:overview}

We give a high-level technical overview of the proofs of \Cref{thm:main_informal} and \ref{thm:qhd_informal} in this section.
\Cref{fig:diagram} summarizes our overall proof structure.

\begin{figure}[!htbp]
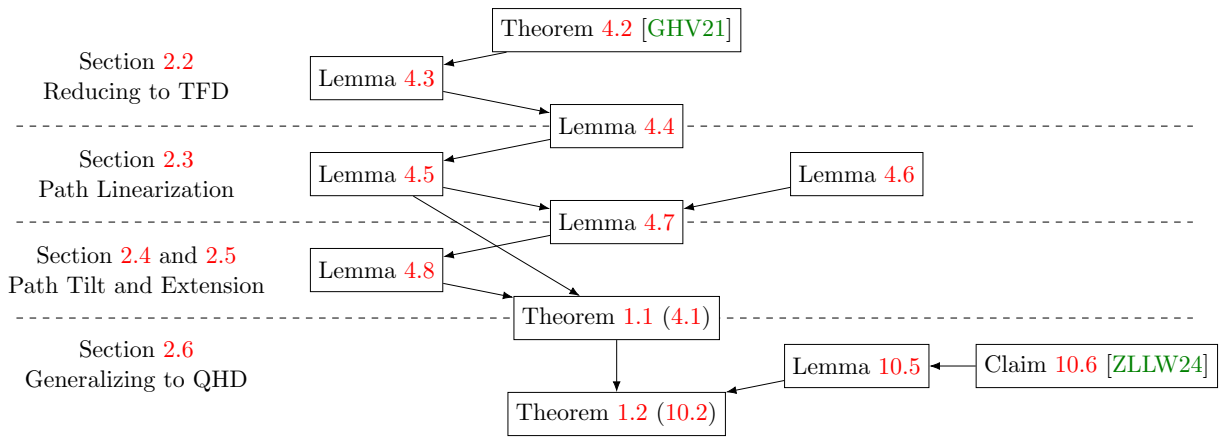

	\centering
	\resizebox{\textwidth}{!}{
		\figdiagram
	}
	\caption{Overview diagram of the proof structure illustrating the logical dependencies among the main lemmas and theorems. The dashed lines separate major steps, which are covered in the corresponding subsections of \Cref{sec:overview}.}
	\label{fig:diagram}
\end{figure}

\subsection{GHV Adiabatic Path}

GHV's construction considers an entrance-to-exit problem over oracularized graphs where the entrance is given and the exit information is encoded by the graph.
It is similar to the Glued Tree problem~\cite{Childs2003} where exponential speedup is achievable by quantum walk, which is although a non-adiabatic evolution.
With additional modifications of the oracularized graphs, GHV managed to solve the entrance-to-exit problem with the following adiabatic path:
\begin{equation}
	H_{\rm init} \pathto H_{\rm graph} \pathto H_{\rm final}.
\end{equation}
The path satisfies that (i) $H_{\rm init} \propto -\ketbra{0}{0}$ (the given entrance) and $H_{\rm final} \propto -\ketbra{u}{u}$ for some computational basis state $\ket{u}$ (the desired exit), and (ii) $H_{\rm graph}$ is sparse and stoquastic, referring to a graph whose vertexes are labeled by $n$-bit strings (i.e., the graph has size $2^n$) and which can be accessed via
an adjacency list oracle.
The minimum spectral gap along the adiabatic path remains at least $1/\poly(n)$ relative to the norms of $H_{\rm init}$, $H_{\rm graph}$, and $H_{\rm final}$, enabling an efficient quantum adiabatic evolution (i.e., $T=\poly(n)$) that locates \(\ket{u}\) with $\poly(n)$ queries.
However, any classical algorithm finding \(\ket{u}\) must make (sub)exponentially many queries to the adjacency list.

\subsection{Perturbative Reductions to TFD Hamiltonians (Lemma \ref{lem:pathA} and \ref{lem:pathB})} \label{sec:overview-pert}

Note that the entire GHV adiabatic path \(H_{\rm init} \pathto H_{\rm graph} \pathto H_{\rm final}\) consists solely of stoquastic sparse Hamiltonians, since any intermediate Hamiltonian is a linear combination of \(H_{\rm init}\), \(H_{\rm graph}\), and \(H_{\rm final}\).
In contrast, for any $t \in [0,1]$, our final target $H(t)$ in \Cref{thm:main_informal} is always a \emph{transverse field diagonal (TFD)} Hamiltonian  of the form
\begin{equation}\label{eq:overview_tfd}
	- \sum_i a_i X_i+D,
\end{equation}
where each \(X_i\) is a Pauli-\(X\) operator on the \(i\)th qubit and \(D\) is diagonal.

Therefore, a natural next step is to consider some \emph{reduction} from a stoquastic sparse Hamiltonian to a TFD Hamiltonian, which will later be applied \emph{pointwise} to the GHV adiabatic path.
The concrete meaning of reduction here is formalized in terms of the \emph{Hamiltonian simulation} framework~\cite{bravyi2017complexity,cubitt2018universal}, where perturbative reduction is usually the primary tool.

Our reduction is performed by introducing an intermediate stoquastic \emph{hypercube} Hamiltonian.
An $n$-qubit hypercube Hamiltonian is defined such that its \emph{interaction graph} is a subgraph of the $n$-dimensional hypercube (including self-loops).
Here, the interaction graph of a Hamiltonian $H$ is obtained by interpreting $H$ as a weighted adjacency matrix.

\begin{itemize}
	\item
	      By introducing a novel edge-subdivision perturbative gadget that allows the amplitude to efficiently tunnel through a \emph{polynomially} long chain, we can alter the interaction graph of a stoquastic sparse Hamiltonian so that it embeds into a hypercube by carefully (re)labeling original nodes and gadget nodes.
	      This yields a second-order perturbative reduction from stoquastic sparse Hamiltonians to stoquastic hypercube Hamiltonians (\Cref{lem:pathA}); see \Cref{fig:to-sparse} for an illustration.

	      \begin{figure}[!htbp]
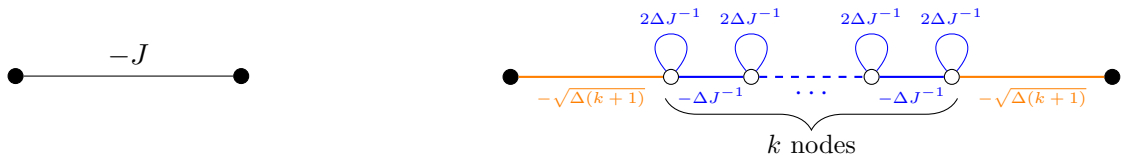

		      \centering
		      \figtriangle
		      \caption{An illustration of how we reduce an edge of weight \(-J\) (\(0 < J \ll 1\)) in a stoquastic sparse Hamiltonian to a gadget in a stoquastic hypercube Hamiltonian, where \(\Delta\) is a parameter being sufficiently large. Nodes represent computational basis states.
			      Edges with different colors correspond to different components\protect\sharedfootnotemark{} from the second-order perturbative reduction; see \Cref{lem:pathA} for details.
			      Note that $k$ can be \emph{polynomially} large since it only introduces a $\sqrt{k+1}$ overhead on weights.
		      }
		      \label{fig:to-sparse}
	      \end{figure}

	\item
	      Drawing inspiration from the reduction mapping transverse field Ising (TFI) Hamiltonians to hard-core dimer Hamiltonians in~\cite{bravyi2017complexity}, we reduce stoquastic hypercube Hamiltonians to TFD Hamiltonians, via another second-order perturbative reduction (\Cref{lem:pathB}); see \Cref{fig:to-hypercube} for an illustration.

	      \begin{figure}[!htbp]
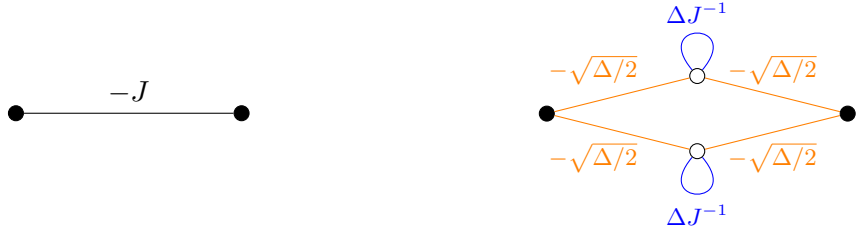

		      \centering
		      \figdiamond
		      \caption{An illustration of how we reduce an edge of weight \(-J\) (\(0 < J \ll 1\)) in a stoquastic hypercube Hamiltonian to a gadget in a TFD Hamiltonian, where \(\Delta\) is a parameter that is sufficiently large. Nodes represent computational basis states.
			      Edges with different colors correspond to different components\protect\sharedfootnoteref\ from the second-order perturbative reduction; see \Cref{lem:pathB} for details.
		      }
		      \label{fig:to-hypercube}
	      \end{figure}
\end{itemize}

\footnotetext{\hypertarget{fn\@savedfn}{}Orange edges come from \(\Delta^{1/2} V_{\rmmain}\), and blue edges come from \(\Delta H_0\) in the proof of \Cref{lem:pathA} and \Cref{lem:pathB}; the precise edge weights there differ slightly for technical reasons.}

Our reduction from stoquastic sparse Hamiltonians to TFD Hamiltonians is general: it applies to any stoquastic sparse Hamiltonian, not solely those from the GHV adiabatic path.
See \Cref{prop:merit_sparse} for the formal statement.

\subsection{Linearization of Adiabatic Path (Lemma \ref{lem:pathB2}, \ref{lem:linearTFI}, and \ref{lem:pathC})} \label{sec:overview-linearize}

Now, suppose that we directly apply the one-to-one reduction---from a stoquastic sparse Hamiltonian to a TFD Hamiltonian as described in \Cref{sec:overview-pert}---to each Hamiltonian along the piecewise-linear GHV adiabatic path \(H_{\rm init} \pathto H_{\rm graph} \pathto H_{\rm final}\).
The resulting path can be written as a time-dependent TFD Hamiltonian
\begin{equation}
	K(t) = -a \sum_i X_i + D(t),\footnote{Note that we can guarantee that the coefficients for every $X_i$ are the same, compared to the most general form in \Cref{eq:overview_tfd}, although the proof idea remains valid even if the coefficients \emph{were} different.}
\end{equation}
where $t \in [0,1]$, and $D(t)$ is diagonal.
However, $K(t)$ is no longer piecewise-linear as the reduction above is not a linear mapping.
Therefore, we introduce a \emph{linearization} technique to further reduce $K(t)$ to a linear adiabatic path $H(t)$ (\Cref{lem:pathC}).

While not piecewise-linear, \(K(t)\) is still continuous in \(t\).
Hence it can be well approximated by a piecewise-linear path $K_0 \pathto K_1 \pathto \cdots \pathto K_\ell$ for some sufficiently large $\ell$ (\Cref{lem:pathB2}).
To better illustrate the idea, we set \(\ell = 2\) and discuss how to linearize \(K_0 \pathto K_1 \pathto K_2\) below.

Consider a toy system with Hilbert space \(\mathbb{C}^3\) governed by a time-dependent Hamiltonian \(A(t)\) that is linear in \(t \in [0,1]\).
We construct $A(t)$ such that it gradually shifts the ground state from \(\ket{0}\) to \(\ket{1}\), and finally to \(\ket{2}\):
\begin{equation}
	A(t)=
	\begin{bmatrix}
		t        & -\lambda & 0        \\
		-\lambda & 1/3      & -\lambda \\
		0        & -\lambda & 1-t
	\end{bmatrix},
\end{equation}
where \(\lambda > 0\) is sufficiently small.
The intuition behind this construction is that if \(\lambda\) \emph{were} \(0\), \(A(t)\) would be diagonal and its ground state would change abruptly from \(\ket{0}\) to \(\ket{1}\) and then to \(\ket{2}\) (see \Cref{fig:spec_0}).
Instead, a small but nonzero \(\lambda\) guarantees that the transition is smooth and that the minimum spectral gap of \(A(t)\) is \(\Theta(\lambda)\) (see \Cref{fig:spec_1}).

\begin{figure}[!htbp]
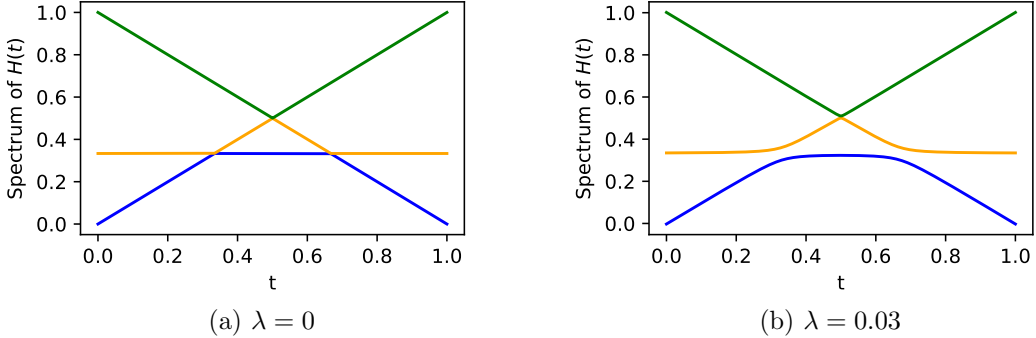

	\centering
	\figspec
	\caption{Spectrum of \(A(t)\) for different \(\lambda\) values. A small but nonzero \(\lambda\) perturbs the spectrum, ensuring an \(\Theta(\lambda)\) spectral gap throughout and making the instantaneous ground state transition of \(A(t)\) smooth.}
	\label{fig:spec}
\end{figure}

Now, construct the linear adiabatic path
\begin{equation} \label{eq:idea-linearization}
	H(t) \coloneqq \Delta \cdot (I \otimes A(t)) + \sum_{\tau =0}^2 K_\tau \otimes \ketbra{\tau}{\tau}
\end{equation}
for a sufficiently large $\Delta$.
By first-order perturbation theory, the degenerate ground space of $I\tensor A(t)$ is perturbed by the minor term $\sum_{\tau} K_\tau \otimes \ketbra{\tau}{\tau}$, and $H(t)$ turns out to be \emph{effectively} a linear combination of $K_0$ and $K_1$ (or $K_1$ and $K_2$) when restricted to the perturbed subspace, because the ground state of $A(t)$ is a linear combination of $\ket{0}$ and $\ket{1}$ (or $\ket{1}$ and $\ket{2}$).
Thus a reduction from $K_0 \pathto K_1 \pathto K_2$ to $H(t)$ is established.

A remaining issue is that the \(H(t)\) obtained in \Cref{eq:idea-linearization} is not a TFD Hamiltonian.
To resolve this, we adapt the unary embedding method from~\cite{leng2024expanding} to transform our toy system \(A(t)\) into a transverse field Ising (TFI) model (\Cref{lem:linearTFI}). With this modification, along with a slight adjustment to \Cref{eq:idea-linearization}, $H(t)$ becomes a time-dependent TFD Hamiltonian that is linear in \(t\).

It is worth noting that the second tensor product component in $H(t)$ can be interpreted as acting on a \emph{clock register}. Clock registers were first introduced in~\cite{Kitaev2002} to prove $\mathsf{QMA}$-hardness of the local Hamiltonian problem. In that work, the authors also used unary representation to reduce the interaction order (i.e., from $\Theta(\log n)$-local to 5-local) of the target Hamiltonian. Note that, however, our construction is entirely distinct from theirs, aside from sharing a similar concept of clock register.

Our linearization technique is general, and thus can be applied to any time-dependent Hamiltonian $K(t)$, not only to $K_0 \pathto K_1 \pathto K_2$ discussed above.
Furthermore, $H(t)$ will be stoquastic if $K(t)$ is stoquastic.
See \Cref{prop:merit_linearize} for the formal statement.

\subsection{Path Tilt and Extension (Lemma \ref{lem:pathD})} \label{sec:overview-tilt}

\Cref{sec:overview-linearize} provides a linear adiabatic path
\begin{equation}
	H(t):= -\sum_i a_i X_i + D(t),
\end{equation}
for \(t\in[-1,1]\) (with \(a_i\) and \(D(t)\) redefined for notational clarity) that encodes the GHV adiabatic path.
This does not match the form required in \Cref{thm:main_informal}, where \(H(t)\) is expected to begin with a TFI Hamiltonian and end at a diagonal Hamiltonian.

Our final step is to introduce a small parameter \(\eta\) and tilt the path by defining
\begin{equation}
	H'(t) := -(1-t\eta) \sum_i a_i X_i + (1+t\eta) D(t).
\end{equation}
It is obvious that for sufficiently small \(\eta\), \(H'(t)\) retains the essential spectral properties of \(H(t)\). We then extend \(H'(t)\) to the larger domain \(t\in[-1/\eta,1/\eta]\) by linearly extending \(D(t)\).
This way, the Hamiltonian begins with \(H'(-1/\eta) = -2 \sum_i a_i X_i\) and ends at \(H'(1/\eta)=2D(t)\), matching the structure presented in \Cref{thm:main_informal}.

This construction, however, introduces technical challenges. Although the spectral behavior of \(H'(t)\) is controlled for \(t\in[-1,1]\), it is nontrivial to guarantee the spectral gap for \(t\in[-1/\eta,-1)\cup(1,1/\eta]\), a requirement for efficient adiabatic evolution. To resolve this, we have to explicitly analyze the components of \(D(t)\) and refine the reductions in \Cref{sec:overview-pert} and \Cref{sec:overview-linearize} so that the resulting paths exhibit additional structural properties.
Moreover, the final form of \(H'(t)\) is adjusted so that \(H'(-1/\eta)\) no longer consists solely of single-qubit Pauli-\(X\) but also $1$- and $2$-qubit Pauli-$Z$.

\subsection{Putting Everything Together (Theorem~\ref{thm:main_informal})} \label{sec:pet}

To summarize, we begin with the \emph{stoquastic sparse} adiabatic path in \cite{gilyen2021sub}; then reduce it to certain \emph{transverse field diagonal} path which is not linear; then we linearize this path to be of the form $-\sum_ia_iX_i+D(t)$ where $D(t)$ is diagonal; finally we slightly tilt this path and extend both its endpoints to arrive at a linear path that starts with a \emph{transverse field Ising} Hamiltonian and ends at some \emph{diagonal} Hamiltonian.

During these reductions, the adiabatic path not only gains more desirable structures, but also preserves the inverse polynomial spectral gap.
This enables efficient simulation via adiabatic quantum computing to locate the ground state of the final diagonal Hamiltonian (i.e., minimizer of the final optimization problem).

On the other hand, our reductions are black-box based on the original GHV adiabatic path. Hence any classical query to our new adiabatic path can be converted efficiently into classical queries to the original path, and the ground state of our final Hamiltonian corresponds directly to the ground state of the final Hamiltonian in the GHV adiabatic path.
As a result, an efficient classical query algorithm for our adiabatic path can also solve the original problem, which is therefore ruled out by the classical hardness of finding the ground state of the ending Hamiltonian established in \cite{gilyen2021sub}.

\subsection{Generalization to Continuous Space (Theorem~\ref{thm:qhd_informal})}\label{sec:overview-qhd}

{ \renewcommand{\ket}[1]{|#1\rangle}
	\renewcommand{\bra}[1]{\langle #1|}
	\renewcommand{\ketbra}[2]{|#1\rangle\langle #2|}
	\renewcommand{\braket}[2]{\langle#1|#2\rangle}
	\renewcommand{\mel}[3]{\langle #1|#2|#3\rangle}

	\Cref{thm:main_informal} establishes an exponential quantum advantage for unconstrained binary optimization problems in an oracular setting. We now extend this construction and demonstrate that an exponential quantum advantage persists in continuous domain $\mathbb{R}^n$~(\Cref{thm:qhd_informal}).

	Our construction is inspired by an intricate connection between two seemingly disparate objects: \textit{TFI Hamiltonians} and \textit{Schr\"odinger operators}, i.e., an unbounded operator of the form $H = -\Delta + V$ with $\Delta$ the Laplace operator and $V$ a continuous function in $\R^n$.
	Despite their structural differences, a recent work~\cite{zheng2024computational} reveals that \sdg\ operators are at least as computationally powerful as TFI Hamiltonians.
	Their key observation is that a TFI Hamiltonian (consisting of 1-qubit Pauli $X$ and up to 2-qubit Pauli $Z$ operators) can be approximately simulated by a Schr\"odinger operator by the following replacement rule:
	\begin{equation} \label{eq:replacement-rule}
		Z \mapsto \hatZ = \sgn(x),\footnote{Actually, \cite{zheng2024computational} uses a smoothed version of $\sgn(x)$, but we stick to $\sgn(x)$ here for simplicity.} \quad -\frac{1}{\Lambda} \cdot X \mapsto \hatX = - \frac{\dee^2 }{\dee x^2}  + \lambda^2 \dwfunc(x),\quad  f_{\rm dw}(x) = \left(x^2-\frac14\right)^2,
	\end{equation}
	for carefully chosen parameters $\lambda = \Theta(\log n)$ and $\Lambda = \Lambda(\lambda) = \poly(n)$.

	In physics, the operator $\hatX$ represents an anharmonic oscillator with a symmetric double-well potential $f_{\rm dw}(x)$. Due to the symmetry, the first two eigenstates (denoted by $\ket{\chi_0}$ and $\ket{\chi_1}$, respectively) of $\hatX$ exhibit symmetric and anti-symmetric ordering, as illustrated in \Cref{fig:dwspec_0}.
	It is easy to see that $\ket{\chi_0}$ and $\ket{\chi_1}$ are analogous to $\ket{+}$ and $\ket{-}$ for the Pauli-$X$ operator, as they respectively form the ground and first excited states of $\hatX$.
	Now, consider the linear combinations of the states $\ket{\chi_0}$ and $\ket{\chi_1}$:
	\begin{equation} \label{eq:hatzeroone-def-ov}
		\ket{\hatzero} \coloneqq \frac{1}{\sqrt{2}} \left( \ket{\chi_0} + \ket{\chi_1} \right), \quad \ket{\hatone} \coloneqq \frac{1}{\sqrt{2}} \left( \ket{\chi_0} - \ket{\chi_1} \right).
	\end{equation}
	A crucial property of $\ket{\hatzero}$ is that its wave function is concentrated in the region $x \ge 0$ with error $\Lambda^{-\Omega(1)}$, as shown in \Cref{fig:dwspec_1}. Therefore, we have $\hatZ \ket{\hatzero} \approx \ket{\hatzero}$ and similarly $\hatZ \ket{\hatone} \approx -\ket{\hatone}$, which justifies the notations $\ket{\hatzero}$ and $\ket{\hatone}$ as they serve as analogues of the standard qubit states $\ket{0}$ and $\ket{1}$ in the context of $\hatZ$ and $\hatX$.

	\begin{figure}[!htbp]
		\centering
		\figdwspec \caption{
			Visualizations of $\ket{\chi_0},\ket{\chi_1}$ regarding $\hatX$ and their linear combination $\ket{\hatzero},\ket{\hatone}$ defined in \Cref{eq:hatzeroone-def-ov}. The parameter for $\hatX$ is chosen to be $\lambda = 30$. The potential field $f_{\rm dw}$ (amplified 10 times) is plotted in a dotted line for reference.
		}
		\label{fig:dwspec}
	\end{figure}

	Building on the established mapping between TFI Hamiltonians and \sdg\ operators, we develop a new technical result (\Cref{lem:tosdg}) that allows us to embed a TFD Hamiltonian $H= \sum_u h_u X_u + D$ into a continuous-space Schr\"odinger operator, where $D$ represents an abstract $n$-qubit diagonal Hamiltonian.
	The idea is to keep the replacement rule $-X \mapsto \Lambda \hatX$ from \Cref{eq:replacement-rule} while introducing a new mapping
	\begin{equation}
		D \mapsto \hatD,\quad \hatD \colon [-1,1]^n \to \mathbb{R}.
	\end{equation}
	In~\Cref{fig:D-heatmap}, we illustrate the construction of $\hatD$ from $D$.
	The function values in the four orthants correspond to the diagonal elements of $D$, and the function equals $0$ on the boundary between orthants.
	From the concentration property of $\ket{\hatzero}$ and $\ket{\hatone}$ it is natural to expect that $\hatD (\ket{\hat{x}_1}\ket{\hat{x}_2} \cdots \ket{\hat{x}_n}) \approx \mel{x_1  \cdots x_n}{D}{x_1 \cdots x_n} (\ket{\hat{x}_1}\ket{\hat{x}_2} \cdots \ket{\hat{x}_n})$, justifying the notation $\hatD$.

	\begin{figure}[!htbp]
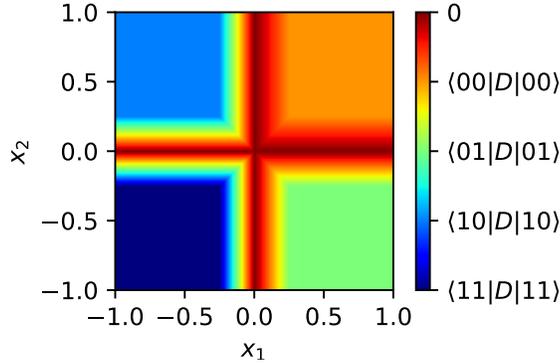

		\centering
		\figheatmap \caption{The Construction of $\hatD$ from a diagonal Hamiltonian $D$ with parameter $w = 1/4$ in \Cref{def:hatD}. It is efficient in the sense that a query to $\hatD(\cdot)$ can be implemented by one query to $D$.}
		\label{fig:D-heatmap}
	\end{figure}

	While the construction of $\hatD$ appears natural, the error analysis of \cite{zheng2024computational} (including, but not limited to, the use of ``$\approx$'' in the last paragraph) does not extend directly. Their approach only provides a trivial error estimate of $\Theta(2^n \Lambda^{-\Omega(1)}) = 2^{(1-o(1))n} \gg 1$, because the nonlocal nature of $D$ forces its representation as a sum of $Z$-tensor products to involve up to $2^n$ terms.
	Thus, it is not immediately clear that the $\hatD$ construction works.
	In contrast, we develop a novel error analysis technique that achieves an error bound of $\poly(n) \Lambda^{-\Omega(1)}$ by leveraging a newfound \emph{symmetry} property of $\ket{\chi_0}$ and $\ket{\chi_1}$, articulated in \Cref{eq:key-sym}.

	Exploiting the correspondence between TFD Hamiltonians and \sdg\ operators, we prove a one-to-one mapping between a linear TFD adiabatic path ($H_{\rm TFI} \pathto D$) from \Cref{thm:main_informal} and a path of \sdg\ operators:
	\begin{equation} \label{eq:H-ov}
		H(t) = - \Delta  + tf(x) +  \nu(t) g(x), \quad \forall t \ge 0,
	\end{equation}
	where $f$ is a real-valued objective function on a box-shaped subset of $\R^n$, $g(x)$ is an auxiliary function independent of $f$, and $\nu(t) = e^{-\Theta(\sqrt{t})}$.

	Similar to the linear adiabatic path, the initial ground state of $H(0)$ can be prepared as a product state, and $H(t)$ always admits a $1/\poly(n)$ spectral gap.
	Moreover, it is obvious from \Cref{eq:H-ov} that the global minimizer of $f$ is encoded in the ground state of $H(t)$ as $t \to \infty$ and can be thus efficiently extracted by computational basis measurement. This leads to an efficient quantum algorithm for minimizing $f$ by simulating the \sdg\ operator evolution in an adiabatic regime.

	On the other hand, an oracle to $f$ can be efficiently implemented with an oracle to $D$. Therefore any efficient classical quantum algorithm finding the global minimum of $f$ can also finds the minimal entry in $D$, violating the classical hardness established in \Cref{thm:main_informal}.

	Our construction of hard continuous instances from hard discrete instances is general: it applies to any discrete objective functions that admit quantum speedups through adiabatic quantum optimization with the mixer being TFD Hamiltonians.
	See \Cref{prop:merit_dtc} for the formal statement.

} 
\section{Preliminaries}\label{sec:prelim}

We use $\mathbb R,\mathbb C$ to denote real number and complex numbers respectively; and use $\mathbb N$ to denote all non-negative integers and use $\mathbb N_+$ for all positive integers.
For a positive integer $n$, we use $[n]$ to denote the set $\{1,2,\ldots,n\}$.
For real numbers $\ell\le r$, we use $[\ell,r]$ to denote the set $\{t\in\mathbb R\colon\ell\le t\le r\}$.
For two binary string $x,y$ of equal length, denote by $\dist(x,y)$ the Hamming distance between \( x \) and \( y \).

\paragraph{Asymptotics.}
Throughout this paper, we use the standard asymptotic notations $O(\cdot)$, $\Theta(\cdot)$, and $\Omega(\cdot)$ to represent \emph{predetermined} expressions that do not depend on any other parameters.\footnote{For example, the statement ``if $f \geq O(g)$ then $\epsilon = \Theta(\eta)$'' means that there exist functions $g_0 = O(g)$ and $\eta_0 = \Theta(\eta)$ such that if $f \geq g_0$, then $\epsilon = \eta_0$.}
Additionally, we use $\poly(\cdot)$ to denote a function that is bounded below by $1$ and bounded above by some implicit polynomial in terms of the parameters within.
We emphasize that \emph{all} implicit functions represented in these asymptotic forms can be computed explicitly and efficiently.

\paragraph{Oracles.}
Our classical hardness and quantum efficiency will be proved in the \emph{oracle} model, a.k.a, the \emph{block-box} model.
The oracles in this paper will be classical functions that allow both classical and quantum access.

Formally, let $f\colon\{0,1\}^n\to\{0,1\}^m$ be a function, where we assume without loss generality that the input and output of $f$ are binary strings.
The oracle $O_f$ provides efficient query access to $f$:
\begin{itemize}
	\item in the classical case, $O_f$ directly evaluates $f$ on the classical query, i.e., $O_f(x)=f(x)$ for any $x\in\{0,1\}^n$;
	\item in the quantum case, $O_f$ handles queries in superposition and evaluates $f$ in superposition, i.e., $O_f\ket{x}_{\mathcal A}\ket{y}_{\mathcal B}=\ket{x}_{\mathcal A}\ket{y\oplus f(x)}_{\mathcal B}$ where $\mathcal A$ is the query register and $\mathcal B$ is the answer register; and $O_f\sum_{x,y}\alpha_{x,y}\ket{x}\ket{y}=\sum_{x,y}\alpha_{x,y}\ket{x}\ket{y\oplus f(x)}$ for any $\{\alpha_{x,y}\}_{x,y}$.
\end{itemize}

We say $O$ is an oracle if $O=O_f$ for some function $f$.
We emphasize that the query oracle $O_f$ ignores the complexity of implementing $f$.

\paragraph{Norms.}
For a matrix (Hamiltonian) $A$, we use $\norm{A}$ to denote its operator norm and use $\norm{A}_\mathrm{F}$ to denote its Frobenius norm.
For a vector (state) $v$, we use $\norm{v}$ to denote its Euclidean (i.e., $\ell_2$) norm.
The distance between two vectors (states) $u,v$ is measured in $\ell_2$ norm; and we say they are $\epsilon$-close if $\norm{u-v}\le\epsilon$.

\begin{fact}\label{fct:mat-ineq}
	The following matrix norm inequalities are standard.
	\begin{enumerate}[label=(\roman*)]
		\item\label{itm:fct:mat-ineq_1} $\norm{A} \geq \norm{A}_\rmmax$, where $\norm{A}_\rmmax \coloneqq \max_{i,j} \{\abs{A_{ij}}\}$ is the \emph{max norm} of $A$.
		\item\label{itm:fct:mat-ineq_2} $\norm{A} \leq \sqrt{\norm{A}_1 \norm{A}_\infty}$, where $\|A\|_1 \coloneqq \max_{j} \sum_i  \abs{A_{ij}}$ is the induced $1$-norm and $\norm{A}_\infty \coloneqq \max_i \sum_j \abs{A_{ij}}$ is the induced $\infty$-norm. In particular, $\norm{A} \leq \norm{A}_1=\norm{A}_\infty$ for Hermitian $A$.
	\end{enumerate}
\end{fact}

\subsection{\aPath}\label{sec:ham-path}

We only work with Hermitian Hamiltonian in this paper. We use ground state to refer to a \emph{normalized} eigenvector of the Hamiltonian corresponding to its minimal eigenvalue.

A Hamiltonian is \emph{stoquastic}~\cite{bravyi2008complexity} if it does not have a sign problem, i.e., with respect to some fixed basis, it has real entries and all off-diagonal entries are non-positive.
The fixed basis will be the computational basis throughout this paper.
While our focus is on stoquastic Hamiltonians, some of our results work for general Hamiltonians.

\paragraph{\apath.}
A \emph{\apath} is a time-dependent Hamiltonian \(H(t)\) defined on the interval \([t_0,t_1]\).
The spectral gap of the path is defined as the infimum of the spectral gaps of \(H(t)\) for all \(t \in [t_0,t_1]\).
Similarly, the norm of the path is given by the supremum of the operator norms of \(H(t)\) for all \(t \in [t_0,t_1]\).
The path is called $L$-Lipschitz if for any $t,t' \in [t_0,t_1]$, we have
\begin{equation}
	\frac{\norm{H(t)-H(t')}}{{\abs{t-t'}}} \leq L.
\end{equation}

\paragraph{Piecewise-linear \apath.}
We use the notation
\begin{equation}
	H_1 \pathto H_2 \pathto \cdots \pathto H_{\ell}
\end{equation}
to describe a time-dependent Hamiltonian \( H(t) \) for \( t \in [0,1] \), defined as
\begin{equation}
	H(t) \coloneqq (1 - \{\tau\}) H_{\lfloor \tau \rfloor} + \{\tau\} H_{\lfloor \tau \rfloor + 1} ,
\end{equation}
where $\tau \coloneqq (\ell-1)t+1$ ranges over $[1,\ell]$, \( \lfloor \tau \rfloor \) is the floor function, and \( \{ \tau\} =  \tau - \lfloor  \tau \rfloor \) represents the fractional part of \( \tau \).

\medskip

The following definition extends the notion of an \emph{interaction graph} from a static Hamiltonian to a time-dependent Hamiltonian.

\begin{definition}[Interaction graph and interaction constraint]
	Given a Hamiltonian $H$, its \emph{interaction graph} is defined as follows. The vertex set $V$ contains all computational basis states and the edge set $E$ contains an edge $(x,y)$ iff $\mel{x}{H}{y} \neq 0$.\footnote{Note that self-loop is allowed in $E$.} For a time-dependent Hamiltonian $H(t)$ and a graph $G = (V, E)$, we say that $H(t)$ has an \emph{interaction constraint} $G$ if the interaction graph for $H(t)$ is a spanning subgraph of $G$ for every $t$.
\end{definition}

\paragraph{Sparsity and oracle access.}
Recall that a Hamiltonian \( H \) is \emph{\( s \)-sparse} if it has at most $s$ nonzero entries per row and per column in the computational basis. In other words, the interaction graph associated with \( H \) has a maximum degree at most \( s \). Therefore, we call a time-dependent Hamiltonian $H(t)$ \emph{explicitly $s$-sparse} if there exists an associated interaction constraint $G$ with maximum degree at most $s$.

We say that \( H(t) \) is accessible via oracle access if there exist oracles that provide efficient access to its structure:
\begin{itemize}
	\item \textsc{Interaction Constraint Oracle.} There exists an interaction constraint $G$ of $H(t)$, such that given a basis state index \( x \), this oracle returns the adjacency list of \( x \) in $G$.
	\item \textsc{Matrix Element Oracle.} Given $t$ and indices \( x, y \), this oracle returns the value of the matrix element \(\mel{x}{H(t)}{y}\).
\end{itemize}

Despite being nonlocal, any \( n \)-qubit explicitly \( \poly(n) \)-sparse Hamiltonian can be efficiently simulated given oracle access~\cite{aharonov2003adiabatic,berry2007efficient}.
Moreover, efficient simulation remains feasible when the Hamiltonian is time-dependent~\cite{berry2020time}.

\paragraph{Pauli operators.}
We denote the standard single-qubit Pauli operators by \(X\), \(Y\), and \(Z\). For multi-qubit systems, we use \(X_u\), \(Y_u\), and \(Z_u\) to represent these operators acting specifically on qubit \(u\).

\paragraph{Hamiltonian with special structures.}
A \emph{transverse field Ising} (TFI) Hamiltonian is a Hamiltonian of the form
\begin{equation}
	H = \sum_{u \in[n]} (h_u X_u + g_u Z_u) + \sum_{1 \leq u<v \leq n} g_{u,v} Z_{u,v},
\end{equation}
where $h_u,g_u,g_{u,v} \in \mathbb{R}$.\footnote{Some definitions of the TFI model exclude the single-qubit $Z$ terms. Nevertheless, including these terms does not alter the computational power of the model~\cite[Eq.~(6)]{bravyi2014complexity}.}
We call $H$ a \emph{transverse field diagonal} (TFD) Hamiltonian if it has the form
\begin{equation}\label{eqn:tfd-defn}
	H = \sum_{u \in [n]} h_u X_u + D,
\end{equation}
where $h_u \in \mathbb{R}$, and $D$ is diagonal.
More generally, we say \( H \) is a \emph{hypercube} Hamiltonian if it has an interaction constraint being the \( n \)-dimensional hypercube (including self-loops). Specifically, this means that $\mel{x}{H}{y} = 0$ holds for all \( x, y \in \binary^n \) with \( \operatorname{dist}(x,y) > 1 \).
Note that all TFI Hamiltonians are TFD ones, and all TFD Hamiltonians are hypercube ones.

\subsection{Hamiltonian Simulation}

We use the term \emph{Hamiltonian simulation}\footnote{Note that this notion of Hamiltonian simulation should not be confused with the more general use of the term in quantum computing, which typically refers to quantum algorithms designed to approximate the unitary dynamics \( e^{-iHt} \) of a given Hamiltonian $H$~\cite{lloyd1996universal,childs2018toward}.}~\cite{bravyi2017complexity,cubitt2018universal} to describe the process of simulating one quantum many-body system using another. Such a simulation involves encoding a Hamiltonian~$H$ within another Hamiltonian~$H'$, enabling the second system to replicate \emph{all} physical behaviors of the first. We formally introduce this notion below.

\begin{definition}[Hamiltonian simulation~\cite{bravyi2017complexity,cubitt2018universal}] \label{def:sim}
	{Let \( H \) be a Hamiltonian acting on a Hilbert space \( \mathcal{H} \) of dimension \( N \).
		A Hamiltonian \( H_{\rmsim} \) on Hilbert space $\mathcal{H}_\rmsim$ and an isometry (encoding) \( \mathcal{E} : \mathcal{H} \to \mathcal{H}_{\rmsim} \) are said to simulate \( H \) with error \( (\epsilon_{\rm enc}, \epsilon) \) if there exists an isometry \( \tilde{\mathcal{E}} : \mathcal{H} \to \mathcal{H}_{\rmsim} \) such that}
	\begin{enumerate}[label=(\roman*)]
		\item The image of \( \tilde{\mathcal{E}} \) coincides with the \emph{low-energy subspace} \( \mathcal{L}_N(H_{\rmsim}) \), which is spanned by the eigenvectors of $H_{\rmsim}$ associated with its $N$ smallest eigenvalues.
		\item \( \| H - \tilde{\mathcal{E}}^\dagger H_{\rmsim} \tilde{\mathcal{E}} \| \leq \epsilon \).
		\item \( \| \mathcal{E} - \tilde{\mathcal{E}} \| \leq \epsilon_{\rm enc} \).
	\end{enumerate}
\end{definition}
\begin{remark}
	If $\epsilon_{\rm enc}=\epsilon$, we say that $(H_\rmsim,\calE)$ simulates $H$ \emph{with error $\epsilon$} for convenience.
\end{remark}

\begin{lemma}[Ground state simulation \cite{bravyi2017complexity}] \label{lem:gsim}
	Suppose $H$ has a non-degenerate ground state $\lvert g \rangle$ separated from excited states by a spectral gap $\delta$. Suppose $(H_{\rmsim}, \mathcal{E})$ simulates $H$ with error $(\epsilon_{\rm enc}, \epsilon)$.
	Then the spectral gap of $H_\rmsim$ is at least $\delta-2\epsilon$.
	In addition, $H_{\rmsim}$ has a non-degenerate ground state $\lvert g_{\rmsim} \rangle$ and
	\begin{equation}
		\|\mathcal{E} \lvert g \rangle - \lvert g_{\rmsim} \rangle \| \leq \epsilon_{\rm enc} + O(\delta^{-1} \epsilon).
	\end{equation}
\end{lemma}

\subsection{Perturbation Theory}

We will use the following standard matrix analysis results.

\begin{fact}[Weyl's inequality, see e.g., \cite{wiki:weyl}]
	\label{fct:weyl}
	Let $A,B$ be Hermitian matrices. Then for each $k$, the $k$th smallest eigenvalue of $A+B$ differs from the $k$th smallest eigenvalue of $A$ by at most $\norm{B}$. As a result, the spectral gap of $A+B$ is at least the spectral gap of $A$ minus $2\norm{B}$.
\end{fact}

\begin{fact}[Davis-Kahan $\sin\theta$ theorem, see e.g., {\cite[Theorem 2]{yu2015useful}}]\label{fct:davis-kahan}
	Let $A,B\in\mathbb C^{n\times n}$ be Hermitian\footnote{\cite[Theorem 2]{yu2015useful} states only for real symmetric matrices, but the analysis naturally generalizes to Hermitian matrices.} matrices.
	Let $d\ge1$ be an integer.
	Define $\delta$ to be the difference between the $(d+1)$th and the $d$th smallest eigenvalue of $A$; and assume $\delta>0$.
	Let $V\in\mathbb C^{n\times d}$ (resp., $\tilde V\in\mathbb C^{n\times d}$) have orthonormal columns corresponding to eigenvectors of the first $d$ smallest eigenvalues of $A$ (resp., $A+B$).
	Then there exists a unitary transform $U\in\mathbb C^{d\times d}$ such that
	\begin{align}
		\norm{VU-\tilde V}_\mathrm{F}
		\le\frac{2^{3/2}\cdot\sqrt d\cdot\norm{B}}{\delta}.
	\end{align}
\end{fact}

\begin{fact}[Property of the Kronecker sum]\label{fct:ksum}
	Let $A_1,A_2,\dots,A_n$ be Hermitian matrices.
	Their Kronecker sum
	\begin{equation}
		A \coloneqq \sum_{i \in [n]} I^{\otimes (i-1)} \otimes A_i \otimes I^{\otimes (n-i)}
	\end{equation}
	satisfies $\sigma(A) = \sigma(A_1) + \cdots + \sigma(A_n)$, where $\sigma(H)$ denotes the spectrum of $H$ (as a set) and $+$ denotes the Minkowski sum.

\end{fact}

Next, we quote perturbative results regarding Hamiltonian simulation.

\paragraph{Taylor series for perturbative reduction.}
Let a finite-dimensional Hilbert space $\mathcal{H}_{\rmsim}$ be decomposed into a direct sum
\begin{equation}
	\mathcal{H}_{\rmsim} = \mathcal{H}_{-} \oplus \mathcal{H}_{+}.
\end{equation}
Let $P_{\pm}$ be the projector onto $\mathcal{H}_{\pm}$.
For any linear operator $O$ on $\mathcal{H}_{\rmsim}$, define
\begin{equation}
	O_{--} = P_- O P_-, \quad O_{-+} = P_- O P_+, \quad O_{+-} = P_+ O P_-, \quad O_{++} = P_+ O P_+.
\end{equation}
The operator $O$ is called block-diagonal if $O_{-+} = 0$ and $O_{+-} = 0$.

Let $H_0$ and $V$ be Hermitian on $\calH_\rmsim$ such that $H_0$ is block-diagonal and $(H_0)_{--}=0$.
Moreover, $(H_0)_{++}$ has all eigenvalues no less than $1$.
Consider a perturbed Hamiltonian
\begin{equation}
	H = \Delta H_0 + V.
\end{equation}

\begin{lemma}[\cite{bravyi2017complexity}] \label{lem:taylor}
	Assume $\Delta > 16 \norm{V}$.
	Then there exists a \emph{Schrieffer-Wolff transformation} $U$~\cite{bravyi2011schrieffer} such that the following holds.
	$U H U^\dagger$ is block-diagonal with regard to $\calH_-$ and $\calH_+$.
	$(UHU^\dagger)_{++}$ has all eigenvalues greater than those of $(UHU^\dagger)_{--}$.
	Furthermore,
	\begin{equation}
		(U H U^\dagger)_{--} = V_{--} - \Delta^{-1} V_{-+} H_0^{-1} V_{+-} + O(\Delta^{-2} \norm{V}^3).
	\end{equation}
\end{lemma}

Suppose now that $H_\rmtarget$ is a fixed  target Hamiltonian acting on
some Hilbert space $\calH_\rmtarget$ and $\calE \colon \calH_\rmtarget\to \calH_\rmsim$
is some fixed isometry (encoding) such that $\im{(\calE)}=\calH_-$.
Define the {\em logical target Hamiltonian} $\bar{H}_\rmtarget$ acting on $\calH_-$ as
\begin{equation}
	\label{target_logical}
	\overline{H}_\rmtarget=\calE H_\rmtarget \calE^\dag.
\end{equation}

\begin{lemma}[First-order reduction~\cite{bravyi2017complexity}]
	\label{lemma:1st}
	Suppose Hermitian operators $H_0,V$ satisfy that $H_0$ is block-diagonal, $(H_0)_{--}=\zeta\cdot P_{-}$ for some $\zeta \in \mathbb{R}$,\footnote{The $\zeta=0$ case is proved in~\cite{bravyi2017complexity} and it is trivial to generalize to arbitrary $\zeta$ by subtracting $\zeta\cdot I$ from $H_0$.} $(H_0)_{++}$ has all eigenvalues at least $\zeta + \Delta$, and
	\begin{equation}
		\label{1st}
		\| \overline{H}_\rmtarget- (H_0)_{--} - (V)_{--} \| \le \epsilon/2.
	\end{equation}
	If $\Delta\ge O(\epsilon^{-1} \|V\|^2 + \epsilon_{\rm enc}^{-1}\|V\|)$, then $H_\rmsim=H_0+V$ and $\mathcal E$ above simulate
	$H_\rmtarget$ with error $(\epsilon_{\rm enc},\epsilon)$.
\end{lemma}

\begin{lemma}[Second-order reduction~\cite{bravyi2017complexity}] \label{lemma:2nd}
	Suppose Hermitian operators $H_0,V_\rmmain,V_\rmextra$ satisfy that $H_0$ is block-diagonal, $(H_0)_{--}=0$, $(H_0)_{++}$ has all eigenvalues at least $1$,
	$V_\rmextra$ is block-diagonal,
	$(V_\rmmain)_{--}=0$,   and
	\begin{equation}
		\label{2nd}
		\|  \overline{H}_\rmtarget- (V_\rmextra)_{--}  +  (V_\rmmain)_{-+} H_0^{-1} (V_\rmmain)_{+-}\| \le \epsilon/2.
	\end{equation}
	Let $\Lambda=\max\{\norm{V_\rmmain},\norm{V_\rmextra}\}$.
	If $\Delta\ge O(\epsilon^{-2} \Lambda^6 + \epsilon_{\rm enc}^{-2} \Lambda^2)$, then $H_\rmsim=\Delta H_0+\Delta^{1/2} V_\rmmain+V_\rmextra$ and $\mathcal E$ above simulate
	$H_\rmtarget$ with error $(\epsilon_{\rm enc},\epsilon)$.
\end{lemma}
\section{Proof Outline for Theorem~\ref{thm:main_informal}} \label{sec:main}

We provide the detailed proof outline of our main result \Cref{thm:main_informal} here.

\begin{restatable}[Formal version of \Cref{thm:main_informal}]{theorem}{thmmain} \label{thm:main}
	For any $n$, there exist an $n$-qubit TFI Hamiltonian $H_{\rm TFI}$ and a family of $n$-qubit diagonal Hamiltonians $D$ such that the linear \apath\ $H(t)$ ($t \in [0,1]$) of the form
	\begin{equation}
		H_{\rm TFI} \pathto D
	\end{equation}
	has exponential quantum advantage given oracle access to $D$ (and an explicit description of $H_{\rm TFI}$):
	\begin{enumerate}[label=(\roman*)]
		\item\label{itm:thm:main_1} Any classical algorithm requires at least $\exp(n^{\Omega(1)})$ queries to find $u \in \binary^n$ that minimizes $\mel{u}{D}{u}$, with success probability greater than $\exp(-n^{\Omega(1)})$.
		\item\label{itm:thm:main_2} $H(t)$ has spectral gap $n^{\Omega(1)}$, norm $\poly(n)$, and the ground state of $H(0) = H_{\rm TFI}$ is $n^{-\Omega(1)}$-close to $\ket{+}^{\otimes (n-\ell)} \ket{0}^{\otimes \ell}$ for some $\ell$.
		Consequently, by \Cref{lem:qatsim}, simulating the \sdg\ dynamics
		\begin{equation}
			i \frac{\dee}{\dee t} \ket{\psi(t)} = T H(t) \ket{\psi(t)}, \quad \ket{\psi(0)} = \ket{+}^{\otimes (n-\ell)} \ket{0}^{\otimes \ell}, \quad T = \poly(n)
		\end{equation}
		yields a final state $\ket{\psi(1)}$ that is $n^{-\Omega(1)}$-close to $\ket{u}$ with query and gate complexity $\poly(n)$.
	\end{enumerate}
\end{restatable}

Our proof of \Cref{thm:main} starts with the following breakthrough construction due to Gily\'en, Hastings, and Vazirani \cite{gilyen2021sub}, which is a family of stoquastic \apaths\ that allows efficient adiabatic quantum computation (see \Crefitem{thm:GHV21}{itm:thm:GHV21_2} and \Cref{lem:qatsim}) but is provably hard for classical query algorithms (\Crefitem{thm:GHV21}{itm:thm:GHV21_1}).

\begin{theorem}[\cite{gilyen2021sub}] \label{thm:GHV21}
	Let $\nGHV$ be an arbitrary positive integer and let $\mGHV=\nGHV^{16/5 - o(1)}$.
	There exists a family of $\nGHV$-qubit stoquastic \apaths\ $H_{\rm GHV}(t)$ that are explicitly $\Theta(\mGHV)$-sparse, of the form
	\begin{equation}
		H_\textrm{\rm init} \pathto H_{\rm graph} \pathto H_\textrm{\rm final},
	\end{equation}
	where $H_\textrm{\rm init} = -\mGHV \ketbra{0^\nGHV}{0^\nGHV}$ and $H_\textrm{\rm final} = -\mGHV \ketbra{u}{u}$ for some $u \in \binary^{\nGHV}$, such that the following holds given oracle access to $H(t)$:
	\begin{enumerate}[label=(\roman*)]
		\item\label{itm:thm:GHV21_1} Any classical algorithm requires at least $\exp(\nGHV^{1/5-o(1)})$ queries to find $u$ with probability greater than $\exp(-\nGHV^{1/5-o(1)})$.
		\item\label{itm:thm:GHV21_2} $H(t)$ has spectral gap $\Omega(\sqrt{\mGHV})$ and norm $O(\mGHV)$.
	\end{enumerate}
\end{theorem}

We reserve $\nGHV,\mGHV$ as the parameter choices from \Cref{thm:GHV21} for frequent reference.

Given the \emph{sparse} \apath\ in \Cref{thm:GHV21}, we first amend it to have \emph{hypercube interaction} while preserving some simplicity of the start and the end Hamiltonian.
The resulting \apath\ of this reduction is formalized in \Cref{lem:pathA} and proved in \Cref{sec:reduction_I}.
We remark that \Cref{lem:pathA} is not a (piecewise) linear path of Hamiltonians but a curve of Hamiltonians with bounded Lipschitz condition.

\begin{restatable}{lemma}{lempathA} \label{lem:pathA}
	Let $\nA=\Theta(\mGHV)$.
	There exists a family of $\nA$-qubit stoquastic hypercube \apaths\ $H(t)$ for $t \in [0,1]$, such that the following holds given oracle access to $H(t)$:
	\begin{enumerate}[label=(\roman*)]
		\item\label{itm:lem:pathA_1} $H(0) = -\mGHV\sum_{i \in [\nA]}X_i$, $H(1) = -\mGHV\ketbra{u}{u}$ for some $u \in \binary^{\nA}$, and $H(t)$ is $\poly(\nGHV)$-Lipschitz.
		\item\label{itm:lem:pathA_2} Any classical algorithm requires at least $\exp(\nGHV^{1/5-o(1)})$ queries to find $u$ with probability greater than $\exp(-\nGHV^{1/5-o(1)})$.
		\item\label{itm:lem:pathA_3} $H(t)$ has spectral gap $\Omega(\sqrt{\mGHV})$ and norm $\poly(\nGHV)$.
	\end{enumerate}
\end{restatable}

With the hypercube structure of \apaths\  from \Cref{lem:pathA}, we use another reduction to further fix it to be \apaths\  consisting of TFD Hamiltonians.
Then the \apath\ evolution only changes the diagonal terms.
This is formalized in \Cref{lem:pathB} below and proved in \Cref{sec:reduction_II}.

\begin{restatable}{lemma}{lempathB}\label{lem:pathB}
	Let $\nB \coloneqq \Theta(\mGHV)$.
	There exist a parameter $\DeltaB \in [\mGHV, \poly(\mGHV)]$ and a family of $\nB$-qubit stoquastic \apaths\ $H(t)$ for $t \in [0,1]$ of the form
	\begin{equation}
		H(t) = \XB + D(t)
		\quad\text{and}\quad
		\XB \coloneqq - \DeltaB \sum_{i \in [\nB]} X_i,
	\end{equation}
	where $D(t)$ is diagonal, such that the following holds given oracle access to $D(t)$:
	\begin{enumerate}[label=(\roman*)]
		\item\label{itm:lem:pathB_1} $D(0) = 0$, $D(1) = - 3 \DeltaB \nB \ketbra{u}{u}$ for some $u \in \binary^{\nB}$, and $D(t)$ is $\poly(\nGHV)$-Lipschitz.
		\item\label{itm:lem:pathB_2} Any classical algorithm requires at least $\exp(\nGHV^{1/5-o(1)})$ queries to find $u$ with probability greater than $\exp(-\nGHV^{1/5-o(1)})$.
		\item\label{itm:lem:pathB_3} $H(t)$ has spectral gap $\Omega(\sqrt{\mGHV})$ and norm $\poly(\nGHV)$.
	\end{enumerate}
\end{restatable}

One evident difference between \Cref{lem:pathB} and \Cref{thm:main} is that the \apath\ in \Cref{lem:pathB} is not a straight line but a complicated curve.
To address this issue, we first approximate the curve by piecewise linear paths as the following \Cref{lem:pathB2}, the validity of which is guaranteed by the Lipschitz condition in \Cref{lem:pathB}.

\begin{restatable}{lemma}{lempathBtwo}\label{lem:pathB2}
	$H(t)$ in \Cref{lem:pathB} can be replaced with \apaths\ of the form
	\begin{equation}\label{eq:lem:pathB2_1}
		\XB + D_0 \pathto \XB + D_1 \pathto \cdots \pathto \XB + D_\ell,
	\end{equation}
	where $\ell = \poly(\nGHV)$.
\end{restatable}
\begin{proof}
	For each $k=0,1,\ldots,\ell$, define $D_k=D(k/\ell)$, where $D(t),t\in[0,1]$ comes from \Cref{lem:pathB}.
	It remains to prove the spectral gap and norm condition for \Cref{eq:lem:pathB2_1}.
	Fix arbitrary $k\in[\ell]$ and $\alpha\in[0,1]$. Define
	\begin{align}
		\tilde H=\XB+\alpha\cdot D_{k-1}+(1-\alpha)\cdot D_k=\alpha\cdot H\left(\frac{k-1}\ell\right)+(1-\alpha)\cdot H\left(\frac k\ell\right),
	\end{align}
	where $H(t),t\in[0,1]$ comes from \Cref{lem:pathB}.
	To bound the operator norm, we simply observe that
	\begin{align}
		\norm{\tilde H}\le\max\{\norm{H((k-1)/\ell)},\norm{H(k/\ell)}\}\le\poly(\nGHV).
	\end{align}
	For the spectral gap, notice that
	\begin{align}
		\norm{\tilde H-H(k/\ell)}
		 & \le\norm{H((k-1)/\ell)-H(k/\ell)}
		=\norm{D((k-1)/\ell)-D(k/\ell)}
		\tag{since $\alpha\in[0,1]$}         \\
		 & \le\poly(\nGHV)/\ell\le1.
		\tag{by \Crefitem{lem:pathB}{itm:lem:pathB_1} and setting $\ell=\poly(\nGHV)$ sufficiently large}
	\end{align}
	Hence by \Cref{fct:weyl}, the spectral gap of $\tilde H$ is at least the spectral gap of $H(k/\ell)$ minus $2$, which is still $\Omega(\sqrt{\mGHV})$ by \Crefitem{lem:pathB}{itm:lem:pathB_3}.
\end{proof}

To further linearize the piecewise linear path from \Cref{lem:pathB2}, we take a slight detour to describe a toy TFI \apath\ that sequentially shifts ground states over the computational basis.
Later we will use this property to embed turning points in \Cref{lem:pathB2}.

\paragraph{TFI Hamiltonian that shifts ground states.}
Define $\ket{\bs{j}} \coloneqq \ket{1^j 0^{\ell-j}}$ and let $\epsilonL, \DeltaL>0$ be parameters to be determined.
We start with an $\ell$-qubit time-dependent TFI Hamiltonian $\HL(t)$ that is \emph{linear} in $t \in \mathbb{R}$.
This \apath\  will gradually shift the ground state from \( \ket{\bs{0}} \) to \( \ket{\bs{1}},\ket{\bs{2}},\ldots,\ket{\bs{\ell}} \) in a sequential order.

Formally, $\HL(t) = \QL + \XL + \ZL(t)$ where $\XL = -\sum_{i \in [\ell]} X_i$ and
\begin{equation}\label{eq:QL_ZL}
	\QL = \DeltaL \left( (\ell-1)I + Z_1 - Z_\ell - \sum_{i \in [\ell-1]} Z_{i}Z_{i+1} \right),
	\quad
	\ZL(t) = \frac{1}{\epsilonL} \sum_{i \in [\ell]} (i-t)(I - Z_i).
\end{equation}

The properties of $\HL(t)$ is listed in the following \Cref{lem:linearTFI} and proved in \Cref{sec:toy_shift}.
We emphasize that \Crefitem{lem:linearTFI}{itm:lem:linearTFI_3} is the key property that allows us to shift ground states.
We also remark that \Crefitem{lem:linearTFI}{itm:lem:linearTFI_4} is an extra benefit that will become useful in a later reduction (\Cref{lem:pathD}).

\begin{restatable}{lemma}{lemlinearTFI} \label{lem:linearTFI}
	Assume $\ell\ge1$ and $\epsilonL\in(0,1/10]$.
	There exists a $\DeltaL = \poly(\ell,1/\epsilonL)$ such that the following holds for $t \in \mathbb{R}$:
	\begin{enumerate}[label=(\roman*)]
		\item\label{itm:lem:linearTFI_1} $\norm{\HL(t)} \leq (\abs{t}+1)\poly(\ell,1/\epsilonL)$.
		\item\label{itm:lem:linearTFI_2} The spectral gap of $\HL(t)$ is at least $1/2$.
		\item\label{itm:lem:linearTFI_3} Let $\ket{\psi(t)}$ be a ground state of $\HL(t)$.
		If $t \in [0, \ell+1]$, then $\ket{\psi(t)}$ is $O(\epsilonL)$-close to a normalized state that is a linear combination of $\ket{ \bs{a}}$ and $\ket{ \bs{b}}$, where $a \coloneqq \max\{\lfloor t-1/2 \rfloor,0\}$ and $b \coloneqq \min\{\lceil t-1/2 \rceil,\ell\}$.
		\item\label{itm:lem:linearTFI_4} For any $\lambda \in [0,2]$, the Hamiltonian
		\begin{equation}
			\HL(\lambda, t) \coloneqq \QL + \lambda \XL + \ZL(t)
		\end{equation}
		satisfies that its spectral gap is at least $1/2$ and its ground state is $O(\epsilonL)$-close to $\ket{\bs{0}}$ (resp., $\ket{\bs{\ell}}$) if $t \leq 0$ (resp., $t \geq \ell + 1$).
	\end{enumerate}
\end{restatable}

Given the toy system $\HL(t)$, we return to our current piecewise linear \apath\ $H(t)$ from \Cref{lem:pathB2} and embed all its turning points into $\HL(t)$.
This addresses another issue that \Cref{lem:pathB2} is not a straight path.

Formally, for $(D_0,D_1,\dots,D_\ell)$ in \Cref{lem:pathB2}, we define the $(\nB+\ell)$-qubit Hamiltonian $\HC(t)$ as follows:
\begin{equation} \label{eq:HC}
	\HC(t) \coloneqq \DeltaC \cdot I_{2^{\nB}} \tensor \HL(t) + \XB \tensor I_{2^\ell} + \DC,\quad \DC \coloneqq \sum_{0 \leq i \leq \ell} D_i \tensor \ketbra{\bs{i}}{\bs{i}},
\end{equation}
where $\DeltaC$ is a parameter to be determined.

Observe that $\HC(t)$ is linear in $t$ since $\HL(t)$ is linear.
In addition, we show that it simulates the \apath\ from \Cref{lem:pathB2}.
The following \Cref{lem:pathC} is proved in \Cref{sec:reduction_III} using first-order perturbative reduction and properties in \Cref{lem:linearTFI}.

\begin{restatable}{lemma}{lempathC} \label{lem:pathC}
	Assume $\epsilonL \leq 1/\poly(\nGHV)$.
	There exist
	\begin{equation}
		\DeltaL = \poly(\nGHV,1/\epsilonL)
		\quad\text{and}\quad
		10\norm{\DC}\le\DeltaC\le\poly(\nGHV)
	\end{equation}
	such that the \apath
	\begin{equation}
		\HC(0) \pathto \HC(\ell+1)
	\end{equation}
	has spectral gap $\Omega(\sqrt{\mGHV})$ and norm $\poly(\nGHV,1/\epsilonL)$.
\end{restatable}

By expanding the definition of $\HC(t)$, we have
\begin{align}
	\HC(t)
	 & =\DeltaC \cdot I_{2^{\nB}} \tensor \HL(t) + \XB \tensor I_{2^\ell} + \DC                    \\
	 & =\DeltaC\cdot I_{2^{\nB}}\tensor\left(\QL+\XL+\ZL(t))\right)+ \XB \tensor I_{2^\ell} + \DC.
\end{align}
Note that $\QL$ and $\ZL(t)$ are classical Ising Hamiltonians that consist of $1$- and $2$-local Pauli-$Z$'s.
On the other hand, $\XL$ and $\XB$ are $1$-local $X$-type Hamiltonian; and $\DC$ is just diagonal (not necessarily local).
Comparing these with our target \Cref{thm:main}, we need to remove $\DC$ from the starting Hamiltonian and remove $\XL,\XB$ from the terminating Hamiltonian.
This is achieved by slightly tilting and then extending the \apath\ $\HC(t)$.

\paragraph{Path tilt and extension.}
Let $\epsilonD>0$ be a parameter to be determined later. We slightly tilt the $\XL,\XB$ and $\DC$ component of $\HC(t)$ as the following \apath\ $\HD$:
\begin{equation}\label{eq:HD}
	\HD(t) \coloneqq \DeltaC\cdot I_{2^{\nB}} \tensor \left( \QL + (1- t \epsilonD)\XL + \ZL(t) \right) + (1- t\epsilonD)\cdot \XB \tensor I_{2^\ell} + (1+ t\epsilonD)\cdot\DC.
\end{equation}

Observe that if $t=-1/\epsilonD$, then $\HD(t)$ is a TFI Hamiltonian (i.e., $\DC$ is removed); and if $t=1/\epsilonD$, then $\HD(t)$ is diagonal (i.e., $\XL,\XB$ are removed).
While this is precisely what \Cref{thm:main} needs, we still need to verify that $\HD(t)$ makes sense for quantum adiabatic computing.
This is formalized in the following \Cref{lem:pathD} and proved in \Cref{sec:reduction_IV}.

\begin{restatable}{lemma}{lempathD} \label{lem:pathD}
	There exist $\epsilonL= 1/\poly(\nGHV)$, $\DeltaL \leq \poly(\nGHV)$,
	\begin{equation}\label{eq:lem:pathD_1}
		10\norm{\DC}\le\DeltaC\le\poly(\nGHV)
		\quad\text{and}\quad
		1/\poly(\nGHV)\le\epsilonD\le1/(10\norm{\DC}+\ell)
	\end{equation}
	such that the \apath
	\begin{equation}\label{eq:HD-path}
		\HD(-1/\epsilonD) \pathto \HD(1/\epsilonD)
	\end{equation}
	has spectral gap $\Omega(\sqrt{\mGHV})$ and norm $\poly(\nGHV)$.
\end{restatable}

To put everything together and prove \Cref{thm:main}, we simply write out the \apath\ in \Cref{lem:pathD}.

\begin{proof}[Proof of \Cref{thm:main}]
	Let $n \coloneqq \nB+\ell=\poly(\nGHV)$. Since $\nGHV$ is arbitrary from \Cref{thm:GHV21}, our construction holds for infinitely many $n$, which can be easily extended to all positive integer $n$ by some dummy padding argument.

	\paragraph{Starting Hamiltonian $H_{\rm TFI}$.}
	Recall \Cref{eq:HD} and \Cref{lem:linearTFI}.
	For the starting Hamiltonian, we have
	\begin{align}
		H_{\rm TFI}\coloneqq\HD(-1/\epsilonD)
		 & =\DeltaC\cdot I_{2^{\nB}}\tensor\HL(2,-1/\epsilonD)+2\XB\tensor I_{2^\ell}.
	\end{align}
	Since $\HL(\lambda,t)$ and $\XB$ are both TFI Hamiltonian, we know that $H_{\rm TFI}$ is indeed a TFI Hamiltonian.

	In addition, by \Crefitem{lem:linearTFI}{itm:lem:linearTFI_4}, we know that the ground state of $\HL(2,-1/\epsilonD)$ (and hence $\DeltaC\HL(2,-1/\epsilonD)$) is $O(\epsilonL)$-close to $\ket{\bs0}=\ket{0^\ell}$.
	Since $\XB=-\DeltaB\sum_{i\in[\nB]}X_i$ from \Cref{lem:pathB}, we know that the ground state of $2\XB$ is $\ket{+}^{\nB}$.
	As a result, the ground state of $H_{\rm TFI}$ is $O(\epsilonL)$-close to $\ket{+}^{\nB}\ket{0^\ell}$.
	As $\epsilonL=1/\poly(\nGHV)=1/\poly(n)$, this, together with the spectral gap condition and norm condition from \Cref{lem:pathD}, verifies \Crefitem{thm:main}{itm:thm:main_2}.

	\paragraph{Ending Hamiltonian $D$.}
	Now for the Hamiltonian in the end, we have
	\begin{align}
		D\coloneqq\HD(1/\epsilonD)
		 & =\DeltaC\cdot I_{2^{\nB}}\tensor\HL(0,1/\epsilonD)+2\DC
		\tag{by \Cref{eq:HD}}                                                                                           \\
		 & =\DeltaC\cdot I_{2^{\nB}}\tensor(\QL+\ZL(1/\epsilonD))+2\sum_{0\le i\le\ell}D_i\tensor\ketbra{\bs i}{\bs i}.
		\tag{by \Cref{eq:HC}}
	\end{align}
	Recall from \Cref{eq:QL_ZL} that $\QL$ and $\ZL(t)$ are diagonal, and each $D_i$ is diagonal.
	Hence $D$ is indeed a diagonal Hamiltonian.
	Moreover, it is obvious from the above that a query to $D$ can be implemented by a query to $D_i$ for some $i=0,1,\ldots,\ell$.

	\paragraph{Classical hardness.}
	Finally we prove \Crefitem{thm:main}{itm:thm:main_1}.
	To this end, it is sufficient to show that any classical algorithm, which makes $q$ queries to $D$ and finds its ground state, can be turned into a classical algorithm $\mathcal{A}$ that makes $q$ queries to $\{D_j\}_{j=0,1,\dots,\ell}$ in \Cref{lem:pathB2} and finds the target $u$ in \Cref{lem:pathB2}, because \Crefitem{lem:pathB2}{itm:lem:pathB_2} then asserts that $q$ must be exponentially large.

	In fact, the construction of $\mathcal{A}$ is straightforward because of (i) the implementation of the oracle to $D$ using oracles to $D_j$, as discussed in the previous paragraph, and (ii) the fact that the ground state of $D$ is exactly $\ket{u,\bs{\ell}}$, since every other computational basis state have larger energy---this is implicit in the proof of \Cref{lem:pathD}, and we include an explicit argument here for completeness:
	\begin{itemize}
		\item For any $u' \neq u$, we have
		      \begin{equation}
			      \mel{u',\bs{\ell}}{D}{u',\bs{\ell}} - \mel{u,\bs{\ell}}{D}{u,\bs{\ell}} = 2\mel{u'}{D_\ell}{u'} - 2\mel{u}{D_\ell}{u} = 6 \DeltaB \nB>0.
		      \end{equation}
		\item For any $v$ and $0 \leq j < \ell$, we have
		      \begin{align}
			       & \phantom{\le}\mel{v,\bs{j}}{D}{v,\bs{j}} - \mel{v,\bs{\ell}}{D}{v,\bs{\ell}}                                                                                             \\
			       & = \bigg(2\mel{v}{D_j}{v} + \DeltaC \mel{\bs{j}}{\ZL(1/\epsilonD)}{\bs{j}}\bigg) -  \bigg(2\mel{v}{D_\ell}{v} + \DeltaC\mel{\bs{\ell}}{\ZL(1/\epsilonD)}{\bs{\ell}}\bigg) \\
			       & = \frac{\DeltaC(\ell-j)(2 /\epsilonD-\ell-j-1)}{\epsilonL} + 2\left( \mel{v}{D_j}{v} - \mel{v}{D_\ell}{v} \right)
			      \tag{by \Cref{eq:QL_ZL}}                                                                                                                                                    \\
			       & >1+ 4 \norm{\DC} - 2 \norm{\DC} >0.
			      \tag{by \Cref{eq:lem:pathD_1}}
		      \end{align}
		\item For any $v$ and $\tau \notin \{\bs{0},\bs{1},\dots,\bs{\ell}\}$, define $w$ as the Hamming weight of $\tau$. We have
		      \begin{align}
			       & \phantom{\le}\mel{v,\tau}{D}{v,\tau} - \mel{v,\bs{w}}{D}{v,\bs{w}}                                                                                                                            \\
			       & =  \DeltaC\left(\mel{\tau}{\ZL(1/\epsilonD)}{\tau}+ \mel{\tau}{\QL}{\tau} \right) - 2\mel{v}{D_w}{v}- \DeltaC\left( \mel{\bs{w}}{\ZL(1/\epsilonD)}{\bs{w}}+\mel{\bs{w}}{\QL}{\bs{w}}  \right) \\
			       & > \DeltaC \left( \mel{\tau}{ \QL}{\tau} - \mel{\bs{w}}{\QL}{\bs{w}} \right) - 2\mel{v}{D_w}{v}
			      \tag{since $\mel{\tau}{\ZL(1/\epsilonD)}{\tau}>\mel{\bs w}{\ZL(1/\epsilonD)}{\bs w}$ from \Cref{eq:QL_ZL}}                                                                                       \\
			       & \geq \DeltaC \cdot 1 - 2\norm{\DC}\ge0,
		      \end{align}
		      where we also use \Cref{eq:QL_ZL} to derive $\mel{\tau}{\QL}{\tau}\ge1+\mel{\bs w}{\QL}{\bs w}$ for the last inequality (see also \Cref{clm:toy_shift_tilde_HL_alpha_x} for a detailed calculation).
	\end{itemize}
	In summary, for each $\ket{v,\tau}$, let $w$ be the Hamming weight of $\tau$ and we have
	\begin{equation}
		\mel{v,\tau}{D}{v,\tau}\ge\mel{v,\bs w}{D}{v,\bs w}\ge\mel{v,\bs\ell}{D}{v,\bs\ell}\ge\mel{u,\bs\ell}{D}{u,\bs\ell}.
	\end{equation}
	In addition, at least one of the inequality becomes strict if $\ket{v,\tau}\ne\ket{u,\bs\ell}$.
	This means that $\ket{u,\bs\ell}$ is the ground state of $D$.
\end{proof}     \section{Reduction I: Proof of Lemma \ref{lem:pathA}}\label{sec:reduction_I}

The goal of this section is to prove \Cref{lem:pathA} (restated below).

\lempathA*

The proof consists of two major steps.
In the first step, we apply a second-order perturbative reduction (\Cref{lemma:2nd}) to transform the \apath\ in \Cref{thm:GHV21} into a \apath\ composed of stoquastic hypercube Hamiltonians.
In the second step, we append it with additional paths at both the beginning and end, to ensure that the final \apath\ starts with $-\mGHV\sum_{i \in [\nA]} X_i$ and terminates at $-\mGHV\ketbra{u}{u}$, as stated in \Crefitem{lem:pathA}{itm:lem:pathA_1}.

\paragraph{Organization.}
The first step for \Cref{lem:pathA} is presented in \Cref{sec:reduction_I_first_step}.
The second step is in \Cref{sec:reduction_I_second_step}.
Then in \Cref{sec:reduction_I_together} we put two steps together and prove \Cref{lem:pathA}.
The missing proofs from this section are deferred to \Cref{sec:missing_proofs_in_sec:reduction_I}.

\paragraph{Notation.}
We (re)define some notation.
Let $K(t)$ for $t \in [0,1]$ be the $\nGHV$-qubit \apath\  $H_{\sf GHV}$ in \Cref{thm:GHV21}.
Write $K(1) = -\mGHV \ketbra{w}{w}$, and reserve the notation $u$ such that it is consistent with the statement of \Cref{lem:pathA}.
Without loss of generality, assume that $G=(V,E)$ is the interaction constraint of $K(t)$.
In particular, $G$ is an $s$-regular graph with self-loops and $s=\Theta(\mGHV)$.
Let $M$ be an upper bound on $\norm{K(t)}$ and $M=O(\mGHV)$.

\subsection{First Step -- Sparse Hamiltonians to Hypercube Hamiltonians}\label{sec:reduction_I_first_step}

Let $\nA \coloneqq 3\nGHV+2s$.
Below we show how to reduce $K(t)$ to an $\nA$-qubit stoquastic hypercube Hamiltonian $H'(t)$ such that $H'(t)$ simulates $K(t)$ in the sense of \Cref{def:sim}.
This reduction is one-to-one and thus for convenience we write $K \coloneqq K(t)$ and $H' \coloneqq H'(t)$.

Note that for hypercube Hamiltonians, non-zero entries can only exist between nodes of distance at most one.
To obtain this for $(x,y)\in E$ which may have arbitrary $\dist(x,y)$, we will use a path $P_{x,y}$ to link $(x,y)\in E$ such that $P_{x,y}$ only consists of hypercube transition edges.
To describe $P_{x,y}$, we define an order for all $\nGHV$-bit binary strings and the rank function among them.

\begin{definition}[Order and rank]\label{def:sparese_to_hypercube_order_and_rank}
	For any $x,y\in\{0,1\}^\nGHV$, we interpret them as $\nGHV$-bit binary numbers and compare them using ``$<$'',``$=$'', and ``$>$'' accordingly.
	For each $x\in\{0,1\}^\nGHV$, we list all $y > x$ with $(x,y) \in E$ in the increasing order and define $\mathrm{rk}(x,y)$ to be the rank of $y$ in this list.
	Similarly we list all $y < x$ with $(x,y) \in E$ in the decreasing order and define $\mathrm{rk}(x,y)$ to be the rank of $y$ in this list.
\end{definition}

Fix an $(x,y)\in E$ pair satisfying $x < y$.
Let \(\ell \coloneqq \dist(x,y) \leq \nGHV\) and consider the sequence
\begin{equation}
	x \eqqcolon z_0,\quad z_1,\quad z_2,\quad \dots,\quad z_{\ell-1},\quad z_{\ell} \coloneqq y
\end{equation}
that connects \(x\) and \(y\), where each \(z_i, z_{i+1}\) differs by exactly one bit (i.e., they have Hamming distance 1).
For uniqueness, we choose \(z_0, z_1, \dots, z_{\ell}\) to be the lexicographically smallest.
While this is already a valid hypercube transition path of $\nGHV$-bit binary strings connecting $x$ and $y$, for the sake of later reduction, we pad it into a valid hypercube transition path of $(\nA=3\nGHV+2s)$-bit binary strings of a \emph{fixed} length of $2\nGHV+3$.

\begin{definition}[Path $P_{x,y}$]\label{def:sparese_to_hypercube_path_P_xy}
	Define $r_1=\mathrm{rk}(x,y),r_2=\mathrm{rk}(y,x)$ and set $e_{r_1}=0^{r_1-1}10^{s-r_1},e_{r_2}=0^{r_2-1}10^{s-r_2}$.
	We connect $(3\nGHV+2s)$-bit strings $(x,x,e_{r_1},0^s,0^\nGHV)$ and $(y,y,0^s,e_{r_2},0^\nGHV)$ with a path $P_{x,y}$ of length $2\nGHV+3$, where indexes of consecutive nodes differ by one bit.
	In $P_{x,y}$, we connect following nodes in order:
	\begin{itemize}
		\item $(x,x,e_{r_1},0^s,0^\nGHV), (x,z_1,e_{r_1},0^s,0^\nGHV)$, \dots, $(x,z_{\ell-1},e_{r_1},0^s,0^\nGHV)$, $(x,y,e_{r_1},0^s, 0^\nGHV)$;
		      \hfill($\ell+1$ nodes)
		\item $(x,y,e_{r_1},0^s, 0^{\nGHV-1}1)$, $(x,y,e_{r_1},0^s, 0^{\nGHV-2}1^2)$, \dots, $(x,y,e_{r_1},0^s, 0^{\ell}1^{\nGHV-\ell})$;
		      \hfill($\nGHV-\ell$ nodes)
		\item $(x,y,e_{r_1},e_{r_2}, 0^{\ell}1^{\nGHV-\ell})$;
		      \hfill(one node)
		\item $(x,y,0^s,e_{r_2},0^{\ell}1^{\nGHV-\ell})$, $(x,y,0^s,e_{r_2},0^{\ell+1}1^{\nGHV-\ell-1})$, \dots, $(x,y,0^s,e_{r_2},0^{\nGHV-1}1)$;
		      \hfill($\nGHV-\ell$ nodes)
		\item $(x,y,0^s,e_{r_2},0^\nGHV)$, $(z_1,y,0^s,e_{r_2},0^\nGHV)$, \dots, $(z_{l-1},y,0^s,e_{r_2},0^\nGHV)$, $(y,y,0^s,e_{r_2},0^\nGHV)$.
		      \hfill($\ell+1$ nodes)
	\end{itemize}
	For convenience we label the nodes in $P_{x,y}$ by $v_{x,y,i}$ for $i \in [2\nGHV+3]$ following the listed order above.
\end{definition}

To establish the reduction using \Cref{lemma:2nd}, we denote by $\calK$ and $\calH$ the Hilbert spaces associated with $K$ and $H'$ respectively.
Define the encoding isometry $\calE \colon \calK \to \calH$ by
\begin{equation}\label{eq:lem:pathA_calE}
	\calE \ket{x} = \ket{x, x, 0^s, 0^s, 0^\nGHV}\quad \text{for } x \in \binary^\nGHV.
\end{equation}
Observe that the image of $\calE$ is $\nA$-bit binary strings, consistent with paths in \Cref{def:sparese_to_hypercube_path_P_xy}.

Let $\calH_- \coloneqq \im \calE$ and we write $\calH = \calH_- \oplus \calH_+$.
In light of \Cref{lemma:2nd}, our plan is to construct
\begin{equation}
	H' \coloneqq \Delta H_0+\Delta^{1/2}V_\rmmain+V_\rmextra,
\end{equation}
for some $\Delta,H_0,V_\rmmain,V_\rmextra$ such that $V_\rmextra-(V_\rmmain)_{-+}H_0^{-1}(V_\rmmain)_{+-}$ is close to $\bar K$.

\paragraph{Entries of $K$ with small offsets.}
Let $\epsilon \coloneqq 1/2$.
For any $x, y \in \binary^\nGHV$ with $(x,y) \in E$, define
\begin{equation}\label{eq:lem:pathA_a-def}
	a_{x,y} \coloneqq \mel{x}{K}{y} - \frac{\epsilon}{2s}.
\end{equation}
Informally, the value $a_{x,y}$ simply records the $(x,y)$th entry of $K$ with a small offset $\epsilon/(2s)$ for non-degeneracy: since $K$ is stoquastic, for any $x\neq y$, we have $K_{x,y}\le0$ and thus $a_{x,y}\le-\epsilon/(2s)$.

\paragraph{Construction of $H_0$ and $V_\rmmain$.}
Initially, all matrix elements of $H_0,V_\rmmain$ are set to zero, and we modify specific entries as needed.

Set $(H_0)_{++}$ as the identity mapping.
Now for each $(x,y)\in E$ with $x<y$, recall the path $P_{x,y}=(v_{x,y,1},\ldots,v_{x,y,2\nGHV+3})$ from \Cref{def:sparese_to_hypercube_path_P_xy}.
We modify the following matrix elements of $H_0$ and $V_\rmmain$, while their corresponding transposed entries are set accordingly due to Hermiticity:
\begin{itemize}
	\item Define
	      \begin{equation} \label{eq:lem:pathA_alpha}
		      \alpha \coloneqq \frac{M+1}{4 \sin^2\left( \frac{\pi}{4\nGHV+8}\right)} \leq O(\nGHV^2M).
	      \end{equation}
	\item $\mel{x,x,0^s,0^s,0^\nGHV}{V_\rmmain}{v_{x,y,1}} = \mel{y,y,0^s,0^s,0^\nGHV}{V_\rmmain}{v_{x,y,2\nGHV+3}}= -\sqrt{\alpha(2\nGHV+4)}$.
	\item $\mel{v_{x,y,i}}{H_0}{v_{x,y,i+1}} = \alpha\cdot a_{x,y}^{-1}$ for $i \in [2\nGHV+2]$.
	\item $\mel{v_{x,y,i}}{H_0}{v_{x,y,i}} = -2 \alpha\cdot a_{x,y}^{-1}$ for $i \in [2\nGHV+3]$.
\end{itemize}
Note that $a_{x,y}^{-1}$ is valid due to the small offsets from \Cref{eq:lem:pathA_a-def}.

By some standard calculation of tridiagonal matrices which we defer to \Cref{sec:missing_proofs_in_sec:reduction_I}, we obtain the following properties on $H_0$ and $V_\rmmain$.

\begin{claim}\label{clm:sparse_to_hypercube_first_step_H_0_V_main}
	$H_0$ is block-diagonal with respect to each $\calH=\calH_-\oplus\calH_+$.
	$(H_0)_{++}$ has minimal eigenvalue at least $1$.
	Moreover,
	\begin{equation}\label{eq:clm:sparse_to_hypercube_first_step_H_0_V_main_2}
		\norm{H_0}\le O\left(\max_{x,y}\frac\alpha{|a_{x,y}|}\right)
	\end{equation}
	and
	\begin{equation}\label{eq:clm:sparse_to_hypercube_first_step_H_0_V_main_1}
		\mel{y,y,0^s,0^s,0^\nGHV}{(V_\rmmain)_{-+}H_0^{-1}(V_\rmmain)_{+-}}{x,x,0^s,0^s,0^\nGHV}
		=\begin{cases}
			-a_{x,y}, & x\ne y,                        \\
			-(2\nGHV+3)\sum_{\substack{y \colon y\ne x \\(x,y)\in E}}a_{x,y}, & x=y.
		\end{cases}
	\end{equation}
\end{claim}

\paragraph{Construction of $V_\rmextra$.}
At this point, we define $V_\rmextra$ as
\begin{equation} \label{eq:lem:pathA_Vextra}
	V_\rmextra \coloneqq \sum_x \left( a_{x,x} -  (2\nGHV+3)\sum_{\substack{y \colon y\neq x\\ (x,y)\in E}}a_{x,y} \right)\ketbra{x,x,0^s,0^s,0^\nGHV}{x,x,0^s,0^s,0^\nGHV}.
\end{equation}
Then by \Cref{eq:lem:pathA_a-def} and \Cref{clm:sparse_to_hypercube_first_step_H_0_V_main}, we have
\begin{equation}
	V_\rmextra-(V_\rmmain)_{-+}H_0^{-1}(V_\rmmain)_{+-}=\overline{K - \frac{\epsilon}{2s}E}  =  \overline{K} -\frac{\epsilon}{2s}\bar{E},
\end{equation}
where $\overline{A} \coloneqq \calE A \calE^\dagger$ for any $A$, and $E$ denotes the unweighted adjacency matrix of $G$ (with a slight abuse of notation).

\paragraph{Parameters in the reduction.}
Note that $\norm{\bar{E}} = \norm{E} \leq \norm{E}_1 \leq s$ by \Crefitem{fct:mat-ineq}{itm:fct:mat-ineq_2}.
Therefore all preconditions of \Cref{lemma:2nd} are met.
Applying \Cref{lemma:2nd}, we obtain some $\Delta \leq \poly(1/\epsilon,\Lambda)$ such that $(H',\calE)$ simulates $K$ with error $\epsilon$, where
\begin{equation}
	\Lambda = \max\{ \norm{V_\rmmain},\norm{V_\rmextra}\} \leq \poly(\nGHV,s,M).
\end{equation}
Hence we get $\Delta \leq \poly(1/\epsilon,\nGHV,s,M)$.
For $H_0$, \Cref{clm:sparse_to_hypercube_first_step_H_0_V_main} gives
\begin{equation}
	\norm{H_0} \leq O \left( \max_{x,y} \left\{ \frac{\alpha}{\abs{a_{x,y}}} \right\} \right) \leq O\left(\frac{\nGHV^2M}{\epsilon/s}\right) \leq \poly(1/\epsilon,\nGHV,s,M).
\end{equation}
Combining these estimates, we get $\norm{H'} \leq \poly(1/\epsilon,\nGHV,s,M)=\poly(\nGHV)$ as $s,M\le O(\mGHV)\le\poly(\nGHV)$ by \Cref{thm:GHV21}.

Note that $(H',\calE)$ simulates $K$ with error $\epsilon=1/2$.
By \Crefitem{thm:GHV21}{itm:thm:GHV21_2} and \Cref{lem:gsim}, we have the following claim.

\begin{claim} \label{clm:lem:pathA_1}
	$H'(t)$ has spectral gap at least $\Omega(\sqrt{\mGHV})$ and norm $\poly(\nGHV)$.
\end{claim}

By mechanically expanding the construction above, it is not hard to see that $H'(t)$ is Lipschitz.
This is stated in the following claim and proved in \Cref{sec:missing_proofs_in_sec:reduction_I}.

\begin{claim}\label{clm:lem:pathA_2}
	$H'(t)$ is $\poly(\nGHV)$-Lipschitz.
\end{claim}

\subsection{Second Step -- Cleaner Start and End}\label{sec:reduction_I_second_step}

Inheriting the notation from \Cref{sec:reduction_I_first_step}, we first append a linear path at the end of $H'(1)$.
Write $H'(1) = \Delta H_0+\Delta^{1/2}V_\rmmain+V_\rmextra$.
Recall that $K(1) = -\mGHV \ketbra{w}{w}$.
Let $\ket{u} \coloneqq \calE \ket{w}$ where $\calE$ is the encoding isometry from \Cref{eq:lem:pathA_calE}.
The path to be appended is
\begin{equation} \label{eq:lem:pathA_tail}
	H'(1)=\Delta H_0+\Delta^{1/2}V_\rmmain+V_\rmextra \pathto \Delta H_0+V_\rmextra \pathto -\mGHV \ketbra{u}{u}. \tag{P1}
\end{equation}

We will show that \Cref{eq:lem:pathA_tail} preserves the spectral gap.

\begin{claim}\label{clm:lem:pathA_tail}
	The \apath\ \Cref{eq:lem:pathA_tail} has spectral gap $\Omega(\mGHV)$ and norm $\poly(\nGHV)$.
	Consequently, it is $\poly(\nGHV)$-Lipschitz by piecewise linearity.
\end{claim}
\begin{proof}
	The norm bound is evident by \Cref{clm:lem:pathA_1} and we focus on the spectral gap.

	We start with the first half of \Cref{eq:lem:pathA_tail}.
	For any $\lambda \in [0,1]$, define $L(\lambda) \coloneqq \Delta H_0+\lambda \Delta^{1/2}V_\rmmain+V_\rmextra$.
	By \Cref{lemma:2nd}, $(L(\lambda),\calE)$ simulates the Hamiltonian
	\begin{equation} \label{eq:lem:pathA_Ltarget}
		\calE^\dagger \left((V_\rmextra)_{--} - \lambda^2 (V_\rmmain)_{-+} H_0^{-1} (V_\rmmain)_{+-}\right) \calE
	\end{equation}
	with error $\epsilon = 1/2$ by the construction in \Cref{sec:reduction_I_first_step}.
	By \Cref{lem:gsim}, this implies that the spectral gap of $L(\lambda)$ is at least the spectral gap of \Cref{eq:lem:pathA_Ltarget} minus $2\epsilon$.
	By \Cref{eq:lem:pathA_Vextra} and \Cref{eq:lem:pathA_a-def}, we have
	\begin{align}
		\calE^\dag(V_\rmextra)_{--}\calE
		 & = -\mGHV\cdot\ketbra{w}{w} + \sum_{x \in \binary^\nGHV} (s-2)(2\nGHV+3)\cdot\frac{\epsilon}{2s}\cdot\ketbra{x}{x}, \\
		 & =-\mGHV\cdot\ketbra{w}{w} +\frac{\epsilon(s-2)(2\nGHV+3)}{2s}\cdot I,
	\end{align}
	where we also use the fact that $K(1)=-\mGHV\ketbra w$.
	By \Cref{clm:sparse_to_hypercube_first_step_H_0_V_main}, we also have
	\begin{align}
		\calE^\dag(V_\rmmain)_{-+} H_0^{-1} (V_\rmmain)_{+-}\calE
		 & = (s-1)(2\nGHV+3)\sum_{x \in \binary^\nGHV}\frac{\epsilon}{2s}\cdot \ketbra{x}{x}
		+\sum_{\substack{x,y \colon x \neq y                                                                   \\ (x,y)\in E}} \frac{\epsilon}{2s}\cdot \ketbra{x}{y}\\
		 & =\left(\frac{\epsilon(s-1)(2\nGHV+3)}{2s}-\frac\epsilon{2s}\right)\cdot I+\frac\epsilon{2s}\cdot E,
	\end{align}
	where $E$ denotes the unweighted adjacency matrix of $G$, with a slight abuse of notation.
	Thus, \Cref{eq:lem:pathA_Ltarget} has the form
	\begin{equation}\label{eq:clm:lem:pathA_tail_1}
		-\mGHV\ketbra{w}{w}- \frac{\lambda^2 \epsilon}{2s}\cdot E+c(s,\nGHV,\epsilon,\lambda)\cdot I,
	\end{equation}
	where $c(s,\nGHV,\epsilon,\lambda)$ is some constant depending only on $\lambda,s,\nGHV,\epsilon$.
	Therefore, its spectral gap is at least $\mGHV - \lambda^2 \epsilon$ by \Cref{fct:weyl} and $\norm{E} \leq \norm{E}_1 = s$.
	Together with \Cref{lem:gsim}, this implies that the first half of the \apath\ \Cref{eq:lem:pathA_tail} is at least $\mGHV - \lambda^2 \epsilon - 2\epsilon\ge \mGHV-3\epsilon=\Omega(\mGHV)$.

	For the second half of \Cref{eq:lem:pathA_tail}, note that it is always block-diagonal with respect to $\calH_{-}$ and $\calH_{+}$.
	The Hamiltonian restricted to $\calH_+$ gradually shifts from $\Delta (H_0)_{++}$ to $0$; and, recalling \Cref{eq:clm:lem:pathA_tail_1}, the restriction to $\calH_-$ shifts from $-\mGHV \ketbra{u}{u} + c(s,\nGHV,\epsilon,0)\cdot I_{--}$ to $-\mGHV \ketbra{u}{u}$.
	Since $\Delta (H_0)_{++}$ has eigenvalues greater than those of $-\mGHV \ketbra{u}{u} + c(s,\nGHV,\epsilon,0)\cdot I_{--}$ by \Cref{lemma:2nd}, it follows that the spectral gap is always exactly $\mGHV$.
	This completes the proof of \Cref{clm:lem:pathA_tail}.
\end{proof}

Write $H'(0)=\Delta \tilde H_0+\Delta^{1/2}\tilde V_\rmmain+\tilde V_\rmextra$.
We append the following \apath\  to the beginning of $H'(t)$:
\begin{equation} \label{eq:lem:pathA_head}
	-\mGHV \ketbra{0^{\nA}}{0^{\nA}} \pathto \Delta \tilde H_0+\tilde V_\rmextra \pathto \tilde H_0+\Delta^{1/2}\tilde V_\rmmain+\tilde V_\rmextra=H'(0). \tag{P2}
\end{equation}

Recall that $K(0)=-\mGHV\ketbra{0^\nGHV}{0^\nGHV}$ and $\calE\ket{0^\nGHV}=\ket{0^{\nA}}$.
By an analogous argument as in \Cref{clm:lem:pathA_tail}, we have the following claim.

\begin{claim} \label{clm:lem:pathA_head}
	The \apath\ \Cref{eq:lem:pathA_head} has spectral gap $\Omega(\mGHV)$ and norm $\poly(\nGHV)$.
	Consequently, it is $\poly(\nGHV)$-Lipschitz by piecewise linearity.
\end{claim}

Finally, consider the following \apath, which we append in front of \Cref{eq:lem:pathA_head}:
\begin{equation} \label{eq:lem:pathA_headhead}
	-\mGHV \sum_{i \in [\nA]} X_i \pathto -\mGHV \sum_{i \in [\nA]} Z_i \pathto -\mGHV \ketbra{0^{\nA}}{0^{\nA}}. \tag{P3}
\end{equation}

\begin{claim} \label{clm:lem:pathA_4}
	The \apath\ \Cref{eq:lem:pathA_headhead} has spectral gap at least $\mGHV$ and norm $O(\mGHV)$.
	Consequently, it is $O(\mGHV)$-Lipschitz by piecewise linearity.
\end{claim}
\begin{proof}
	The second half of \Cref{eq:lem:pathA_headhead} is always diagonal with spectral gap decreasing linearly from $2\mGHV$ to $\mGHV$.
	For the first half, we apply \Cref{fct:ksum} and its spectral gap is identical to that of the single-qubit \apath\ $-\mGHV X \pathto -\mGHV Z$.
	By direct calculation, the spectral gap of $\lambda X + (1-\lambda)Z$ is given by
	$
		2\sqrt{2\lambda^2 - 2\lambda + 1} \geq \sqrt{2}.
	$
	Hence the first half of \Cref{eq:lem:pathA_headhead} has spectral gap at least $\sqrt{2} \mGHV$.
\end{proof}

\subsection{Putting Everything Together}\label{sec:reduction_I_together}

Finally we prove \Cref{lem:pathA} by combining the hypercube \apath\ in \Cref{sec:reduction_I_first_step} and the additional decoration in \Cref{sec:reduction_I_second_step}.

\begin{proof}[Proof of \Cref{lem:pathA}]
	Recall \apaths\  $H'([0,1])$ from \Cref{sec:reduction_I_first_step} and \Cref{eq:lem:pathA_headhead}, \Cref{eq:lem:pathA_head}, \Cref{eq:lem:pathA_tail} from \Cref{sec:reduction_I_second_step}.
	We sequentially connect \Cref{eq:lem:pathA_headhead}, \Cref{eq:lem:pathA_head}, $H'([0,1])$, and \Cref{eq:lem:pathA_tail}.
	This forms a \apath\ $H:[0,4] \to \mathcal{H}$. By rescaling the domain linearly from $[0,4]$ to $[0,1]$, we obtain the final path $H(t)$.

	By \Cref{clm:lem:pathA_1}, \Cref{clm:lem:pathA_2}, \Cref{lem:pathA}, \Cref{clm:lem:pathA_head}, and \Cref{clm:lem:pathA_4}, $H(t)$ has spectral gap $\Omega(\sqrt{\mGHV})$, norm $\poly(\nGHV)$, and it is $\poly(\nGHV)$-Lipschitz.
	This verifies \Crefitem{lem:pathA}{itm:lem:pathA_1} and \Crefitem{lem:pathA}{itm:lem:pathA_3}.

	For \Crefitem{lem:pathA}{itm:lem:pathA_2}, we note that any query to $\mel{x}{H(t)}{y}$ reduces to queries of the $(x,y)$ entries of $H_0(t')$, $V_\rmmain(t')$, and $V_\rmextra(t')$ for some $t' \coloneqq t'(t)$. By our construction, these can be implemented using $O(1)$, $O(1)$, and $O(s)$ queries, respectively, to the entries of $K(t')$.
	Consequently, any $q$-query classical algorithm that finds $\ket{u}$ can also find $\ket{w}$, since $\ket{w} = \calE^\dagger \ket{u}$. This transformation results in an $O(qs)$-query classical algorithm with oracle access to $K(t')$.
	By \Crefitem{thm:GHV21}{itm:thm:GHV21_1}, we have $O(qs) \geq \exp(\nGHV^{1/5-o(1)})$.
	Since $s=\poly(\nGHV)$, this implies that $q \geq\exp(\nGHV^{1/5-o(1)})$.
\end{proof}

\subsection{Deferred Proofs}\label{sec:missing_proofs_in_sec:reduction_I}

\begin{proof}[Proof of \Cref{clm:sparse_to_hypercube_first_step_H_0_V_main}]
	Observe that every index $\ket{v,v',a,b,c}$ in $\calH_+$ is only used at most once in some path, because $v=x$ when $a \neq 0$, and $v'=y$ when $b \neq 0$.
	This implies that $(H_0)_{++}$ is block-diagonal with respect to each $P_{x,y}$, and thus also block-diagonal on $\calH$.

	In fact, observe that $(H_0)_{++}$ restricted on (nodes of) $P_{x,y}$ is a $(2\nGHV+3)\times(2\nGHV+3)$ tridiagonal matrix of the form
	\begin{equation}\label{eq:lem:pathA_H0-bd}
		-\frac{\alpha}{a_{x,y}} \cdot
		\begin{bmatrix}
			2  & -1     &        &    \\
			-1 & 2      & \ddots &    \\
			   & \ddots & \ddots & -1 \\
			   &        & -1     & 2
		\end{bmatrix}.
	\end{equation}
	By standard calculations (see e.g., \cite[Page 20]{elliott1953characteristic}), \Cref{eq:lem:pathA_H0-bd} has eigenvalues
	\begin{equation}
		-\frac{\alpha}{a_{x,y}}\left(2-2\cos\left(\frac{k\pi}{2\nGHV+4}\right)\right) = -\frac{\alpha}{a_{x,y}}\cdot4\sin^2\left(\frac{k\pi}{4\nGHV+8}\right)
		\quad\text{for $k \in [2\nGHV+3]$.}
	\end{equation}
	This readily proves \Cref{eq:clm:sparse_to_hypercube_first_step_H_0_V_main_2}.
	In addition, the minimal eigenvalue of \Cref{eq:lem:pathA_H0-bd} is
	\begin{equation}
		\begin{aligned}
			-\alpha\cdot \frac{1}{a_{x,y}}\cdot4\sin^2\left(\frac{\pi}{4\nGHV+8}\right) & = \frac{M+1}{4 \sin^2\left( \frac{\pi}{4\nGHV+8}\right)}\cdot \frac{1}{-a_{x,y}}\cdot4\sin^2\left(\frac{\pi}{4\nGHV+8}\right) & \text{(by \Cref{eq:lem:pathA_alpha})} \\
			                                                                            & \geq \frac{M+1}{4 \sin^2\left( \frac{\pi}{4\nGHV+8}\right)}\cdot \frac{1}{M + 1}\cdot4\sin^2\left(\frac{\pi}{4\nGHV+8}\right)                                         \\
			                                                                            & = 1,
		\end{aligned}
	\end{equation}
	which implies that $(H_0)_{++}$ has minimal eigenvalue at least $1$.
	Standard calculations (see e.g., \cite[Lemma 3]{usmani1994inversion}) also show that the inverse of \Cref{eq:lem:pathA_H0-bd} is
	\begin{equation}\label{eq:lemma:to-hypercube_2}
		-\frac{a_{x,y}}{\alpha}\cdot \frac{1}{2\nGHV + 4} \sum_{i,j \in [2\nGHV+3]} \ketbra{i}{j} \cdot  \min\{i,j\} (2\nGHV + 4 - \max\{i,j\}).
	\end{equation}
	Now we compute $(V_\rmmain)_{-+}H_0^{-1}(V_\rmmain)_{+-}$ as follows:
	For $x<y$, we have
	\begin{equation} \label{eq:lem:pathA_Vmain-offdiag}
		\begin{aligned}
			 & \phantom{=}\ \mel{y,y,0^s,0^s,0^\nGHV}{(V_\rmmain)_{-+}H_0^{-1}(V_\rmmain)_{+-}}{x,x,0^s,0^s,0^\nGHV} \\
			 & = \alpha (2\nGHV+4)
			\mel{v_{x,y,2\nGHV+3}}{H_0^{-1}}{v_{x,y,1}}
			 & \text{(by definition of $H_0$ and $V_\rmmain$)}                                                       \\
			 & = \alpha (2\nGHV+4)  \left( -\frac{a_{x,y}}{\alpha}\cdot \frac{1}{2\nGHV + 4} \right)
			 & \text{(by \Cref{eq:lemma:to-hypercube_2})}                                                            \\
			 & =-a_{x,y}.
		\end{aligned}
	\end{equation}
	On the other hand, for diagonal terms,
	\begin{equation} \label{eq:lem:pathA_Vmain-diag}
		\begin{aligned}
			 & \phantom{=}\ \mel{x,x,0^s,0^s,0^\nGHV}{(V_\rmmain)_{-+}H_0^{-1}(V_\rmmain)_{+-}}{x,x,0^s,0^s,0^\nGHV} \\
			 & = \sum_{\substack{y \colon y>x                                                                        \\ (x,y)\in E}} \alpha (2\nGHV+4)
			\mel{v_{x,y,1}}{H_0^{-1}}{v_{x,y,1}}
			\\
			 & \phantom{=}\ + \sum_{\substack{y \colon y<x                                                           \\ (x,y)\in E}} \alpha (2\nGHV+4)
			\mel{v_{y,x,2\nGHV+3}}{H_0^{-1}}{v_{y,x,2\nGHV+3}}
			 & \text{(by definition of $H_0$ and $V_\rmmain$)}                                                       \\
			 & = \sum_{\substack{y \colon y \neq x                                                                   \\ (x,y)\in E}} \alpha (2\nGHV+4)  \left( -\frac{a_{x,y}}{\alpha}\cdot \frac{2\nGHV+3}{2\nGHV + 4} \right)
			 & \text{(by \Cref{eq:lemma:to-hypercube_2})}                                                            \\
			 & =-\sum_{\substack{y \colon y\neq x                                                                    \\ (x,y)\in E}} (2\nGHV+3)\cdot a_{x,y}.
		\end{aligned}
	\end{equation}
	This completes the proof of \Cref{clm:sparse_to_hypercube_first_step_H_0_V_main}.
\end{proof}

\begin{proof}[Proof of \Cref{clm:lem:pathA_2}]
	To analyze the Lipschitz constant of $H'(t)$, consider our reduction $K(t) \mapsto H'(t)$.
	We know that $K(t)$ is $O(\mGHV)$-Lipschitz since it is piecewise linear and has norm bounded by $O(\mGHV)$ as in \Cref{thm:GHV21}.
	Therefore, $\norm{K(t)}_{\rm max}$ is also $O(\mGHV)$-Lipschitz due to \Crefitem{fct:mat-ineq}{itm:fct:mat-ineq_1}.

	Recall the definition in \Cref{eq:lem:pathA_a-def} that
	\begin{equation}
		a_{x,y}(t) = \mel{x}{K(t)}{y} - \epsilon/(2s).
	\end{equation}
	This $a_{x,y}(t)$ is $O(\mGHV)$-Lipschitz and $1/a_{x,y}(t)$ is $O(\mGHV s^2/\epsilon^2)$-Lipschitz thanks to the small offset.

	Now, we analyze the Lipschitz continuity of the individual terms in $H'(t)$.
	\begin{itemize}
		\item Looking into the construction of $H_0(t)$, we observe that $H_0(t)$ is block-diagonal with each block taking value in $\left\{0, 1, -\frac{\alpha}{a_{x,y}(t)} T\right\}$, where $T$ is a constant matrix \Cref{eq:lem:pathA_H0-bd} in the proof of \Cref{clm:sparse_to_hypercube_first_step_H_0_V_main}. Given that $\alpha \leq O(\nGHV^2 M)$ by \Cref{eq:lem:pathA_alpha}, we conclude that $H_0(t)$ is $O(\nGHV^2 M \mGHV s^2/\epsilon^2)$-Lipschitz.
		\item $V_\rmmain(t)$ has exactly $s-1$ nonzero entries per row, each equal to $-\sqrt{\alpha(2\nGHV+4)}$ by construction. Thus, it is $O(\nGHV^{2.5}M)$-Lipschitz.
		\item $V_\rmextra(t)$ is diagonal, and its Lipschitz constant is $O(s M)$ by its definition in \Cref{eq:lem:pathA_Vextra}.
	\end{itemize}
	Finally, since $\Delta \leq \poly(1/\epsilon, \nGHV, s, M)$, we conclude that $H'(t) = \Delta H_0(t)+\Delta^{1/2}V_\rmmain(t)+V_\rmextra(t)$ is $\poly(\nGHV,s,\mGHV,M,1/\epsilon)$-Lipschitz as claimed.
\end{proof}     \section{Reduction II: Proof of Lemma \ref{lem:pathB}}\label{sec:reduction_II}

In this section we prove \Cref{lem:pathB} (restated below).

\lempathB*

The proof consists of two major steps.
In the first step, we apply a second-order perturbative reduction (\Cref{lemma:2nd}) to transform the \apath\ in \Cref{lem:pathA} into a \apath\ composed of Hamiltonians of the form $H'(t) = \XB + D'(t)$, where $D'(t)$ is diagonal and $\XB=-\DeltaB\sum_{i\in[\nB]}X_i$ is some fixed $X$-type interaction.
In the second step, we further decorate the obtained \apath\ by appending additional paths at both its beginning and end, ensuring that the final \apath\ starts with $\XB$ and terminates at $\XB - 3 \DeltaB \nB \ketbra{u}{u}$, as stated in \Crefitem{lem:pathB}{itm:lem:pathB_1}.

\paragraph{Organization.}
The first step is formalized in \Cref{sec:reduction_II_hypercube_to_X} and the second step is in \Cref{sec:reduction_II_cleaner}.
In \Cref{sec:reduction_II_together} we put them together and prove \Cref{lem:pathB}.
The missing proofs from this section can be found in \Cref{sec:missing_proofs_in_sec:reduction_II}.

\paragraph{Notation.}
We (re)define some notation.
Let $n$ be $\nA$ and let $K(t)$ for $t \in [0,1]$ be the $n$-qubit \apath\ in \Cref{lem:pathA}.
Write $K(1) = -\mGHV \ketbra{w}{w}$, and reserve the notation $u$ such that it is consistent with the statement of \Cref{lem:pathB}.
Let $M$ denote an upper bound on $\norm{K(t)}$ and $\mGHV\le M\le\poly(n)=\poly(\nGHV)$.

\subsection{First Step -- Hypercube Hamiltonians to TFD Hamiltonians}\label{sec:reduction_II_hypercube_to_X}

Let $\nB \coloneqq 2n$.
Below we show how to reduce $K(t)$ to an $\nB$-qubit Hamiltonian $H'(t) = \XB + D'(t)$ for any $t \in [0,1]$ such that $H'(t)$ simulates $K(t)$ in the sense of \Cref{def:sim}.
This reduction is one-to-one and thus for convenience we write $K \coloneqq K(t)$ and $H' \coloneqq H'(t)$.

To establish the proof using \Cref{lemma:2nd}, we denote $\calK$ and $\calH$ by the Hilbert spaces associated with $K$ and $H'$ respectively.
Define the encoding isometry $\calE \colon \calK \to \calH$ by
\begin{equation}\label{eq:lem:pathB_calE}
	\calE \ket{x} = \ket{x, x}\quad \text{for } x \in \binary^n.
\end{equation}
Let $\calH_- \coloneqq \im \calE$ and we write $\calH = \calH_{-} \oplus \calH_+$.
In light of \Cref{lemma:2nd}, our plan is to construct
\begin{equation}
	H' \coloneqq \Delta H_0+\Delta^{1/2}V_\rmmain+V_\rmextra,
\end{equation}
for some $\Delta,H_0,V_\rmmain,V_\rmextra$ such that $V_\rmextra-(V_\rmmain)_{-+}H_0^{-1}(V_\rmmain)_{+-}$ is close to $\bar K$.

\paragraph{Entries of $K$ with small offsets.}
Let $\epsilon \coloneqq 1/2$.
For any $x,y \in \binary^n$ with $\dist(x,y) \leq 1$, define
\begin{equation} \label{eq:lem:pathB_a-def}
	a_{x,y} \coloneqq \mel{x}{K}{y} - \frac{\epsilon}{2(n+1)}.
\end{equation}
Since $K$ is stoquastic, the small offset ensures that $a_{x,y}<\frac\epsilon{2(n+1)}$ whenever $\dist(x,y)=1$.

\paragraph{Construction of $H_0$.}
We define the diagonal matrix $H_0$ as follows:
\begin{equation} \label{eq:lem:pathB_H0}
	\mel{x,y}{H_0}{x,y} = \begin{cases}
		0                                                                       & x=y,             \\
		-(M+1)\cdot a_{x,y}^{-1}                                                & \dist(x,y)=1,    \\
		(M+1)\left(\mGHV + \frac{\epsilon}{2(n+1)}\right)^{-1} \cdot \dist(x,y) & {\rm otherwise}.
	\end{cases}
\end{equation}
We note the following easy observation.

\begin{claim}\label{clm:reduction_II_first_step_H_0}
	$(H_0)_{--}=0$ and eigenvalues of $(H_0)_{++}$ are at least $1$.
\end{claim}
\begin{proof}
	Note that $\calH_-$ is spanned by $\ket{x,x},x\in\{0,1\}^n$. Hence $(H_0)_{--}=0$.
	For $x \neq y$, we have
	\begin{equation}
		a_{x,y} \geq (M + 1) \left( M + \frac{\epsilon}{2(n+1)} \right)^{-1} \geq (M + 1)(M + 1)^{-1} = 1,
	\end{equation}
	which means that eigenvalues of $(H_0)_{++}$ are at least $1$.
\end{proof}

\paragraph{Construction of $V_\rmmain$ and $V_\rmextra$.}
We construct $V_\rmmain$ and $V_\rmextra$ as follows:
\begin{equation} \label{eq:lem:pathB_Vextra}
	V_\rmmain \coloneqq - \sqrt{\frac{M+1}{2}}\sum_{i \in [2n]} X_i, \quad
	V_\rmextra \coloneqq \sum_x\left(a_{x,x}
	-\sum_{y\colon\dist(x,y)=1} a_{x,y} \right)\ketbra{x,x}.
\end{equation}

By a direct calculation, which we defer to \Cref{sec:missing_proofs_in_sec:reduction_II}, we have the following desired relation.

\begin{claim}\label{clm:reduction_II_first_step_Vmain_Vextra}
	$(V_\rmmain)_{--}=0$ and
	\begin{equation} \label{eq:lem:pathB_Heff}
		(V_\rmextra)_{--} - (V_\rmmain)_{-+} H_0^{-1} (V_\rmmain)_{+-} = \overline{K} - \frac{\epsilon}{2(n+1)}\cdot\overline{E},
	\end{equation}
	where $\overline{A} \coloneqq \calE A \calE^\dagger$ for any $A$ and $E \coloneqq \sum_{x,y \colon \dist(x,y) \leq 1} \ketbra{x}{y}$.
\end{claim}

\paragraph{Parameters in the reduction.}
Note that $\norm { \overline{E} } = \norm{E} \leq \norm{E}_1 = n+1$ by \Crefitem{fct:mat-ineq}{itm:fct:mat-ineq_2}.
Therefore all preconditions of \Cref{lemma:2nd} are met.
Applying \Cref{lemma:2nd}, we have that for any $\Delta \geq \poly(1/\epsilon,\Lambda)$, $(H',\calE)$ simulates $K$ with error $\epsilon$, where
\begin{equation}
	\Lambda = \max\{ \norm{V_\rmmain},\norm{V_\rmextra}\} \leq O(nM)=\poly(\nGHV).
\end{equation}
Hence it suffices to ensure $\Delta \geq \poly(1/\epsilon,\nGHV)$.
For $H_0$, it is diagonal and
\begin{equation}
	\norm{H_0} \leq (M+1)\cdot\frac{2(n+1)}{\epsilon}\cdot n \leq \poly(1/\epsilon,n,M)=\poly(1/\epsilon,\nGHV)
\end{equation}
by \Cref{eq:lem:pathB_H0} and \Cref{eq:lem:pathB_a-def}.
Combining these estimates we get $\norm{H'} \leq \poly(1/\epsilon,\nGHV,\Delta)$.

Finally, note that \( H_0 \) and \( V_\rmextra \) are diagonal, while $V_\rmmain=- \sqrt{\frac{M+1}{2}}\sum_{i \in [2n]} X_i$.
Thus, the Hamiltonian \( H' \) can be written in the desired form:
\begin{equation} \label{eq:lem:pathB_desired}
	H' = \XB + D', \quad \XB \coloneqq -\DeltaB \sum_{i \in [2n]} X_i = \Delta^{1/2}V_\rmmain, \quad D' \coloneqq \Delta H_0 + V_\rmextra,
\end{equation}
where we set
\begin{equation} \label{eq:lem:pathB_DeltaBdet}
	\DeltaB = \sqrt{\frac{\Delta(M+1)}{2}}.
\end{equation}
Since \( \epsilon = 1/2 \), the parameter \( \Delta \) determined implicitly in \Cref{eq:lem:pathB_DeltaBdet} can be chosen to satisfy the requirement \( \Delta \geq \poly(1/\epsilon,\nGHV) \), provided that \( \DeltaB \geq \poly(\nGHV) \).

With the spectral property of $K(t)$ from \Crefitem{lem:pathA}{itm:lem:pathA_3},
the fact that $(H',\calE)$ simulates $K$ with error $\epsilon$ gives us the following claim.

\begin{claim} \label{clm:lem:pathB_1}
	$H'(t)$ has spectral gap at least $\Omega(\sqrt{\mGHV})$ and norm $\poly(\nGHV, \DeltaB)$.
\end{claim}

We can also analyze the Lipschitz condition for $H'(t)$.

\begin{claim} \label{clm:lem:pathB_2}
	$H'(t)$ is $\poly(\nGHV,\DeltaB)$-Lipschitz.
\end{claim}
\begin{proof}
	Consider our reduction $K(t) \mapsto H'(t)$ where $K(t)$ is $\poly(\nGHV)$-Lipschitz by \Crefitem{lem:pathA}{itm:lem:pathA_1}.
	Therefore, $\norm{K(t)}_{\rm max}$ is also $\poly(\nGHV)$-Lipschitz by \Crefitem{fct:mat-ineq}{itm:fct:mat-ineq_1}.
	Hence, recalling \Cref{eq:lem:pathB_a-def}, we know that $a_{x,y}(t)$ is $\poly(\nGHV)$-Lipschitz and $1/a_{x,y}(t)$ is $O(\poly(\nGHV)/\epsilon^2)$-Lipschitz for any $x,y \in \binary^n$.
	Now we analyze the Lipschitz continuity of the individual terms in $H'(t)$.
	\begin{itemize}
		\item $H_0(t)$ is diagonal. Therefore it is $O(M\poly(\nGHV)/\epsilon^2)=\poly(\nGHV)$-Lipschitz by its definition in \Cref{eq:lem:pathB_H0}.
		\item $V_\rmmain$ remains constant with respect to $t$.
		\item $V_\rmextra(t)$ is diagonal, and its Lipschitz constant is $O(n M)=\poly(\nGHV)$ by its definition in \Cref{eq:lem:pathB_Vextra}.
	\end{itemize}
	Finally, since $\Delta \leq O(\DeltaB^2)$ by \Cref{eq:lem:pathB_DeltaBdet}, we conclude that $H'(t) = \Delta H_0(t)+\Delta^{1/2}V_\rmmain+V_\rmextra(t)$ is $\poly(\nGHV,\Delta)=\poly(\nGHV,\DeltaB)$-Lipschitz.
\end{proof}

\subsection{Second Step -- Cleaner Start and End}\label{sec:reduction_II_cleaner}

Inheriting the notation from \Cref{sec:reduction_II_hypercube_to_X}, we first append a linear path at the end of $H'(1)$.
Based on the construction of $H'(t)$ from \Cref{sec:reduction_II_hypercube_to_X}, we write
\begin{equation}
	H'(1) = \Delta H_0+ V \quad\text{and}\quad V \coloneqq \Delta^{1/2}V_\rmmain+V_\rmextra.
\end{equation}
Recall that $K(1) = -\mGHV \ketbra{w}{w}$.
Let $\ket{u} \coloneqq \calE \ket{w}$ where $\calE$ is the encoding isometry from \Cref{eq:lem:pathB_calE}.
The path to be appended is
\begin{equation} \label{eq:lem:pathB_tail}
	H'(1) \pathto H'(1) - 3 \DeltaB \nB \ketbra{u}{u} \pathto \XB - 3 \DeltaB \nB \ketbra{u}{u}. \tag{P4}
\end{equation}

We will show that \Cref{eq:lem:pathB_tail} preserves the spectral gap.
The first half uses Taylor expansion of perturbative reduction (\Cref{lem:taylor}) and the second half is a direct calculation.
The proof of \Cref{clm:lem:pathB_3} is deferred to \Cref{sec:missing_proofs_in_sec:reduction_II}.

\begin{claim} \label{clm:lem:pathB_3}
	Assume $\DeltaB=\poly(\nGHV)$.
	The \apath\ \Cref{eq:lem:pathB_tail} has spectral gap $\Omega(n\DeltaB)$ and norm $\poly(\nGHV)$.
	Consequently, it is $\poly(\nGHV)$-Lipschitz by piecewise linearity.
\end{claim}

To decorate the start, we explicitly calculate $H'(0)$.
Since $K(0) = -\mGHV \sum_{i \in [n]} X_i$, we have
\begin{align}
	H'(0) & = \XB + \Delta H_0(0) + V_\rmextra(0)
	\tag{by \Cref{eq:lem:pathB_desired}}                                                                         \\
	      & = \XB + \Delta \sum_{x,y} \alpha \cdot \dist(x,y) \ketbra{x,y}{x,y} + \sum_x \beta \ketbra{x,x}{x,x}
	\tag{by \Cref{eq:lem:pathB_H0} and \Cref{eq:lem:pathB_Vextra}}                                               \\
	      & = \XB + \frac12 \Delta \alpha \left( nI - \sum_{i \in [n]} Z_i \tensor Z_{n+i} \right)+\beta I_{--},
	\tag{since $Z_i\otimes Z_{n+i}\ket{x,y}=(-1)^{x_i+y_i}\ket{x,y}$}
\end{align}
where
\begin{equation}\label{eq:sec:reduction_II_cleaner_1}
	\alpha =(M+1) \left(\mGHV + \frac{\epsilon}{2(n+1)} \right)^{-1}
	\quad\text{and}\quad
	\beta = n\mGHV + \frac{(n-1)\epsilon}{2(n+1)}.
\end{equation}

We now append the following \apath\  to the beginning of $H'(0)$:
\begin{equation} \label{eq:lem:pathB_head}
	\XB \pathto \XB +\frac12 \Delta \alpha \left( nI - \sum_{i \in [n]} Z_i \tensor Z_{n+i} \right) \pathto H'(0) \tag{P5}
\end{equation}

\begin{claim} \label{clm:lem:pathB_4}
	The \apath\ \Cref{eq:lem:pathB_head} has spectral gap $\Omega(\min\{\mGHV,\DeltaB\})$ and norm $\poly(\nGHV,\DeltaB)$.
	Consequently, it is $\poly(\nGHV,\DeltaB)$-Lipschitz by piecewise linearity.
\end{claim}
\begin{proof}
	We have shown that $(H'(0), \calE)$ simulates $K(0)$ with error $\epsilon=1/2$ in \Cref{sec:reduction_II_hypercube_to_X}.
	Therefore, every Hamiltonian on the second half of the \apath\ \Cref{eq:lem:pathB_head}, which has format $H'(0)-\lambda \beta I_{--}$, simulates $K(0) - \lambda \beta I$ by the same argument.
	Since $K(0)=-\mGHV\sum_{i\in[n]}X_i$ has spectral gap $\mGHV$, the spectral gap of the second half of \Cref{eq:lem:pathB_head} is at least $\mGHV - 2 \epsilon=\Omega(\mGHV)$ by \Cref{lem:gsim}.

	For the first half of \Cref{eq:lem:pathB_head}, applying \Cref{fct:ksum}, the spectral gap is identical to that of the \apath
	\begin{equation}
		-\DeltaB (X \tensor I + I \tensor X) \pathto -\DeltaB (X \tensor I + I \tensor X) - \frac12 \Delta \alpha (Z \tensor Z).
	\end{equation}
	Let $L(\lambda) \coloneqq -\DeltaB (X \tensor I + I \tensor X) - \lambda \cdot \frac12 \Delta \alpha (Z \tensor Z)$.
	Note that $L(\lambda)$ is block-diagonal in the Bell basis $\{\ket{00}\pm\ket{11},\ket{01}\pm\ket{10}\}$ and its eigenvalues are $\pm \frac12 \lambda \Delta \alpha$ and $\pm \sqrt{\left(\frac12 \lambda \Delta \alpha\right)^2+4 \DeltaB^2}$.
	Therefore, the spectral gap is
	\begin{align}
		\Omega\left(\min \left\{ \DeltaB, \frac{\DeltaB^2}{\lambda \Delta \alpha} \right\}\right)
		 & = \Omega\left(\min \left\{ \DeltaB, \frac{\Delta(M+1)/2}{\lambda \Delta (M+1) \left(\mGHV + \frac{\epsilon}{2(n+1)} \right)^{-1}} \right\}\right)
		\tag{by \Cref{eq:sec:reduction_II_cleaner_1} and \Cref{eq:lem:pathB_DeltaBdet}}                                                                      \\
		 & = \Omega\left(\min \left\{ \DeltaB, \frac{\mGHV}{\lambda} \right\}\right)
		=\Omega\left(\min\{\DeltaB,\mGHV\}\right)
		\tag{since $\lambda\in[0,1]$}
	\end{align}
	as claimed.
\end{proof}

\subsection{Putting Everything Together}\label{sec:reduction_II_together}

Finally we prove \Cref{lem:pathB} by combining the \apath\ in \Cref{sec:reduction_II_hypercube_to_X} and the additional decoration in \Cref{sec:reduction_II_cleaner}.

\begin{proof}[Proof of \Cref{lem:pathB}]
	Recall \apaths\ $H'([0,1])$ from \Cref{sec:reduction_II_hypercube_to_X} and \Cref{eq:lem:pathB_head}, \Cref{eq:lem:pathB_tail} from \Cref{sec:reduction_II_cleaner}.
	We sequentially connect \Cref{eq:lem:pathB_head}, $H'([0,1])$, and \Cref{eq:lem:pathB_tail}.
	This forms a \apath\ $H:[0,3] \to \mathcal{H}$. By rescaling the domain linearly from $[0,3]$ to $[0,1]$, we obtain the final path $H(t)$.

	Set $\DeltaB\in[\mGHV,\poly(\mGHV)]$ and recall that $\mGHV,\nGHV,n$ are polynomially related.
	By \Cref{clm:lem:pathB_1}, \Cref{clm:lem:pathB_2}, \Cref{clm:lem:pathB_3}, and \Cref{clm:lem:pathB_4}, $H(t)$ has spectral gap $\Omega(\sqrt{\mGHV})$, norm $\poly(\nGHV)$, and it is $\poly(\nGHV)$-Lipschitz.
	This verifies \Crefitem{lem:pathB}{itm:lem:pathB_1} and \Crefitem{lem:pathB}{itm:lem:pathB_3}.

	For \Crefitem{lem:pathB}{itm:lem:pathB_2}, we note that any query to $\mel{x}{H(t)}{y}$ reduces to queries of the $(x,y)$ entries of $H_0(t')$ and $V_\rmextra(t')$ for some $t' \coloneqq t'(t)$. By our construction, these can be implemented using $O(1)$ and $O(n)$ queries, respectively, to the entries of $K(t')$.
	Consequently, any $q$-query classical algorithm that finds $\ket{u}$ can also find $\ket{w}$, since $\ket{w} = \calE^\dagger \ket{u}$.
	This transformation results in an $O(qn)$-query classical algorithm with oracle access to $K(t')$.
	By \Crefitem{lem:pathA}{itm:lem:pathA_2}, we obtain $O(qn) \geq \exp(\nGHV^{1/5-o(1)})$.
	Since $n=\nA=\Theta(\mGHV)=\nGHV^{16/5-o(1)}$, this implies that $q\ge\exp(\nGHV^{1/5-o(1)})$.
\end{proof}

\subsection{Deferred Proofs}\label{sec:missing_proofs_in_sec:reduction_II}

\begin{proof}[Proof of \Cref{clm:reduction_II_first_step_Vmain_Vextra}]
	The first term is obvious by the definition of $\calH_-$.
	For the second term, observe that $H_0$ is diagonal and hence
	\begin{equation}
		\mel{x,y}{(H_0^{-1})_{++}}{x,y} = \begin{cases}
			-a_{x,y}\cdot(M+1)^{-1}                                               & \dist(x,y)=1, \\
			\left(\mGHV + \frac{\epsilon}{2(n+1)}\right)(M+1)^{-1}\dist(x,y)^{-1} & \dist(x,y)>1.
		\end{cases}
	\end{equation}
	Therefore
	\begin{align}
		\mel{x,x}{(V_\rmmain)_{-+} H_0^{-1} (V_\rmmain)_{+-}}{x,x}
		 & =\frac{M+1}2\cdot\sum_{\substack{(x',y') \\\dist(x,x')+\dist(x,y')=1}}\mel{x',y'}{(H_0^{-1})_{++}}{x',y'}\\
		 & =-\sum_{y\colon\dist(x,y)=1}a_{x,y}.
	\end{align}
	For $\dist(x,y)=1$, we have
	\begin{align}
		\mel{y,y}{(V_\rmmain)_{-+} H_0^{-1} (V_\rmmain)_{+-}}{x,x}
		 & =\frac{M+1}2\cdot\left(\mel{x,y}{H_0^{-1}}{x,y}+\mel{y,x}{H_0^{-1}}{y,x}\right) \\
		 & =-a_{x,y}.
	\end{align}
	For $\dist(x,y)>1$, we have
	\begin{align}
		\mel{y,y}{(V_\rmmain)_{-+} H_0^{-1} (V_\rmmain)_{+-}}{x,x}=0.
	\end{align}
	In summary, we have
	\begin{align}
		(V_\rmextra)_{--} - (V_\rmmain)_{-+} H_0^{-1} (V_\rmmain)_{+-}
		 & =\sum_{\substack{(x,y)                     \\\dist(x,y)\le1}}a_{x,y}\ketbra{x,x}{y,y}\\
		 & =\sum_{\substack{(x,y)                     \\\dist(x,y)\le1}}\left(\mel xKy-\frac{\epsilon}{2(n+1)}\right)\ketbra{x,x}{y,y}
		\tag{by \Cref{eq:lem:pathB_a-def}}            \\
		 & =\bar K-\frac{\epsilon}{2(n+1)}\cdot\bar E
	\end{align}
	as claimed.
\end{proof}

\begin{proof}[Proof of \Cref{clm:lem:pathB_3}]
	We only analyze the spectral gap of this \apath.
	For the first half of \Cref{eq:lem:pathB_tail}, we define for any $\lambda \in [0,1]$
	\begin{equation}
		L(\lambda) \coloneqq H'(1) - \lambda \cdot 3 \DeltaB \nB \ketbra{u}{u} = \Delta H_0 + (V - \lambda \cdot 3 \DeltaB \nB \ketbra{u}{u}).
	\end{equation}
	Recall from \Cref{eq:lem:pathB_DeltaBdet} that $\DeltaB=\sqrt{\frac{\Delta(M+1)}2}$ and $M=\poly(\nGHV)$.
	Setting $\DeltaB=\poly(\nGHV)$ sufficiently large, we have $\Delta>16\norm{V - \lambda \cdot 3  \DeltaB \nB \ketbra{u}{u}}$.
	By \Cref{lem:taylor}, the low-energy part of $L(\lambda)$ is given by
	\begin{equation}
		\underbrace{(V_{--} - \lambda \cdot 3 \DeltaB \nB \ketbra{u}{u}) - \Delta^{-1}V_{-+} H_0^{-1} V_{+-}}_{G(\lambda)} + O(\Delta^{-2} \norm{V - \lambda \cdot 3 \DeltaB \nB \ketbra{u}{u}}^3).
	\end{equation}
	The $O(\cdot)$ term above is $o(1)$ as $\DeltaB = \poly(\nGHV)$.
	Therefore, the spectral gap of $L(\lambda)$ is $o(1)$-close to the spectral gap of
	\begin{align}
		G(\lambda)
		 & = (V_{--} - \Delta^{-1}V_{-+} H_0^{-1} V_{+-}) - \lambda \cdot 3 \DeltaB \nB \ketbra{u}{u}                                   \\
		 & = \bar{K} - \frac{\epsilon}{2(n+1)} \bar{E} - \lambda \cdot 3 \DeltaB \nB \ketbra{u}{u}
		\tag{by \Cref{clm:reduction_II_first_step_Vmain_Vextra}}                                                                        \\
		 & = -\mGHV\cdot\calE \ketbra{w}{w} \calE^\dagger - \frac{\epsilon}{2(n+1)} \bar{E} - \lambda \cdot 3 \DeltaB \nB \ketbra{u}{u}
		\tag{since $K(1)=-\mGHV\ketbra w$}                                                                                              \\
		 & = -(\mGHV + 3 \lambda \DeltaB \nB) \ketbra{u}{u} - \frac{\epsilon}{2(n+1)} \bar{E}.
		\tag{since $\calE\ket{w} = \ket{u}$}
	\end{align}
	Recall that $E \coloneqq \sum_{x,y \colon \dist(x,y) \leq 1} \ketbra{x}{y}$ and $\norm { \overline{E} } = \norm{E} \leq n+1$.
	Therefore, the spectral gap of $L(\lambda)$, i.e., the first half of the \apath\ \Cref{eq:lem:pathB_tail}, is at least $\mGHV + 3 \lambda \DeltaB \nB - O(1)=\Omega(n\DeltaB)$ by \Cref{fct:weyl}.

	For the second half of \Cref{eq:lem:pathB_tail}, we rewrite it as
	\begin{equation}\label{eq:clm:lem:pathB_3_2}
		\XB + \Delta H_0 + V_\rmextra - 3 \DeltaB \nB \ketbra{u}{u} \pathto \XB - 3 \DeltaB \nB \ketbra{u}{u}.
	\end{equation}
	By \Cref{fct:weyl}, its spectral gap is at least $-2 \norm{\XB}$ plus the spectral gap of the \apath
	\begin{equation}\label{eq:clm:lem:pathB_3_1}
		\Delta H_0 + V_\rmextra - 3 \DeltaB \nB \ketbra{u}{u} \pathto - 3 \DeltaB \nB \ketbra{u}{u},
	\end{equation}
	which is a \apath\ of diagonal Hamiltonians.
	Let $\lambda\in[0,1]$.
	By \Cref{clm:reduction_II_first_step_H_0}, we know that $(\lambda\Delta H_0)_{--}=0$ and eigenvalues of $(\lambda\Delta H_0)_{++}$ are always non-negative.
	By \Cref{eq:lem:pathB_Vextra} and \Cref{eq:lem:pathB_a-def} and the fact that $\bar K=-\mGHV\ketbra u$, we know that
	\begin{equation}
		\lambda V_\rmextra=\lambda\cdot\left(-\mGHV\ketbra u-\frac{\epsilon(n-1)}{2(n+1)}\sum_x\ketbra{x,x}{x,x}\right)=\lambda\cdot\left(-\mGHV\ketbra u-\frac{\epsilon(n-1)}{2(n+1)}\cdot I_{--}\right).
	\end{equation}
	This means the spectral gap of $\lambda\left(\Delta H_0+V_\rmextra\right)-3 \DeltaB\nB\ketbra u$ (i.e., \Cref{eq:clm:lem:pathB_3_1}) is $3 \DeltaB\nB+\lambda \mGHV$.
	Since $\norm{\XB} = \norm{-\DeltaB \sum_{i\in\nB} X_i} = \DeltaB \nB$, \Cref{eq:clm:lem:pathB_3_2} has spectral gap at least $3 \DeltaB\nB-2\DeltaB\nB=\DeltaB\nB$.
	This verifies the spectral gap of the second half of \Cref{eq:lem:pathB_tail} and completes the proof of \Cref{clm:lem:pathB_3}.
\end{proof}     \section{Toy TFI Hamiltonian: Proof of Lemma \ref{lem:linearTFI}}\label{sec:toy_shift}

In this section we prove \Cref{lem:linearTFI} (restated below), which constructs a toy TFI Hamiltonian that gradually shifts ground states.

\lemlinearTFI*

\paragraph{Notation.}
We recall related definitions here.
For each $j=0,1,\ldots,\ell$, we define $\ket{\bs{j}} \coloneqq \ket{1^j 0^{\ell-j}}$.
Let $\epsilonL, \DeltaL>0$ be parameters to be determined.
Define the TFI Hamiltonian $\HL(t) = \QL + \XL + \ZL(t)$ where $\XL = -\sum_{i \in [\ell]} X_i$ and
\begin{equation}\label{eq:sec:toy_shift_1}
	\QL = \DeltaL \left( (\ell-1)I + Z_1 - Z_\ell - \sum_{i \in [\ell-1]} Z_{i}Z_{i+1} \right),
	\quad
	\ZL(t) = \frac{1}{\epsilonL} \sum_{i \in [\ell]} (i-t)(I - Z_i).
\end{equation}

We first verify the easy norm condition.

\begin{proof}[Proof of \Crefitem{lem:linearTFI}{itm:lem:linearTFI_1}]
	By the definition of $\HL(t)$, we have
	\begin{align}
		\norm{\HL(t)}
		 & \le\norm{\QL}+\norm{\XL}+\norm{\ZL(t)}                               \\
		 & \le\DeltaL\cdot10\ell+\ell+\frac2{\epsilonL}\sum_{i\in[\ell]}(i+|t|)
		\tag{by \Cref{eq:sec:toy_shift_1}}                                      \\
		 & \le(|t|+1)\poly(\ell,1/\epsilonL),
	\end{align}
	where we use the fact that $\ell\ge1$, $\epsilonL\in(0,1]$, and $\DeltaL=\poly(\ell,1/\epsilonL)$ for the last inequality.
\end{proof}

To verify other properties, we expand $\HL(\lambda,t)$ as follows:
\begin{align}
	\HL(\lambda,t)
	 & =\QL+\lambda\XL+\ZL(t)                                                                                                                              \\
	 & =-\lambda\sum_{i\in[\ell]}X_i+\DeltaL\left((\ell-1)I+Z_1-Z_\ell-\sum_{i\in[\ell-1]}Z_iZ_{i+1}\right)+\frac1{\epsilonL}\sum_{i\in[\ell]}(i-t)(I-Z_i) \\
	 & =
	-\lambda\sum_{(x,x')\colon\dist(x,x')=1}\ketbra{x}{x'}
	+\sum_{x\in\{0,1\}^\ell}(\alpha_x+\beta_x)\ketbra x,
	\label{eq:sec:toy_shift_2}
\end{align}
where
\begin{align}\label{eq:sec:toy_shift_3}
	\alpha_x
	=\DeltaL\left(\ell-1+(-1)^{x_1}-(-1)^{x_\ell}-\sum_{i\in[\ell-1]}(-1)^{x_i+x_{i+1}}\right)
\end{align}
and
\begin{align}\label{eq:sec:toy_shift_4}
	\beta_x=\frac1{\epsilonL}\sum_{i\in[\ell]}(i-t)\left(1-(-1)^{x_i}\right).
\end{align}
Note that $\HL(t)=\HL(1,t)$.

Recall that $\ket{\bs j}=\ket{1^j0^{\ell-j}}$. Define $\mathcal I=\{\bs0,\bs1,\ldots,\bs\ell\}$.
We also note the following easy claim for $\alpha_x,\beta_x$, the proof of which is deferred to \Cref{sec:missing_proofs_in_sec:toy_shift}.

\begin{claim}\label{clm:toy_shift_tilde_HL_alpha_x}
	For $x=\bs j\in\mathcal I$, we have $\alpha_x=0$ and $\beta_x=j(j+1-2t)/\epsilonL$.
	For $x\notin\mathcal I$, we have $\alpha_x\ge\DeltaL$ and $|\beta_x|\le2(\ell^2+|t|\cdot\ell)/\epsilonL$.
\end{claim}

In the rest of the section, we verify the rest of the properties in \Cref{lem:linearTFI}.
This is divided into two cases depending on the range of $t$.

\paragraph{Organization.}
In \Cref{sec:toy_shift_unbounded}, we prove \Cref{lem:linearTFI} for the unbounded time regime where $|t|$ is extremely large.
In \Cref{sec:toy_shift_bounded}, we handle the bounded time regime for reasonably large $|t|$.
In \Cref{sec:toy_shift_tridiag}, we prove some useful spectral results of certain tridiagonal matrices, which are used in \Cref{sec:toy_shift_bounded}.
Missing proofs from the section can be found in \Cref{sec:missing_proofs_in_sec:toy_shift}.

\subsection{Unbounded Time Regime}\label{sec:toy_shift_unbounded}

We first handle the unbounded time regime where $|t|\ge10\ell^2$.

In this case, we define
\begin{align}
	\HL'(t)=\sum_{x\in\{0,1\}^\ell}(\alpha_x+\beta_x)\ketbra x.
\end{align}
Since $\HL(\lambda,t)-\HL'(t)=-\lambda\sum_{i\in[\ell]}X_i$. we have the following easy fact.

\begin{claim}\label{clm:toy_shift_unbounded_diff}
	If $\lambda\in[0,2]$, then $\norm{\HL(\lambda,t)-\HL'(t)}=\lambda\ell$.
\end{claim}

Since $\HL'(t)$ is diagonal, its spectrum is also easy to analyze.

\begin{claim}\label{clm:toy_shift_unbounded_HL'}
	Assume $\ell\ge1$ and $\DeltaL\ge100\ell^4/\epsilonL$.
	\begin{itemize}
		\item It $t\le-10\ell^2$, then $\HL'(t)$ has spectral gap at least $4\ell/\epsilonL$ and $\ket{\bs0}$ is its ground state.
		\item If $t\ge10\ell^2$, then $\HL'(t)$ has spectral gap at least $4\ell/\epsilonL$ and $\ket{\bs\ell}$ is its ground state.
	\end{itemize}
\end{claim}
\begin{proof}
	By \Cref{clm:toy_shift_tilde_HL_alpha_x}, we first note that for $x\notin\mathcal I$, we have
	\begin{align}\label{eq:clm:toy_shift_unbounded_HL'_1}
		\alpha_x+\beta_x\ge\DeltaL-\frac{2(\ell^2+|t|\cdot\ell)}{\epsilonL}\ge\frac{50\ell^4}{\epsilonL}
		\quad\text{for $|t|\le10\ell^2$.}
	\end{align}

	Assume $t\le-10\ell^2$ and apply \Cref{clm:toy_shift_tilde_HL_alpha_x}.
	For $x=\bs0$, we have $\alpha_x+\beta_x=0$.
	For $x=\bs j$ where $j\in[\ell]$, we have $\alpha_x+\beta_x=j(j+1-2t)/\epsilonL\ge4\ell/\epsilonL$.
	Combining with \Cref{eq:clm:toy_shift_unbounded_HL'_1}, we know that the spectral gap of $\HL'(t)$ is at least $4\ell/\epsilonL$ and $\ket{\bs0}$ is its ground state.

	Assume $t\ge10\ell^2$ and the calculation is similar.
	For $x=\bs\ell$, we have $\alpha_x+\beta_x=\ell(\ell+1-2t)/\epsilonL<0$.
	For $x=\bs j$ where $j=0,1,\ldots,\ell-1$, we have
	\begin{align}
		(\alpha_x+\beta_x)-(\alpha_{\bs\ell}+\beta_{\bs\ell})
		=\frac{j(j+1-2t)-\ell(\ell+1-2t)}{\epsilonL}
		\ge\frac{2t-\ell^2-\ell}{\epsilonL}
		\ge\frac{4\ell}{\epsilonL}.
	\end{align}
	Combining with \Cref{eq:clm:toy_shift_unbounded_HL'_1}, we know that the spectral gap of $\HL'(t)$ is at least $4\ell/\epsilonL$ and $\ket{\bs\ell}$ is its ground state.
\end{proof}

Now we can complete the proof of \Cref{lem:linearTFI} in this unbounded time regime.
Note that \Crefitem{lem:linearTFI}{itm:lem:linearTFI_3} does not apply in this regime and hence we only verify \Crefitem{lem:linearTFI}{itm:lem:linearTFI_2} and \Crefitem{lem:linearTFI}{itm:lem:linearTFI_4}.

\begin{proof}[Proof of \Crefitem{lem:linearTFI}{itm:lem:linearTFI_2} and \ref{itm:lem:linearTFI_4}: the $|t|\ge10\ell^2$ Case]
	For \Crefitem{lem:linearTFI}{itm:lem:linearTFI_2}, by \Cref{fct:weyl}, the spectral gap of $\HL(t)=\HL(1,t)$ is at least the spectral gap of $\HL'(t)$ minus $2\norm{\HL(1,t)-\HL'(t)}$.
	By \Cref{clm:toy_shift_unbounded_HL'} and \Cref{clm:toy_shift_unbounded_diff}, this is at least $(4\ell/\epsilonL)-2\ell\ge2\ell\ge 2$.

	For \Crefitem{lem:linearTFI}{itm:lem:linearTFI_4}, we apply \Cref{fct:davis-kahan} with \Cref{clm:toy_shift_unbounded_diff} and \Cref{clm:toy_shift_unbounded_HL'}, which implies that the ground states of $\HL(\lambda,t)$ and $\HL'(t)$ are within distance
	\begin{equation}
		\frac{2^{3/2}\cdot1\cdot2\ell}{4\ell/\epsilonL}\le\epsilonL
	\end{equation}
	in the $\ell_2$ norm.
\end{proof}

\subsection{Bounded Time Regime}\label{sec:toy_shift_bounded}

Now we handle the bounded time regime where $|t|\le10\ell^2$.

We will show that $\HL(\lambda,t)$ embeds a smaller Hamiltonian of dimension $\ell+1$.
To specify this Hamiltonian, we denote $\calH$ by the Hilbert space associated with $\HL(\lambda,t)$ and let $\calK$ be the Hilbert space spanned by $\ket{0},\ket1,\ldots,\ket\ell$.
Define the encoding isometry $\calE\colon\calK\to\calH$ by
\begin{equation}\label{eq:sec:toy_shift_5}
	\calE\ket j=\ket{\bs j}=\ket{1^j0^{\ell-j}}.
\end{equation}
Let $\calH_-:=\Im\calE$ and we write $\calH=\calH_-\oplus\calH_+$.
Define Hamiltonian $\tilde\HL(\lambda,t)$ by
\begin{align}\label{eq:sec:toy_shift_6}
	\tilde\HL(\lambda,t) \coloneqq
	-\lambda\sum_{j\in[\ell]}\left(\ketbra{j-1}{j}+\ketbra{j}{j-1}\right)
	+\frac1{\epsilonL}\sum_{j=0}^\ell\beta_{\bs j}\ketbra j.
\end{align}

We now show by first-order reduction (\Cref{lemma:1st}) that $\HL(\lambda,t)$ simulates $\tilde\HL(\lambda,t)$.

\begin{lemma}\label{lem:toy_shift_HL_to_tilde_HL}
	Assume $|t|\le10\ell^2$, $\lambda\in[0,2]$, and $\DeltaL=\poly(\ell,1/\epsilonL)$.
	Then $\HL(\lambda,t)$ and $\calE$ simulate $\tilde\HL(\lambda,t)$ with error $\epsilonL$.
\end{lemma}
\begin{proof}
	Fix $\lambda,t$ and shorthand $\HL=\HL(\lambda,t),\tilde\HL=\tilde\HL(\lambda,t)$.
	In light of \Cref{lemma:1st}, define
	\begin{equation}
		V=-\lambda\sum_{\substack{(x,x')\colon\dist(x,x')=1\\x\in\mathcal I~\lor~x'\in\mathcal I}}\ketbra{x}{x'}+\sum_{x\in\{0,1\}^\ell}\beta_x\ketbra x
	\end{equation}
	and
	\begin{align}
		H_0
		 & =\tilde\HL-V=-\lambda\sum_{\substack{(x,x')\colon\dist(x,x')=1 \\x\notin\mathcal I~\land~x'\notin\mathcal I}}\ketbra{x}{x'}
		+\sum_{x\in\{0,1\}^\ell}\alpha_x\ketbra x                         \\
		 & =-\lambda\sum_{\substack{(x,x')\colon\dist(x,x')=1             \\x\notin\mathcal I~\land~x'\notin\mathcal I}}\ketbra{x}{x'}
		+\sum_{x\notin\mathcal I}\alpha_x\ketbra x.
		\tag{by \Cref{clm:toy_shift_tilde_HL_alpha_x}}
	\end{align}
	Hence $H_0$ is block-diagonal and $(H_0)_{--}=0$.
	In addition by \Cref{fct:weyl} and \Cref{clm:toy_shift_tilde_HL_alpha_x}, the minimal eigenvalue of $(H_0)_{++}$ is at least
	\begin{align}
		\min_{x\notin\mathcal I}\alpha_x-2\lambda\ell
		\ge\DeltaL-2\lambda\ell\ge\DeltaL/2,
	\end{align}
	where we recall $\lambda\in[0,2]$ and $\DeltaL=\poly(\ell,1/\epsilonL)$ for the last inequality.
	To apply \Cref{lemma:1st}, it remains to show $\DeltaL/2\ge O(\epsilonL^{-1}\norm{V}^2)$.
	To this end, we simply note that
	\begin{align}
		\norm{V}
		 & \le\lambda\ell+\max_x|\beta_x|
		\le\poly(|t|,\lambda,\ell,1/\epsilonL)
		\tag{by \Cref{clm:toy_shift_tilde_HL_alpha_x}} \\
		 & =\poly(\ell,1/\epsilonL).
		\tag{since $|t|\le10\ell^2$ and $\lambda\in[0,2]$}
	\end{align}
	This completes the proof of \Cref{lem:toy_shift_HL_to_tilde_HL}.
\end{proof}

By \Cref{lem:toy_shift_HL_to_tilde_HL}, we aim to study the spectrum of $\tilde\HL(\lambda,t)$, which is essentially the low energy spectrum of $\HL(\lambda,t)$.
Since $\tilde\HL(\lambda,t)$ is a tridiagonal matrix, we need the following result which will be proved in \Cref{sec:toy_shift_tridiag}.

\begin{lemma}\label{lem:convex}
	Assume $0<\epsilonL \leq 1/10$ and $\lambda\in[0,2]$. Then $\tilde\HL(\lambda,t)$ satisfies the following properties.
	\begin{enumerate}[label=(\roman*)]
		\item\label{itm:lem:convex_3}
		For $t \in [0,\ell+1]$, let $\ket{\psi}$ be a ground state of $\tilde\HL(1,t)$. Then the spectral gap of $\tilde\HL(1,t)$ is at least $3/4$ and $\ket{\psi}$ is $O(\epsilonL)$-close to a normalized state that is a linear combination of $\ket{a}$ and $\ket{b}$, where $a \coloneqq \max\{\lfloor t-1/2 \rfloor,0\}$ and $b \coloneqq \min\{\lceil t-1/2 \rceil,\ell\}$.
		\item\label{itm:lem:convex_4}
		For any $\lambda \in [0,2]$ and $t \leq 0$ (resp., $t \geq \ell+1$), let $\ket{\psi}$ be a ground state of $\tilde\HL(\lambda ,t)$. Then the spectral gap of $\tilde\HL(1,t)$ is at least $3/4$ and $\ket{\psi}$ is $O(\epsilonL)$-close to $\ket{0}$ (resp., $\ket{\ell}$).
	\end{enumerate}
\end{lemma}

Given \Cref{lem:convex} and \Cref{lem:toy_shift_HL_to_tilde_HL}, we prove \Cref{lem:linearTFI} in this bounded time regime.
Together with the analysis in the unbounded time regime (\Cref{sec:toy_shift_unbounded}), this completes the whole proof of \Cref{lem:linearTFI}.

\begin{proof}[Proof of \Crefitem{lem:linearTFI}{itm:lem:linearTFI_2}, \ref{itm:lem:linearTFI_3}, and \ref{itm:lem:linearTFI_4}: the $|t|\le10\ell^2$ Case]
	Recall from \Cref{lem:toy_shift_HL_to_tilde_HL} that ${\HL}(\lambda,t)$ and $\calE$ simulate $\tilde{\HL}(\lambda,t)$ with error $\epsilonL$, where $\calE \colon \ket{j} \mapsto \ket{\bs{j}}$.

	The spectral gap of $\HL(t) = \HL(1,t)$ is thus at least that of $\tilde{\HL}(1,t)$ minus $2 \epsilonL$ by \Cref{fct:weyl}, which is $3/4 - 2\epsilonL \geq 1/2$ by \Crefitem{lem:convex}{itm:lem:convex_3} and \ref{itm:lem:convex_4}, which concludes the proof of \Crefitem{lem:linearTFI}{itm:lem:linearTFI_2} and the spectral gap part of \Crefitem{lem:linearTFI}{itm:lem:linearTFI_4}.

	Now we turn to ground states.
	By \Crefitem{lem:convex}{itm:lem:convex_3}, the ground state of $\tilde{\HL}(1,t)$ is $O(\epsilonL)$-close to a normalized state that is a linear combination of $\ket{a}$ and $\ket{b}$.
	Thus by \Cref{lem:gsim}, the ground state of $\HL(t) = \HL(1,t)$ is then $(O(\epsilonL)+\epsilonL)$-close to that of $\ket{\bs{a}}$ and $\ket{\bs{b}}$, as claimed in \Crefitem{lem:linearTFI}{itm:lem:linearTFI_3}.
	Similarly, \Crefitem{lem:convex}{itm:lem:convex_4} implies the ground state part of \Crefitem{lem:linearTFI}{itm:lem:linearTFI_4}.
\end{proof}

\subsection{Properties of Certain Tridiagonal Matrices}\label{sec:toy_shift_tridiag}

To prove \Cref{lem:convex}, we introduce \Cref{lem:tridiag_eigengap_single} and \Cref{lem:tridiag_eigengap_double}.
Both lemmas analyze the smallest eigenvalues, and corresponding eigenspaces, of the tridiagonal matrix
\begin{equation} \label{eq:sec:toy_shift_Adef}
	A \coloneqq \begin{bmatrix}
		p_0      & -\lambda &          &            &            \\
		-\lambda & p_1      & -\lambda &            &          & \\
		         & \ddots   & \ddots   & \ddots     &            \\
		         &          & -\lambda & p_{\ell-1} & -\lambda   \\
		         &          &          & -\lambda   & p_\ell
	\end{bmatrix}
	\quad\text{where}\quad\lambda > 0
\end{equation}
but under slightly different assumptions.

\begin{lemma}\label{lem:tridiag_eigengap_single}
	Let $ \ket{\psi} \in\mathbb R^{\ell+1}$ be a ground state of $A$.
	Assume for some $i \in \{0,1,\dots,\ell\}$ and $C>0$, we have $p_j\ge p_i+C$ for all $j\in\{0,1,\dots,\ell\}\setminus\{i\}$.
	Then the spectral gap of $A$ is at least $C-4 \lambda$, and $\ket{\psi}$ is $(8\lambda/C)$-close to $\alpha\ket{i}$ for some $\abs{\alpha}=1$.
\end{lemma}
\begin{proof}
	Let $D=\mathrm{diag}(p_0,\ldots,p_\ell)$.
	Then $A=D+E$ where $E$ contains the off-diagonal terms in $A$ and $\norm{E}\le2\lambda$.
	Hence by \Cref{fct:weyl}, the spectral gap of $A$ is lower bounded by the spectral gap of $D$ (which is $-p_i+\min_{j\ne i}p_j$) minus $2\cdot\norm{E}\le4\lambda$.
	Since $p_j\ge p_i+C$, the spectral gap of $A$ is at least $C-4\lambda$.

	Now we argue about $v_i \coloneqq \braket{i}{\psi}$. First assume without loss of generality $v_i\ge0$.
	Then observe that the ground space of $D$ is spanned by $\ket{i}$.
	Hence by \Cref{fct:davis-kahan}, we have
	\begin{equation}
		\norm{\alpha \ket{i} - \ket{\psi}} \le\frac{2^{3/2}\norm{E}}C\le\frac{8\lambda}C
	\end{equation}
	as desired.
\end{proof}

\begin{lemma}\label{lem:tridiag_eigengap_double}
	Let $ \ket{\psi} \in\mathbb R^{\ell+1}$ be a ground state of $A$.
	Assume for some $i \in [\ell]$, we have $|p_i-p_{i-1}|\le\epsilon$ and $p_j\ge\max\{p_{i-1},p_{i}\}+C$ for all $j\in \{0,1,\dots,\ell\} \setminus\{i-1,i\}$, where $C\ge4\lambda+\epsilon$.
	Then the spectral gap of $A$ is at least $3\lambda/4$, and $\ket{\psi}$ is $(8\lambda/C)$-close to $\alpha \ket{i-1} + \beta \ket{i}$ for some $\alpha,\beta$ such that $\abs{\alpha}^2 + \abs{\beta}^2=1$.
\end{lemma}
\begin{proof}
	Define matrix $E$ to have the following non-zero entries: $\mel{i-2}{E}{i-1} = \mel{i-1}{E}{i-2}=\mel{i}{E}{i+1} = \mel{i+1}{E}{i}=-\lambda$ and $\mel{i-1}{E}{i} = \mel{i}{E}{i-1} =\lambda$ (omit if indices run out of range).
	Then we have the following cases:
	\begin{itemize}
		\item if $i=1$, then $E$ is $\begin{bmatrix}&\lambda\\\lambda & & -\lambda\\&-\lambda\end{bmatrix}$ padded with zeros, which has $\norm{E}=\sqrt2\cdot\lambda$;
		\item if $i=\ell$, then $E$ is $\begin{bmatrix}&-\lambda\\-\lambda & &\lambda\\&\lambda\end{bmatrix}$ padded with zeros, which has $\norm{E}=\sqrt2\cdot\lambda$;
		\item otherwise $2\le i\le \ell-1$, then $E$ is $\begin{bmatrix}&-\lambda\\-\lambda & &\lambda\\&\lambda&&-\lambda\\&&-\lambda\end{bmatrix}$ padded with zeros, which has $\norm{E}=\frac{\sqrt5+1}2\cdot\lambda$.
	\end{itemize}
	Let $D=A-E$.
	Then $D$ is block-diagonal of tridiagonal matrices:
	\begin{equation}
		D=\begin{bmatrix}
			p_0    & -\lambda                                                                   \\
			\ddots & \ddots   & \ddots                                                          \\
			       & -\lambda & p_{i-2}                                                         \\
			       &          &         & p_{i-1}   & -2\lambda                                 \\
			       &          &         & -2\lambda & p_{i}                                     \\
			       &          &         &           &           & p_{i+1} & -\lambda            \\
			       &          &         &           &           & \ddots  & \ddots   & \ddots   \\
			       &          &         &           &           &         & -\lambda & p_{\ell}
		\end{bmatrix},
	\end{equation}
	where the middle block $\begin{bmatrix}p_{i-1}&-2\lambda\\-2\lambda&p_{i}\end{bmatrix}$ has eigenvalues
	\begin{equation}
		\mu_1=\frac12\left(p_{i-1}+p_{i}-\sqrt{16\lambda^2+(p_{i-1}-p_{i})^2}\right)
		\quad\text{and}\quad
		\mu_2=\frac12\left(p_{i-1}+p_{i}+\sqrt{16\lambda^2+(p_{i-1}-p_{i})^2}\right).
	\end{equation}
	Note that
	\begin{equation}
		\mu_1\le\frac{p_{i-1}+p_{i}}2-2\lambda<\frac{p_{i-1}+p_{i}}2+2\lambda\le\mu_2\le\frac{p_{i-1}+p_{i}}2+2\lambda+\frac\epsilon2.
	\end{equation}
	Whereas the eigenvalues of $D$ outside the middle block are, by \Cref{fct:weyl}, at least $\min_{j\neq i-1,i}\{p_j\}-2\lambda\ge\max\{p_{i-1},p_{i}\}+C-2\lambda$.
	Since $C\ge4\lambda+\epsilon$, we know that $\mu_1,\mu_2$ are the smallest two eigenvalues of $D$.
	Hence by \Cref{fct:weyl} again, the spectral gap of $A$ is lower bounded by the spectral gap of $D$ (which is $\mu_2-\mu_1\ge4\lambda$) minus $2\cdot\norm{E}\le(\sqrt5+1)\lambda$.
	Since $\sqrt5+1<3.25$, the spectral gap of $A$ is at least $3\lambda/4$ as claimed.

	To argue about the eigenvector $\ket{\psi}$, notice that $\{\ket{i-1},\ket{i}\}$ span the ground space of $D$.
	Let $\ket{\phi}$ be a eigenvector of $A$ corresponding to the second smallest eigenvalue.
	By \Cref{fct:davis-kahan}, there is a $2 \times 2$ unitary $U$ such that
	\begin{equation}
		\norm{R}_{\rm F} \leq \frac{2^{3/2}\cdot2^{1/2}\cdot\norm{E}}{C} \le\frac{2\cdot(\sqrt5+1)\cdot\lambda}{C}\le\frac{8\lambda}{C},
	\end{equation}
	where $R \coloneqq (\ketbra{i-1}{\texttt{1}} + \ketbra{i}{\texttt{2}})U - (\ketbra{\psi}{\texttt{1}} + \ketbra{\phi}{\texttt{2}})$.
	Therefore, the desired bound can be obtained via
	\begin{equation}
		\norm{ \mel{\texttt{1}}{U}{\texttt{1}} \ket{i-1} + \mel{\texttt{2}}{U}{\texttt{1}} \ket{i} - \ket{\psi}  }_{\rm F} = \norm{R \ket{\texttt{1}}}_{\rm F} \leq \norm{R}_{\rm F} \norm{\ket{\texttt{1}}} = \norm{R}_{\rm F},
	\end{equation}
	where the last inequality follows from the sub-multiplicativity of Frobenius norm.
	Setting $\alpha=\mel{\texttt{1}}{U}{\texttt{1}}$ and $\beta=\mel{\texttt{2}}{U}{\texttt{1}}$ completes the proof of \Cref{lem:tridiag_eigengap_double}.
\end{proof}

Now we are ready to prove \Cref{lem:convex}.

\begin{proof}[Proof of \Cref{lem:convex}]
	Recall from \Cref{eq:sec:toy_shift_6} and \Cref{clm:toy_shift_tilde_HL_alpha_x} that
	\begin{equation}
		\tilde\HL(\lambda,t) = -\lambda\sum_{j\in[\ell]}\left(\ketbra{j-1}{j}+\ketbra{j}{j-1}\right) +\frac1{\epsilonL}\sum_{j=0}^\ell j(j+1-2t) \ketbra j.
	\end{equation}
	We write $p_j \coloneqq \mel{j}{\tilde\HL(\lambda,t)}{j}$ for $j=0,1,\dots,\ell$ and therefore $\tilde\HL(\lambda,t)$ is exactly $A$ defined in \Cref{eq:sec:toy_shift_Adef}.

	For \Cref{itm:lem:convex_4}, we consider two cases separately:
	\begin{itemize}
		\item If $t \leq 0$, we apply \Cref{lem:tridiag_eigengap_single} with $(C,i)=(2/\epsilonL,0)$ to get the spectral gap lower bound $12$ and that $\ket{\psi}$ is $O(1/\epsilonL)$-close to $\ket{0}$ since $\epsilonL \leq 1/10$ and $\lambda \leq 2$ by assumption.
		\item If $t \geq \ell + 1$, similarly we apply  \Cref{lem:tridiag_eigengap_single} with $(C,i)=(2/\epsilonL,\ell + 1)$.
	\end{itemize}

	To prove \Cref{itm:lem:convex_3},
	we divide $t$ into separate cases and frequently use the following equation:
	\begin{equation}
		\epsilonL(p_i - p_j) = i(i+1-2t) - j(j+1-2t) = (i-j)(i+j+1-2t).
	\end{equation}
	\begin{itemize}
		\item Assume $t\in[0,3/4]$.
		      Then for any $i\ge 1$,
		      \begin{equation}
			      \epsilonL(p_i - p_0) =  i(i+1-2t) \underbrace{\ge}_{\text{set }i=1}2-2t\ge\frac12.
		      \end{equation}
		      Hence we can apply \Cref{lem:tridiag_eigengap_single} with $C=1/(2\epsilonL)$.
		\item Assume $t\in[\ell+1/4,\ell+1]$.
		      Then for any $i\le\ell-1$,
		      \begin{equation}
			      \epsilonL(p_i - p_\ell)= (\ell-i)(2t-\ell-1-i) \underbrace{\ge}_{\text{set }i=\ell-1}2t - 2\ell \ge\frac12.
		      \end{equation}
		      Hence we can apply \Cref{lem:tridiag_eigengap_single} with $C=1/(2\epsilonL)$.
		\item Assume $t\in[i-1/4,i+1/4]$ for some integer $i=1,2,\ldots,\ell-1$.
		      Then $\abs{\epsilonL(p_i - p_{i-1})}=\abs{2i-2t}\le1/2$.
		      For $j\ne i-1,i$,
		      \begin{equation}
			      \epsilonL(p_j - p_i) \underbrace{\ge}_{j=i-2\text{ or }i+1}\min\{2(2t+1-2i),2(i+1-t)\}\ge 1,
		      \end{equation}
		      and
		      \begin{equation}
			      \epsilonL(p_j - p_{i-1}) \underbrace{\ge}_{j=i-2\text{ or }i+1}\min\{2(t+1-i),2(2i+1-2t)\}\ge 1.
		      \end{equation}
		      Hence we can apply \Cref{lem:tridiag_eigengap_double} with $(\epsilon,C)=(1/(2\epsilonL), 1/\epsilonL)$.
		\item Otherwise $t\in[i+1/4,i+3/4]$ for some integer $i=1,2,\ldots,\ell-1$.
		      Then for any $j\ne i$,
		      \begin{equation}
			      \epsilonL(p_j - p_i) \underbrace{\ge}_{j=i-1\text{ or }i+1}\min\{2(t-i),2(i+1-t)\}\ge\frac12.
		      \end{equation}
		      Hence we can apply \Cref{lem:tridiag_eigengap_single} with $C=1/(2\epsilonL)$.
		      \qedhere
	\end{itemize}
\end{proof}

\subsection{Deferred Proofs}\label{sec:missing_proofs_in_sec:toy_shift}

\begin{proof}[Proof of \Cref{clm:toy_shift_tilde_HL_alpha_x}]
	We do a direct calculation based on \Cref{eq:sec:toy_shift_2} and \Cref{eq:sec:toy_shift_3}.
	If $x=\bs j$, then
	\begin{equation}
		\sum_{i\in[\ell-1]}(-1)^{x_i+x_{i+1}}=\begin{cases}
			\ell-1, & j=0,\ell,     \\
			\ell-3, & j\in[\ell-1],
		\end{cases}
		\quad\text{and}\quad
		(-1)^{x_1}-(-1)^{x_\ell}=\begin{cases}
			0,  & j=0,\ell,     \\
			-2, & j\in[\ell-1],
		\end{cases}
	\end{equation}
	which implies $\alpha_x=0$.
	We also have
	\begin{equation}
		\sum_{i\in[\ell]}(i-t)\left(1-(-1)^{x_i}\right)
		=\sum_{1\le i\le j}2(i-t)=j(j+1-2t),
	\end{equation}
	which implies $\beta_x=j(j+1-2t)/\epsilonL$.

	Now assume $x\notin\mathcal I$.
	By \Cref{eq:sec:toy_shift_4}, we have
	\begin{align}
		|\beta_x|\le\frac1{\epsilonL}\sum_{i\in[\ell]}(i+|t|)\cdot 2=\frac{\ell(\ell+1)+2|t|\cdot\ell}{\epsilonL}\le\frac{2(\ell^2+|t|\cdot\ell)}{\epsilonL},
	\end{align}
	where we use $\ell\ge1$ for the last inequality.
	Regarding $\alpha_x$, define $s=\abs{\{i\in[\ell-1]\colon x_i\ne x_{i+1}\}}$.
	Since $x\notin\mathcal I$, we either have (1) $s\ge2$ or (2) $s=1$ but $x=\ket{0^j1^{\ell-j}}$ for some $1\le j\le\ell-1$.
	In case (1), we have
	\begin{align}
		\alpha_x
		\ge\DeltaL\left(\ell-1-2-(-s+\ell-1-s)\right)\ge2\DeltaL(s-1)\ge\DeltaL.
	\end{align}
	In case (2), we have
	\begin{align}
		\alpha_x=\DeltaL\left(\ell-1+1+1-(-s+\ell-1-s)\right)=4\DeltaL\ge\DeltaL.
	\end{align}
	Hence we always have $\alpha_x\ge\DeltaL$.
\end{proof}     \section{Reduction III: Proof of Lemma \ref{lem:pathC}}\label{sec:reduction_III}

We prove \Cref{lem:pathC} (restated below) in this section.

\lempathC*

The proof idea is to view the shifts of the $\ell+1$ Hamiltonians as the ground state shifts of $\HL(t)$ in \Cref{lem:linearTFI}, which is formalized via a first-order perturbative reduction (\Cref{lemma:1st}).

\paragraph{Notation.}
We (re)define some notation.
Let $n \coloneqq \nB$, $\Delta \coloneqq \DeltaC$, and $\epsilon \coloneqq \epsilonL$.
For $\tau \in \binary^n$, define
\begin{equation}
	K_\tau \coloneqq
	\begin{cases}
		\XB + D_j & \tau=\bs{j}= 1^j0^{\ell-j} \text{ for some $j$}, \\
		\XB       & \text{otherwise.}
	\end{cases}
\end{equation}
Note that by \Crefitem{lem:pathB2}{itm:lem:pathB_3} the \apath
\begin{equation} \label{eq:lem:pathC_K}
	K_{\bs{0}} \pathto K_{\bs{1}}\pathto \cdots \pathto K_{\bs{\ell}}
\end{equation}
has spectral gap at least $\delta := \Omega(\sqrt{\mGHV})$ and norm $\poly(\nGHV)$.
Denote $M \coloneqq \max_{\tau \in \binary^\ell} \norm{K_\tau}$ and we have $M \leq \poly(\nGHV)$ by \Crefitem{lem:pathB2}{itm:lem:pathB_3} and the bound $\DeltaB \leq \poly(\nGHV)$ in \Cref{lem:pathB}.

Let $t$ be fixed.
We shorthand $\HL \coloneqq \HL(t)$, $H_0 \coloneqq \Delta \cdot I_{2^n} \tensor\HL(t)$, and $V \coloneqq \XB \tensor I_{2^\ell} + \DC$.
The Hamiltonian of interest $\HC\coloneqq\HC(t)$ in \Cref{eq:HC} then becomes
\begin{equation}
	\HC =  H_0 + V,
	\quad\text{where}\quad
	H_0 = \Delta \cdot I_{2^n} \tensor \HL,\quad
	V = \sum_{\tau \in \binary^\ell} K_\tau \tensor \ketbra{\tau}{\tau}.
\end{equation}
We also denote $a \coloneqq \max\{\lfloor t-1/2 \rfloor,0\}$ and $b \coloneqq \min\{\lceil t-1/2 \rceil,\ell\}$, as in \Crefitem{lem:linearTFI}{itm:lem:linearTFI_3}.

Now we prove \Cref{lem:pathC}.

\begin{proof}[Proof of \Cref{lem:pathC}]
	Let $\calH$ be the Hilbert space associated with $H_0$, $\HC$, and $V$.
	Let $\calH_{-}$ be the ground space of $H_0$.
	Let $\zeta$ be the ground energy of $\HL$ and let $\ket v$ be a ground state of $\HL$.
	By definition, $\calH_{-}$ is spanned by $\ket{\psi} \tensor \ket{v}$ for any $n$-qubit state $\ket{\psi}$.
	We explicitly write out $\ket{v}$ as $\ket{v} = \sum_{\tau \in \binary^\ell} v_\tau \ket{\tau}$.
	Then by \Crefitem{lem:linearTFI}{itm:lem:linearTFI_3}, we have
	\begin{equation} \label{eq:lem:pathC_vtau}
		\sum_{\tau \notin \{\bs{a},\bs{b}\}} v_\tau^2 \leq O(\epsilon^2).
	\end{equation}

	Let $ P_{-} \coloneqq I_{2^n} \tensor \ketbra{v} $ be the projection onto $\calH_{-}$.
	Then
	\begin{equation}
		(H_0)_{--} = P_-H_0P_-= \Delta \zeta \cdot P_-= \Delta \zeta \cdot I_{2^n}\otimes\ketbra v
	\end{equation}
	and
	\begin{equation}
		V_{--} = P_-VP_-=\left(\sum_{\tau \in \binary^\ell }v_\tau^2\cdot K_\tau\right)\otimes\ketbra v.
	\end{equation}

	Define
	\begin{align}
		H_\rmtarget \
		 & \coloneqq (H_0)_{--} + V_{--}
		= \left(  \Delta \zeta \cdot I_{2^n} + \sum_{\tau}v_\tau^2\cdot K_\tau \right) \tensor \ketbra{v}                                                       \\
		 & =\left(  \Delta \zeta \cdot I_{2^n} + v_{\bs a}^2K_a+v_{\bs b}^2K_b+\sum_{\tau\notin\{\bs a,\bs b\}}v_\tau^2\cdot K_\tau \right) \tensor \ketbra{v}.
	\end{align}
	Observe that the spectrum of $H_\rmtarget$ is the same to that of $\sum_{\tau}v_\tau^2\cdot K_\tau$ up to an additive $\Delta \zeta$ shift.
	In addition, since $K_{\bs{a}}\pathto K_{\bs{b}}$ is part of \Cref{eq:lem:pathC_K} and has minimal spectral gap $\delta$, the spectral gap of $\frac{v_{\bs{a}}^2K_{\bs{a}}+v_{\bs{b}}^2K_{\bs{b}}}{v_{\bs{a}}^2+v_{\bs{b}}^2}$ is at least $\delta$.
	By \Cref{fct:weyl}, the spectral gap of $H_\rmtarget$ is
	\begin{align}
		\delta_\rmtarget
		 & \ge\left(v_{\bs{a}}^2+v_{\bs{b}}^2\right)\cdot\delta
		-2\cdot\norm{\sum_{\tau \notin \{\bs{a},\bs{b}\}}v_\tau^2\cdot K_\tau}
		\\
		 & \ge\left(v_{\bs{a}}^2+v_{\bs{b}}^2\right)\cdot\delta-2\cdot(1-v_{\bs{a}}^2-v_{\bs{b}}^2)\cdot M
		\tag{since $\norm{K_\tau}\le M$}                                                                   \\
		 & \ge(1-O(\epsilon^2))\cdot\delta-O(\epsilon^2M)
		\tag{by \Cref{eq:lem:pathC_vtau}}                                                                  \\
		 & \ge0.99\cdot\delta,
		\label{eq:lem:pathC_deltatgt}
	\end{align}
	where we assume $\epsilon\le1/\poly(\nGHV)\le O(\sqrt{\delta/M})$ for the last inequality.

	Now, write $\calH = \calH_- \oplus \calH_+$ and regard $P_{-}$ as an encoding isometry from $\calH_{-}$ to $\calH$.
	Let $P_+$ be the projection onto $\calH_+$.
	Since $\calH_-$ is defined to be the ground space of $H_0$, the minimal eigenvalue of $(H_0)_{++} = P_+ H_0 P_+$ is the second smallest eigenvalue of $H_0$.
	Since $H_0=\Delta \cdot I\otimes \HL$ and by \Crefitem{lem:linearTFI}{itm:lem:linearTFI_2}, we know that eigenvalues of $(H_0)_{++}$ are greater than those of $(H_0)_{--}$ by at least $\Delta/2$.
	Hence by \Cref{lemma:1st}, ($\HC$,$P_-$) simulates $H_\rmtarget$ with error $1/2$, provided that $\Delta \geq \poly(M)$ as $\|V\| = \max_\tau \|K_\tau\| = M$.
	By \Cref{lem:gsim}, the spectral gap of $\HC$ is at least $\delta_\rmtarget - 1 = \Omega(\delta) = \Omega(\sqrt{\mGHV})$.
	This verifies the spectral gap property of \Cref{lem:pathC}.

	To check the norm condition of \Cref{lem:pathC}, we recall \Crefitem{lem:linearTFI}{itm:lem:linearTFI_1} that $\norm{\HL}\le\poly(\nGHV,1/\epsilon)$ as $t\in[0,\ell+1]$ and $\ell=\poly(\nGHV)$ from \Cref{lem:pathB2}.
	Hence
	\begin{equation}
		\norm{\HC} \leq \norm{H_0} + \norm{V} = \Delta \norm{\HL} + M \leq \poly(\nGHV,1/\epsilon)
	\end{equation}
	as claimed.
\end{proof}

\section{Reduction IV: Proof of Lemma \ref{lem:pathD}}\label{sec:reduction_IV}

We prove \Cref{lem:pathD} (restated below) in this section.

\lempathD*

The proof of \Cref{lem:pathD} is divided into two parts.
The first part handles the middle \apath\ $\HD(0)\pathto\HD(\ell+1)$, which is essentially \Cref{lem:pathC}.
The second part handles the start $\HD(-1/\epsilonD)\pathto\HD(0)$ and the end $\HD(\ell+1)\pathto\HD(1/\epsilonD)$, which relies on the first-order perturbative reduction (\Cref{lemma:1st}) similar to the proof of \Cref{lem:pathC}.

\paragraph{Notation.}
We (re)define some notation.
Let $t\in[-1/\epsilonD,1/\epsilonD]$ and let $n \coloneqq \nB$, $\Delta \coloneqq \DeltaC$, and $\epsilon \coloneqq \epsilonL$.
Recall from \Cref{eq:HD} that we have
\begin{equation} \label{eq:HD2}
	\HD(t) = \Delta I_{2^{n}} \tensor \HL(1-t\epsilonD,t) + (1- t\epsilonD) \XB \tensor I_{2^\ell} + (1+ t\epsilonD)\DC,
\end{equation}
where $\HL(1-t\epsilonD,t) =  \QL + (1- t \epsilonD)\XL + \ZL(t) $ as defined in \Crefitem{lem:linearTFI}{itm:lem:linearTFI_4}.

For $\tau \in \binary^n$, define
\begin{equation} \label{eq:Ktaut}
	K_\tau(t) \coloneqq
	\begin{cases}
		(1-t \epsilonD) \XB + (1+t\epsilonD) D_j
		 & \tau=\bs{j} \coloneqq 1^j0^{\ell-j} \text{ for some $j$}, \\
		(1-t \epsilonD) \XB
		 & \text{otherwise.}
	\end{cases}
\end{equation}
Note that we have
\begin{equation} \label{eq:D0Dl}
	D_0 = 0 \quad \text{and} \quad D_\ell = -3\DeltaB n \ketbra{u}{u}
\end{equation}
for some $u \in \binary^n$ and $\DeltaB \in [\mGHV,\poly(\mGHV)]$ by \Cref{lem:pathB} and \Cref{lem:pathB2}.
We also note the following easy observation.
\begin{claim} \label{clm:lem:pathD_1}
	The spectral gap of $K_{\bs{0}}(t)$ is at least $(1-t\epsilonD)\cdot \DeltaB n$. The spectral gap of $K_{\bs{\ell}}(t)$  is at least $(1 + 5t\epsilonD)\cdot \DeltaB n$.
\end{claim}
\begin{proof}
	Recall that $\XB=-\DeltaB\sum_{i\in[n]}X_i$. The statement follows directly from \Cref{eq:Ktaut} and \Cref{eq:D0Dl}.
\end{proof}

We inherit the value of $\Delta$ and the dependency $\DeltaL = \poly(\ell,1/\epsilon)$ from \Cref{lem:pathC}.
Therefore we only need to show that $\epsilon,\epsilonD\ge1/\poly(\nGHV)$ suffices for our purposes.

\begin{claim} \label{clm:lem:pathD_tmid}
	Assume $\epsilon,\epsilonD \leq 1/\poly(\nGHV)$. $\HD(t)$ has spectral gap $\Omega(\sqrt{\mGHV})$ for $t \in [0,\ell+1]$.
\end{claim}
\begin{proof}
	For $t \in [0,\ell+1]$, we have
	\begin{equation}
		\norm{\HD(t) - \HC(t)} = \norm{t\epsilonD \left( \DC - \Delta I_{2^n} \otimes \XL - \XB \otimes I_{2^\ell} \right) } \leq \epsilonD\cdot \poly(\nGHV)\le1,
	\end{equation}
	where we set $\epsilonD\le1/\poly(\nGHV)$ sufficiently small.
	Since $\epsilon\le1/\poly(\nGHV)$, we can apply \Cref{lem:pathC} to get an $\Omega(\sqrt{\mGHV})$ spectral gap lower bound for $\HC(t)$.
	Then the claim follows from \Cref{lem:gsim}.
\end{proof}

\begin{claim} \label{clm:lem:pathD_tout}
	Assume $\epsilon,\epsilonD \leq 1/\poly(\nGHV)$. $\HD(t)$ has spectral gap $\Omega(\sqrt{\mGHV})$ for $t \in [-1/\epsilonD,1/\epsilonD] \setminus [0,\ell+1]$.
\end{claim}
\begin{proof}
	Fix $t \in [-1/\epsilonD,1/\epsilonD] \setminus [0,\ell+1]$.
	Let $a \coloneqq 0$ if $t \leq 0$ and $a \coloneqq \ell$ if $t \geq \ell+1$.
	We shorthand $K_\tau\coloneqq K_\tau(t)$, $\HL \coloneqq \HL(1-t\epsilonD, t)$, $H_0 \coloneqq \Delta \cdot I_{2^n} \tensor \HL$, and $V \coloneqq (1-t \epsilonD) \XB \tensor I_{2^\ell} + (1 + t \epsilonD)\DC$.
	The Hamiltonian of interest $\HD:=\HD(t)$ in \Cref{eq:HD2} then becomes
	\begin{equation}
		\HD =  H_0 + V,
		\quad\text{where}\quad
		H_0 = \Delta \cdot I_{2^n} \tensor \HL,
		\quad V = \sum_{\tau \in \binary^\ell} K_\tau(t) \tensor \ketbra{\tau}{\tau}.
	\end{equation}
	Denote $M \coloneqq \max_{\tau \in \binary^\ell} \norm{K_\tau}$ and, by \Cref{lem:pathB} and \Cref{lem:pathB2}, we have
	\begin{equation} \label{eq:Mbound}
		M \leq(1 + t \epsilonD ) \poly(n) + 2 t \epsilonD \DeltaB n\leq \poly(\nGHV).
	\end{equation}

	Our analysis here is similar to the proof of \Cref{lem:pathC}.
	Let $\calH$ be the Hilbert space associated with $H_0$, $\HD$, and $V$.
	Let $\calH_{-}$ be the ground space of $H_0$.
	Let $\zeta$ be the ground energy of $\HL$ and $\ket{v}$ be a ground state of $\HL$.
	By definition, $\calH_{-}$ is spanned by $\ket{\psi} \tensor \ket{v}$ for any $n$-qubit state $\ket{\psi}$.
	We explicitly write out $\ket{v}=\sum_{\tau \in \binary^\ell} v_\tau \ket{\tau}$.
	Then by \Crefitem{lem:linearTFI}{itm:lem:linearTFI_4}, we have
	\begin{equation} \label{eq:lem:pathD_vtau}
		\sum_{\tau \neq \bs{a}} v_\tau^2 \leq O(\epsilon^2).
	\end{equation}

	Let $ P_{-} \coloneqq I_{2^n} \tensor \ketbra{v} $ be the projection onto $\calH_{-}$.
	Then
	\begin{equation}
		(H_0)_{--} = P_-H_0P_-= \Delta \zeta \cdot P_-= \Delta \zeta \cdot I_{2^n}\otimes\ketbra v
	\end{equation}
	and
	\begin{equation}
		V_{--} \coloneqq P_-VP_-=\left(\sum_{\tau \in \binary^\ell }v_\tau^2\cdot K_\tau\right)\otimes\ketbra v.
	\end{equation}

	Define
	\begin{align}
		H_\rmtarget
		 & \coloneqq (H_0)_{--} + V_{--}
		= \left(  \Delta \zeta \cdot I_{2^n} + \sum_{\tau}v_\tau^2\cdot K_\tau \right) \tensor \ketbra{v}                                 \\
		 & =\left(  \Delta \zeta \cdot I_{2^n} + v_{\bs a}^2K_{\bs a}+ \sum_{\tau\ne\bs a}v_\tau^2\cdot K_\tau \right) \tensor \ketbra{v}
	\end{align}
	Observe that the spectrum of $H_\rmtarget$ is the same to that of $\sum_{\tau}v_\tau^2\cdot K_\tau$ up to an additive $\Delta \zeta$ shift.
	In addition, $K_{\bs{a}}$ has minimal spectral gap at least $\delta \coloneqq \DeltaB n$ by \Cref{clm:lem:pathD_1}.
	By \Cref{fct:weyl}, the spectral gap of $H_\rmtarget$ is
	\begin{align}
		\delta_\rmtarget
		 & \ge v_{\bs{a}}^2 \cdot\delta
		-2\cdot\norm{\sum_{\tau \neq \bs{a}}v_\tau^2\cdot K_\tau(t)}
		\\
		 & \ge v_{\bs{a}}^2 \cdot\delta-2\cdot(1-v_{\bs{a}}^2)\cdot M
		\tag{since $\norm{K_\tau(t)}\le M$}                           \\
		 & \ge(1-O(\epsilon^2))\cdot\delta-O(\epsilon^2M)
		\tag{by \Cref{eq:lem:pathD_vtau}}                             \\
		 & \ge0.99\cdot\delta,
		\label{eq:lem:pathD_deltatgt}
	\end{align}
	where we assume $\epsilon\le1/\poly(\nGHV)\le O(\sqrt{\delta/M})$ for the last inequality.

	Now, write $\calH = \calH_- \oplus \calH_+$ and regard $P_{-}$ as an encoding isometry from $\calH_{-}$ to $\calH$.
	Let $P_+$ be the projection onto $\calH_+$.
	Since $\calH_-$ is defined to be the ground space of $H_0$, the minimal eigenvalue of $(H_0)_{++} = P_+ H_0 P_+$ is the second smallest eigenvalue of $H_0$.
	Since $H_0=\Delta \cdot I\otimes \HL$ and by \Crefitem{lem:linearTFI}{itm:lem:linearTFI_4}, we know that eigenvalues of $(H_0)_{++}$ are greater than those of $(H_0)_{--}$ by at least $\Delta/2$.
	Hence by \Cref{lemma:1st}, ($\HD$,$P_-$) simulates $H_\rmtarget$ with error $1/2$, provided that $\Delta \geq \poly(M)$ as $\|V\| = \max_\tau \|K_\tau\| = M$.
	Recall \Cref{eq:Mbound} that $M\le\poly(\nGHV)$. Hence it suffices to assume $\Delta=\poly(\nGHV)$.
	Finally by \Cref{lem:gsim}, the spectral gap of $\HD$ is at least $\delta_\rmtarget - 1 = \Omega(\delta) = \Omega(\sqrt{\mGHV})$.
	This completes the proof of \Cref{clm:lem:pathD_tout}.
\end{proof}

Now we wrap up the analysis of \Cref{lem:pathD}.

\begin{proof}[Proof of \Cref{lem:pathD}]
	The spectral gap condition follows directly from \Cref{clm:lem:pathD_tmid} and \Cref{clm:lem:pathD_tout}.

	To verify the norm condition, we recall \Cref{eq:HD} and perform a direct calculation
	\begin{align}
		\norm{\HD(t)}
		 & \leq\Delta\cdot\norm{\HL(1-t\epsilonD,t)} + |1-t\epsilonD|\cdot\norm{\XB}+|1+t\epsilonD|\cdot\norm{\DC}  \\
		 & =\poly(\nGHV,\norm{\HL(1-t\epsilonD,t)},\norm{\DC})
		\tag{since $\Delta,\norm{\XB}\le\poly(\nGHV)$}                                                              \\
		 & =\poly(\nGHV,\norm{\HL(1-t\epsilonD,t)})
		\tag{since $\norm{\DC}=\max_{0\le i\le\ell}\norm{D_i}\le\poly(\nGHV)$ by \Cref{eq:HC} and \Cref{lem:pathB}} \\
		 & =\poly(\nGHV),
	\end{align}
	where we use \Cref{lem:linearTFI} and $\ell=\poly(\nGHV),\epsilon\ge1/\poly(\nGHV)$ for the last inequality.
\end{proof}     { \renewcommand{\ket}[1]{|#1\rangle}
\renewcommand{\bra}[1]{\langle #1|}
\renewcommand{\ketbra}[2]{|#1\rangle\langle #2|}
\renewcommand{\braket}[2]{\langle#1|#2\rangle}
\renewcommand{\mel}[3]{\langle #1|#2|#3\rangle}

\section{Extension to Continuous Optimization: Theorem~\ref{thm:qhd_informal}} \label{sec:qhd}

In this section, we prove an exponential quantum-classical separation for continuous optimization. We consider the following box-constrained minimization problem
\begin{equation}
	\argmin_{x\in \calX} f(x), \tag{P}\label{tag:p}
\end{equation}
where we assume the objective function $f$ is Lipschitz continuous in a box-shaped domain $\calX = \{(x_1,\dots,x_n) \in \R^n \colon -J_j \le x_j \le J_j,~\forall j \in [n]\}$, where $J_j \leq \poly(n)$ for any $j\in [n]$.
We denote the global minimizer of $f$ in $\calX$ as $x^*$.

This section is organized as follows.
In \Cref{sec:qhd-1}, we present the main result of this section, \Cref{thm:qhd},
whose proof is detailed in \Cref{sec:proof-of-qhd}.
The proof is based on \Cref{lem:tosdg} that reduces TFD Hamiltonians (as in~\eqref{eqn:tfd-defn}) to a family of (unbounded) Hamiltonian operators defined over the box $\calX$.
Then, in \Cref{sec:proof-of-tosdg}, we prove \Cref{lem:tosdg} by leveraging the properties of the $\hatX$ operator introduced in \cite{zheng2024computational} and incorporating new insights.
Some missing proofs from \Cref{sec:proof-of-tosdg} are deferred to \Cref{sec:missing-proof-qhd} for better readability and flow.

\subsection{Hamiltonian Formulation of Continuous Optimization} \label{sec:qhd-1}

\begin{definition}[Hamiltonian path in real space]\label{defn:real-space-path}
	Let $n$ be a positive integer and $\calX$ is a compact set in $\R^n$.
	We call a 1-parameter family $\{H(t): t \ge 0\}$ a Hamiltonian path over $\calX$ if for each $t \ge 0$, we have
	\begin{equation}
		H(t) = - \Delta + V(t,x),
	\end{equation}
	where $\Delta = \sum_{j \in [n]}\frac{\partial^2}{\partial^2_{x_j}}$ is the Laplace operator and $V(t,x)\colon [0,\infty)\times \calX \to \R$ is a time-dependent potential function.
\end{definition}

For each fixed $t$, $H(t)$ is a Hamiltonian operator that acts on quantum states in $\calX$. In particular, for a smooth wave function $\Psi\colon \calX \to \mathbb{C}$ with unit $L^2$-norm (i.e., a test function), we have
\begin{equation}
	H(t) \Psi(x) = - \Delta \Psi(x) + V(t,x)\Psi(x).
\end{equation}
This is a natural generalization of the finite-dimensional Hamiltonian path (as in Section \ref{sec:ham-path}) to unbounded Hamiltonians.
Similar to the finite-dimensional cases, the continuous-space Hamiltonian path $H(t)$ can lead to a polynomial-time quantum algorithm to solve the minimization problem~\eqref{tag:p} if the following conditions holds: (1) the ground state of $H(0)$ can be efficiently prepared, (2) $H(t)$ admits a polynomial small spectral gap, (3) $\dot{V}(t,x)$ and $\ddot{V}(t,x)$ is polynomially bounded, (4) the minimizer of $f(x^*)$ is ``encoded'' in the ground state of $H(t)$ for some large $t$, and (5) the Hamiltonian $H(t)$ can be efficiently simulated.
When all the requirements are satisfied, we can implement an adiabatic-type quantum simulation for $H(t)$ and extract the information of $x^*$ by measuring the ground state of $H(t)$ for some large $t$. The details of this protocol and its computational complexity can be found in~Apppendix~\ref{sec:qat}.

In particular,~\Cref{defn:real-space-path} encompasses the quantum Hamiltonian descent algorithm as a special case. Let $f$ be an objective function over $\calX$. With an appropriate time rescaling, QHD is governed by the evolution of the following quantum Hamiltonian:
\begin{equation}
	H(t) = - \Delta + s(t) f(x),
\end{equation}
where $s(t)$ is certain time-dependent functions. In other words, QHD follows a Hamiltonian path over $\calX$ with $V(t,x) = s(t)f(x)$, which is separable in $t$ and $x$.

\begin{theorem}[Formal version of \Cref{thm:qhd_informal}] \label{thm:qhd}
	For any $n$, there exists a family of continuous functions $\calM_n = \{f\colon \calX \to \R\}$ with a fixed box domain $\calX \subset \R^n$ where the diameter of $\calX$ is at most $\poly(n)$. For any $f \in \calM_n$, we can construct a Hamiltonian path
	\begin{equation}\label{eq:ham-path-qhd}
		H(t) \coloneqq - \Delta + V(t,x),\quad V(t,x) \coloneqq \nu(t) g(x) + tf(x), \quad \forall t \ge 0, x \in \calX
	\end{equation}
	where the auxiliary function $g(x)$ and $\nu(t)$ are efficiently computable (see~\Cref{eq:defn-g}, \Cref{eq:nu-def}, and \Cref{eq:from-lambda-to-t}) and they only depend on $n$.
	Moreover, $g(x)$ is bounded within $\pm \poly(n)$, and $\nu(t) = e^{-\Theta(\sqrt{t})}$.
	The Hamiltonian path $\{H(t)\}_{t\ge 0}$ exhibits an exponential quantum advantage given oracle access to $f$:
	\begin{enumerate}[label=(\roman*)]
		\item\label{itm:thm:qhd_1} Any classical algorithm requires at least $\exp(n^{\Omega(1)})$ queries to find an $x$ such that
		$f(x) - f(x^*) \leq 1/\poly(n)$, with success probability greater than $\exp(-n^{\Omega(1)})$.
		\item\label{itm:thm:qhd_2}
		Let $\phi(t,x)$ be the instantaneous ground state of $H(t)$ for all $t \ge 0$. The following properties hold for $H(t)$:
		\begin{itemize}
			\item The initial ground state $\phi(0,x)$ is $n^{-\Omega(1)}$-close to a product state of the form $\bigotimes_{i \in [n]} \ket{\vartheta_i}$, with each $\ket{\vartheta_i}$ being efficiently preparable.
			\item For any $t \geq 0$, $H(t)$ has spectral gap $1/\poly(n)$, and $\| \dot{H}(t) \|, \| \ddot{H}(t) \| \leq \poly(n)$.
			\item For any $\eta>0$, there exists $\tend = \poly(n,1/\eta)$ such that $\langle \phi(\tend) |f|\phi(\tend)\rangle - f(x^\ast) \leq \eta/10$.
		\end{itemize}
		Consequently, simulating the \sdg\ dynamics for $0 \le t \le \tend$:
		\begin{equation}
			i \frac{\dee}{\dee t} \ket{\psi(t)}= T H(t) \ket{\psi(t)}, \quad \ket{\psi(0)} = \bigotimes_{i \in [n]} \ket{\vartheta_i}, \quad T = \poly(n),
		\end{equation}
		and measuring the final state $\ket{\psi(\tend)}$ using the position (i.e., computational) basis, we obtain an $x \in \calX$ with query and gate complexity $\poly(n,1/\eta)$ such that $f(x) - f(x^*) \le  \eta$ with probability at least $9/10 - n^{-\Omega(1)}$.
	\end{enumerate}
\end{theorem}
\begin{remark}
	We refer readers to \Cref{prop:merit_dtc} for more quantitative properties regarding $f(x)$ and $g(x)$.
\end{remark}

The proof of \Cref{thm:qhd} is established by constructing a Hamiltonian path from the time-dependent TFD Hamiltonian path
\begin{equation}\label{eq:qhd-pre-path}
	H_{\rm TFI} \pathto H_{\rm prob}
\end{equation}
in \Cref{thm:main}\footnote{Note that we use $H_{\rm prob}$ to refer the original diagonal Hamiltonian $D$ in \Cref{thm:main}. The symbol $D$ is reserved for a different purpose in this section.} via a key reduction lemma formulated in~\Cref{lem:tosdg}.
Here, $H_{\rm TFI}$ denotes a TFD Hamiltonian of the form:
\begin{equation} \label{eq:HTFI-def}
	H_{\rm TFI} = -\sum_i a_i X_i + D_{\rm TFI},
\end{equation}
and $H_{\rm prob}$ represents a diagonal problem Hamiltonian.
As in~\Cref{thm:main}, the ground state of $H_{\rm TFI}$ is a simple product state $\ket{+}^{\otimes (n-\ell)} \ket{0}^{\otimes \ell}$ for some $\ell$, up to a small error, and the linear Hamiltonian path \Cref{eq:qhd-pre-path} admits an $n^{\Omega(1)}$ spectral gap.

We now give an explicit construction of the objective function $f$. The proof of \Cref{thm:qhd} and the key reduction lemma, including the details on $g$ and $\nu$, will be provided in Section~\ref{sec:proof-of-qhd}.

\begin{definition}[Construction of $\hatD$] \label{def:hatD}
	First, we define a helper function $\theta \colon [-1,1]\to \{0,1\}$:
	$$\theta(\xi) \coloneqq \bs{1}_{\xi < 0},\quad \xi \in [-1,1],$$
	which can be naturally extended to multi-dimensional settings $\theta \colon [-1,1]^n \to \{0,1\}^n$:
	\begin{equation} \label{eq:theta-def}
		\theta(\xi) = \theta(\xi_1,\dots,\xi_n) \coloneqq (\theta(\xi_1),\dots,\theta(\xi_n)).
	\end{equation}
	Suppose that $D$ is an arbitrary $n$-qubit diagonal Hamiltonian. We define a function $\hatD\colon [-1,1]^n \to \R$ by
	\begin{equation}\label{eq:hatd-def}
		\hatD(\xi_1, \dots, \xi_n) \coloneqq \min\{1, \xi_{\rm min}/w\} \cdot  \mel{\theta(\xi)}{D}{\theta(\xi)},\quad \xi_{\rm min} \coloneqq \min_{i \in [n]} \{ \abs{\xi_i} \},
	\end{equation}
	where $w \ll 1$ is a universal constant that will be specified later (see~\Crefitem{clm:properties-of-X}{itm:clm:properties-of-X_3}).
	When the context is unambiguous, we will also use $\hatD$ to denote the diagonal operator defined by the function $\hatD(\xi)$, i.e., for any test function $\phi\colon [-1,1]^n\to \mathbb{R}$,
	$$(\hatD \phi)(\xi) = \hatD(\xi) \phi(\xi).$$
\end{definition}

We refer the readers to~\Cref{fig:D-heatmap} for an intuitive illustration of our construction of the continuous function $\hatD$.

\begin{fact} \label{fct:hatD}
	The function $\hatD$ satisfies the following properties:
	\begin{enumerate}[label=(\roman*)]
		\item\label{itm:fct:hatD_1} $\hatD$ is continuous in $[-1,1]^n$ and $O(\norm{D})$-Lipschitz.
		\item\label{itm:fct:hatD_2} The global minimum of $\hatD$ is $\min\{\eigen_0(D),0\}$, where $\mu_0(D)$ denotes the ground energy (i.e., minimal diagonal element) of $D$.
	\end{enumerate}
\end{fact}

\paragraph{Rescaling the path in \Cref{thm:main}.}
Without loss of generality, we assume that all diagonal elements of $H_{\rm prob}$ lie in $[-1,0)$ for the linear Hamiltonian path from \Cref{thm:main},
\begin{equation} \label{eq:modified-main-path}
	H_{\rm TFI} = -\sum_{i \in [n]} a_i X_i + D_{\rm TFI} \pathto H_{\rm prob},
\end{equation}
by an observation that the path $A \pathto \alpha B + \beta I$ has spectral gap at least $\alpha \delta$ if the path $A \pathto B$ has spectral gap $\delta$ for $\alpha \leq 1$.
This way, the spectral gap of \Cref{eq:modified-main-path} is at least $1/\poly(n)$ and all other claims from \Cref{thm:main} still hold.

\paragraph{Construction of the objective function $f$.}
We define
\begin{equation}\label{eq:qhd-obj-defn}
	f(\xi) = \hatH_{\rm prob} + \sum_{i \in [n]}a_i f_{\rm dw}(\xi_i),\quad f_{\rm dw}(z) = (z-1/2)^2(z+1/2)^2.
\end{equation}
Note that $\dwfunc(z)$ is a quartic double well function with global minima at $z = \pm 1/2$, as shown in \Cref{fig:dwspec}.
In other words, the objective function $f$ is essentially a continuous-space interpolation of the problem Hamiltonian $H_{\rm prob}$. Due to the double well function $f_{\rm dw}$, the local minima of $f$ are attained at the points $\{\pm 1/2\}^n$.
Since all the diagonal elements in $H_{\rm prob}$ are within $[-1,0)$, so by \Crefitem{fct:hatD}{itm:fct:hatD_2} applied on $\hatH_{\rm prob}$,
\begin{equation}
	\min_{\xi \in [-1,1]^n} f(\xi) = f(\xi^*) = \mu_0(H_{\rm prob}) < 0,\quad \xi^* = \argmin_{\xi \in [-1,1]^n}f(\xi).
\end{equation}
The negativity of the global minimum of $f(\xi)$ follows from the fact that all diagonal elements in $H_{\rm prob}$ are negative.

\begin{remark}
	The actual potential function $f(x)$ in~\Cref{eq:ham-path-qhd} is a (spatially) rescaled version of $f(\xi)$. The details can be found in the proof of~\Cref{thm:qhd}; cf.~\Cref{sec:proof-of-qhd}.
\end{remark}

\subsection{Proof of the Theorem} \label{sec:proof-of-qhd}

\paragraph{Notation and Convention.}
From here to the end of this section, we adopt the following notational conventions and assumptions.
Denote a Hamiltonian acting on position space by a hat, e.g., $\hatH$.
All wave functions introduced are \emph{real}-valued.\footnote{This assumption is valid because, for example, any \sdg\ operator is a real differential operator, and thus all its eigenfunctions can be chosen to be real.}
We will frequently use $\xi$ to denote an element in $\mathbb{R}$ or $\mathbb{R}^n$ while reserving $x$ for binary strings in $\binary^n$.
For a state $\ket{\phi}$, we denote its corresponding wave function by $\Psi_{\phi}(\cdot)$ or $\phi(\cdot)$.
Let $\eigen_j(\hatH)$ be the $(j+1)$th smallest eigenvalue of $\hatH$, counting multiplicity.\footnote{$\eigen_j(\cdot)$ is well defined for any self-adjoint operator that is bounded from below; see \cite[Lemma 4.26]{zheng2024computational}.}

First, we present the key reduction lemma.
We define an operator $\hatX$ in the interval $[-1,1]$ (with vanishing boundary condition):
\begin{equation} \label{eq:hatX-def}
	\hatX(\lambda) = - \frac{\dee^2 }{\dee \xi^2}  + \lambda^2 \dwfunc(\xi),
\end{equation}
parameterized by some $\lambda > 0$.
Let $\ket{\chi_0}$ and $\ket{\chi_1}$ be the ground and first-excited state of $\hatX$, respectively; and we assume $\chi_0(\xi),\chi_1(\xi) \geq 0$ for $\xi\geq 0$ by fixing their global phases; see \Cref{fig:dwspec_0}.
Let $\mathfrak{T} \coloneqq \spn(\ket{\chi_0},\ket{\chi_1})$ be the low-energy subspace of $\hatX$.
Let $\frakS \coloneqq \mathfrak{T}^{\tensor n}$ and we define an isometry $\calE \colon \mathbb{C}^{2^n} \to \frakS$:
\begin{equation} \label{eq:cale-def}
	\calE \coloneqq \left( \ketbra{\chi_0}{+} + \ketbra{\chi_1}{-}\right)^{\otimes n}, \quad \text{where } \ket{\pm} \coloneqq \frac{1}{\sqrt{2}}( \ket{0} \pm \ket{1} ).
\end{equation}

\begin{lemma} \label{lem:tosdg}
	Let $\{a_i \ge 0\}_i$ be a set of fixed coefficients and $D$ be an arbitrary $n$-qubit diagonal Hamiltonian.
	For any $\lambda > 0$, we define two operators:
	\begin{itemize}
		\item The operator $\hatH(\lambda)$ is a Hamiltonian defined over the box $[-1,1]^n \subset \R^n$:
		      \begin{equation}
			      \hatH(\lambda) \coloneqq \sum_{i \in [n]} a_i \hatX_i(\lambda) + \hatD, \quad \text{\normalfont where } \hatX_i \coloneqq - \frac{\dee^2 }{\dee \xi_i^2}  + \lambda^2 \dwfunc(\xi_i) \text{\normalfont \ for } i\in[n],
		      \end{equation}
		\item The operator $H(\lambda)$ is an $n$-qubit Hamiltonian:
		      \begin{equation} \label{eq:lem:tosdg_H}
			      H(\lambda) \coloneqq - \sum \frac{a_i}{ \Lambda(\lambda)} X_i + D, \quad \text{\normalfont where } \frac{\dee^k \Lambda}{\dee \lambda^k} = \left(e^{\lambda/6}\right)^{1\pm o_k(1)}  \text{\normalfont \ for } k \in \mathbb{N}.\footnote{$\Lambda(\lambda)$ is an efficiently computable function being half the inverse of the spectral gap of $\hatX$; see~\Cref{clm:properties-of-X}.}
		      \end{equation}
	\end{itemize}
	Then, there exists $\epsilon \leq O(\sqrt{n}\Lambda^{-1/6}\norm{D})$ such that the followings hold:
	\begin{enumerate}[label=(\roman*)]
		\item\label{itm:lem:tosdg_1} The ground energy of $\hatH(\lambda)$, denoted by $\eigen_0(\hatH)$, is no more than
		\begin{equation}
			\eigen_0(H) + \sum_{i \in [n]} a_i \cdot \Theta(\lambda) + \epsilon.
		\end{equation}
		\item\label{itm:lem:tosdg_2} The spectral gap of $\hatH(\lambda)$, denoted by $\delta \coloneqq \eigen_1(\hatH)-\eigen_0(\hatH)$, is at least
		\begin{equation}
			\delta \coloneqq \min\left\{\eigen_1(H), \min_{i \in [n]} \{ a_i\} \cdot \Theta(\lambda) + \eigen_0(D) \right\} - \eigen_0(H) - \epsilon.
		\end{equation}
		\item\label{itm:lem:tosdg_3} If $\delta > 0$, the ground state of $\hatH$ is $O(\epsilon/\delta)$-close to $\calE \ket{\psi_0}$, where $\ket{\psi_0}$ denotes the ground state of $H(\lambda)$.
	\end{enumerate}
\end{lemma}

Now, we are ready to prove~\Cref{thm:qhd}.
It is worth noting that, in the proof, we use $\hatH(t)$ to denote the Hamiltonian path in the real space to avoid confusion with~\Cref{lem:tosdg}, which is slightly different from the theorem statement.

\begin{proof}[Proof of~\Cref{thm:qhd}]
	\Cref{itm:thm:qhd_1} follows directly from our construction of the objective function $f$. Specifically, by the definition as in~\eqref{eq:qhd-obj-defn}, $f$ can be efficiently computed from a given linear Hamiltonian path $H(t)$, assuming oracle access to the $n$-qubit diagonal Hamiltonian $D$. Consequently, any classical query algorithm that minimizes $f$ would also effectively solve the discrete optimization problem defined by $D$, incurring at most a polynomial overhead. Thus, the classical query lower bound established in~\Crefitem{thm:main}{itm:thm:main_1} still applies.

	Next, we turn to \Cref{itm:thm:qhd_2}, showing the quantum advantage of the Hamiltonian path~\eqref{eq:ham-path-qhd}.
	Let $a_i$, $D_{\rm TFI}$, and $H_{\rm prob}$ be the same as in a chosen Hamiltonian path $H_{\rm TFI} \pathto H_{\rm prob}$.
	Let $\Lambda(\lambda)$ be the same as in~\Cref{lem:tosdg}, and $\lambda_0 > 0$ be a fixed number that will be specified later.
	We denote $\Lambda_0 = \Lambda(\lambda_0)$ and define
	\begin{align}
		g(\xi) \coloneqq \frac{1}{\Lambda_0} \hatD_{\rm TFI}(\xi) + \lambda^2_0 \sum^{n}_{i} a_i f_{\rm dw}(\xi_i),\label{eq:defn-g} \\
		\nu(\lambda) \coloneqq \frac{\Lambda_0}{\Lambda(\lambda)},\quad \varphi(\lambda) \coloneqq \lambda^2 - \nu(\lambda) \lambda_0^2. \label{eq:nu-def}
	\end{align}
	By~\Cref{lem:tosdg}, we have that $\nu(\lambda) = O(e^{-c\lambda})$ for some absolute constant $c> 0$. For all $\lambda \ge \lambda_0$, we have $\varphi(\lambda) \ge 0$ and $\varphi(\lambda) = \Theta(\lambda^2)$.

	As discussed earlier, the eigenvalues of $H_{\rm diag}$ from \Cref{eq:modified-main-path} lie in $[-1,0)$.
	Let $f(\xi)$ be the objective function described in~\eqref{eq:qhd-obj-defn}. We define the following Hamiltonian path for any $\lambda > \lambda_0$:
	\begin{align}
		\hatH(\lambda) & \coloneqq - \sum_i a_i \frac{\partial^2}{\partial^2_{\xi_i}} +  \nu(\lambda) g(\xi) + \varphi(\lambda) f(\xi)                 \\
		               & = - \sum_i a_i \hatX_i(\lambda) + \left(\frac{1}{\Lambda(\lambda)}\hatD_{\rm TFI} + \varphi(\lambda) \hatH_{\rm prob}\right).
	\end{align}
	According to~\Cref{lem:tosdg}, the spectral properties of the real-space Hamiltonian path $\hatH(\lambda)$ can be studied by investigating the following $n$-qubit Hamiltonian:
	\begin{align}
		H(\lambda) & \coloneqq - \frac{1}{\Lambda(\lambda)}\sum_i a_i X_i + \left(\frac{1}{\Lambda(\lambda)}D_{\rm TFI} + \varphi(\lambda) H_{\rm prob}\right) \\
		           & = \frac{1}{\Lambda(\lambda)} H_{\rm TFI} + \varphi(\lambda) H_{\rm prob}.
	\end{align}
	We note that the ground state of the initial Hamiltonian $H(\lambda_0)$ is equivalent to that of $H_{\rm TFI}$. By \Crefitem{thm:main}{itm:thm:main_2}, the ground state of $H_{\rm TFI}$ can be approximated by a product state up to an error $n^{-\Omega(1)}$. Therefore, by~\Crefitem{lem:tosdg}{itm:lem:tosdg_3}, the ground state of $\hatH(\lambda_0)$ can be efficiently prepared up to an error $n^{-\Omega(1)} + O(\epsilon/\delta) \leq n^{-\Omega(1)}$, where $\epsilon$ and $\delta=\delta(\lambda)$ are specified later in \Cref{eq:eps-and-delta}.

	By~\Cref{thm:main}, the Hamiltonian $H(\lambda)$ has a spectral gap:
	\begin{equation}
		\mu_1(H(\lambda)) - \mu_0(H(\lambda)) \ge \left(\varphi(\lambda) + \frac{1}{\Lambda(\lambda)}\right) \cdot \frac{1}{\poly(n)} = \Theta\left( \frac{\lambda^{2}}{\poly(n)} \right).
	\end{equation}
	For notational convenience, we denote $D(\lambda) \coloneqq \Lambda^{-1}(\lambda)D_{\rm TFI} + \varphi(\lambda) H_{\rm prob}$. By Weyl's inequality\footnote{This is the usual Weyl's inequality~\cite{wiki:weyl} generalized to semi-bounded self-adjoint operators; a proof is implicit in \cite[Lemma 4.27]{zheng2024computational}.}, we have
	\begin{equation}
		\mu_0(D(\lambda)) - \mu_0(H(\lambda)) \ge - \frac{1}{\Lambda(\lambda)}\sum_i a_i \ge - \frac{\poly(n)}{\Lambda(\lambda)}.
	\end{equation}
	Since $\Lambda(\lambda)$ decays exponentially fast, we can choose $\lambda_0 = \log(\poly(n))$ so that for $\lambda \ge \lambda_0$,
	\begin{equation}
		\min_{i}\{a_i\}\cdot \Theta(\lambda) + \mu_0(D) - \mu_0(H) \ge \Theta(\lambda),
	\end{equation}
	where we used the fact $\min_{i}\{a_i\} \geq 1$ by examining \Cref{eq:HD-path}, the explicit form of the path in \Cref{thm:main}.
	The spectral gap of $\hatH(t)$ can be deduced from the above results and~\Crefitem{lem:tosdg}{itm:lem:tosdg_2}:
	\begin{equation} \label{eq:eps-and-delta}
		\delta(\lambda) \ge \min\left(\frac{\varphi(\lambda) + \frac{1}{\Lambda(\lambda)}}{\poly(n)},\Theta(\lambda)\right) - \epsilon, \quad \epsilon \le O(\sqrt{n}\|D\|e^{-c \lambda/6}).
	\end{equation}
	Since $\|D\|$ is continuous in $\lambda$ and tends to infinity when $\lambda \to +\infty$, there exists a $\lambda' = \Theta(1)$ such that $\|D(\lambda)\| \ge 1$ for all $\lambda \ge \lambda'$. Therefore, we may choose
	\begin{equation} \label{eq:lambda0-bound}
		\lambda_0 = \max\left(1, \lambda', \log(\poly(n))\right),
	\end{equation}
	such that
	\begin{equation}
		\delta(\lambda) \ge \frac{1}{\poly(n)} \quad \text{and} \quad \epsilon \leq \frac{\delta(\lambda)}{n^{\Omega(1)}} ,\quad  \forall \lambda \ge \lambda_0.
	\end{equation}

	To recover the Hamiltonian path as described in~\Cref{thm:qhd}, we need to apply both spatial and temporal rescalings.
	First, we consider the change of (spatial variable):
	\begin{equation}
		z_i = \frac{1}{\sqrt{a_i}} \xi_i, \quad \forall i\in [n].
	\end{equation}
	Then, by the chain rule, $\frac{\partial}{\partial_{\xi_i}} = \frac{\partial z_i }{\partial {\xi_i}}\frac{\partial}{\partial_{z_i}} = {a}^{-1/2}_i \frac{\partial}{\partial_{z_i}}$. Now, we consider the rescaled box $\calX = \bigtimes^n_{i=1} [-a^{-1/2}_i, a^{-1/2}_i]$ and the new Hamiltonian operator:
	\begin{equation}
		\hatH(\lambda, z) = - \Delta + \nu(\lambda) g(\xi(z)) + \varphi(\lambda) f(\xi(z)), \quad z \in \calX.
	\end{equation}
	Since all eigenstates can be rescaled in the same way, this new Hamiltonian has a $1/\poly(n)$ spectral gap, as discussed before. Moreover, the diameter of the new box $\calX$ is at most $\poly(n)$.

	Next, we consider temporal scaling.
	Note that by $\lambda_0 \geq \log(\poly(n))$ from \Cref{eq:lambda0-bound}, the function $\varphi(\lambda)$ is  monotonically increasing in $\lambda$ with $\varphi(\lambda_0) = 0$.
	We can thus introduce a new variable:
	\begin{equation} \label{eq:from-lambda-to-t}
		t \coloneqq \varphi^{-1}(\lambda) \in [0, \infty],\quad \lambda \in [\lambda_0, \infty].
	\end{equation}
	This change of variable leads to the Hamiltonian path:
	\begin{equation}
		\hatH(t) = -\Delta + \nu(t) g(z) + t f(z),
	\end{equation}
	where $\nu(t) = \nu(\varphi^{-1}(\lambda))$. Since $\varphi(\lambda) = \Theta(\lambda^2)$, the function $\nu(t)$ decays at a (sub)exponential rate: $\nu(t) = e^{-\Theta(\sqrt{t})}$.
	Note that both $\dot{\nu}(t)$ and $\ddot{\nu}(t)$ decay at a (sub)exponential rate that is \emph{independent of} $n$ by \Cref{eq:nu-def} and \Cref{eq:lem:tosdg_H}.
	Moreover, $g$ and $f$ are bounded within $\poly(n)$ by construction.
	It follows that $\|\dot{\hatH}\|$ and $\|\ddot{\hatH}\|$ is also bounded within $\poly(n)$ for all $t \ge 0$.

	Therefore, by~\Cref{lem:qatsim2}, we can simulate the \sdg\ dynamics for $0 \le t \le \tend$ and obtain a final state $\ket{\psi(\tend)}$ that is $n^{-\Omega(1)}$-close to $\ket{\phi(\tend)}$.
	Define the projector $\hatP$ that projects a wave function onto the subspace supported on the region where $f(\xi)$ is $\eta$-close to the global minimum $f(\xi^*)$:
	\begin{equation}
		(\hatP \psi)(\xi) = \psi(\xi) \cdot \bs{1} \{ f(\xi) - f(\xi^*) \leq \eta\}.
	\end{equation}

	Now, we claim $\langle \phi(\tend) |f|\phi(\tend)\rangle - f(\xi^\ast) \le \eta/10$ for some $\tend = \poly(n, 1/\eta)$. Recall that $\lambda = \varphi(t)$. By the positivity of the operator $-\frac{\partial^2}{\partial^2_{\xi} }$ and~\Crefitem{lem:tosdg}{itm:lem:tosdg_1}, we have
	\begin{equation}
		\varphi(\lambda) \langle \phi(t) |f|\phi(t) \rangle \le \langle \phi(t) |\hatH(t)|\phi(t) \rangle \le \mu_0(H(\lambda)) + \sum_i a_i \cdot \Theta(\lambda) + \epsilon,
	\end{equation}
	where $\epsilon \le O(\sqrt{n}\|D\|e^{-c \lambda})$ is exponentially small in $\lambda$. Also, we have that $H(\lambda) = -\frac{1}{\Lambda(\lambda)} H_{\rm TFI} + \varphi(\lambda) H_{\rm prob}$, and $1/\Lambda(\lambda)$ is again exponentially small in $\lambda$. By Weyl's inequality,
	\begin{equation}
		\langle \phi(t) |f|\phi(t)\rangle \le \mu_0(H_{\rm prob}) + \sum_i a_i \cdot \Theta(\lambda^{-1}) + \epsilon',
	\end{equation}
	where $\epsilon'$ is exponentially small in $\lambda$. By construction, we have $\mu_0(H_{\rm prob}) = f(\xi^*)$; meanwhile, $\sum_i a_i \le \poly(n)$ by the norm condition of the path in \Crefitem{thm:main}{itm:thm:main_2}. Therefore, for any $\eta>0$, we can choose $\lambda' = \poly(n, 1/\eta)$ and $\tend = \varphi^{-1}(\lambda') = \poly(n, 1/\eta)$ such that
	\begin{equation} \label{eq:pre-markov}
		\langle \phi(\tend) |f|\phi(\tend) \rangle - f(\xi^*) \le \eta/10.
	\end{equation}
	Finally, by invoking Markov's inequality on \Cref{eq:pre-markov}, we have
	\begin{equation}
		\mel{\phi(\tend)}{\hatP}{\phi(\tend)} = \Pr_{\xi \sim|\phi(\tend)|^2}\left[f(\xi) - f(\xi^*) \le \eta \right] \geq 1 - \frac{1}{10} = 0.9.
	\end{equation}
	As a result,
	\begin{align*}
		 & \phantom{=}\ \Pr_{\xi \sim|\psi(\tend)|^2}\left[f(\xi) - f(\xi^*) \le \eta \right]                                                                                \\
		 & = \mel{\psi(\tend)}{\hatP}{\psi(\tend)}                                                                                                                           \\
		 & = \mel{\phi(\tend)}{\hatP}{\phi(\tend)} + \bra{\phi(\tend)}\hatP (\ket{\psi(\tend)}-\ket{\phi(\tend)}) + (\bra{\psi(\tend)} - \bra{\phi(\tend)}) \hatP \ket{\psi} \\
		 & \geq 0.9 - 2 \cdot \norm{ \ket{\psi(\tend)}-\ket{\phi(\tend)} } \cdot \norm{P}                                                                                    \\
		 & = 0.9 - n^{-\Omega(1)},
	\end{align*}
	as stated in \Crefitem{thm:qhd}{itm:thm:qhd_2}.
\end{proof}

\subsection{Simulating TFD Hamiltonians with \sdg\ Operators} \label{sec:proof-of-tosdg}

To prove \Cref{lem:tosdg}, we first introduce the required definitions and notation.

\paragraph{Effective computational basis.}
Recall that $\ket{\chi_0}$ and $\ket{\chi_1}$ are the ground state and the first excited state of $\hatX$ defined in \Cref{eq:hatX-def}.
We define the following states which serve as analogues of the standard qubit states $\ket{0}$ and $\ket{1}$ in the context of $\hatX$:
\begin{equation} \label{eq:hatzeroone-def}
	\ket{\hatzero} \coloneqq \frac{1}{\sqrt{2}} \left( \ket{\chi_0} + \ket{\chi_1} \right) \quad \text{and} \quad \ket{\hatone} \coloneqq \frac{1}{\sqrt{2}} \left( \ket{\chi_0} - \ket{\chi_1} \right).
\end{equation}
Moreover, for any $x = x_1x_2\dots x_n \in \binary^n$, define the effective computational state
\begin{equation}
	\ket{\hatx} \coloneqq \ket{\hatx_1} \tensor \ket{\hatx_2} \tensor \cdots \tensor \ket{\hatx_n}.
\end{equation}
Recall the definition of the encoding isometry $\calE$ from \Cref{eq:cale-def}.
We have $\calE = \left( \ketbra{\chi_0}{+} + \ketbra{\chi_1}{-}\right)^{\otimes n} = ( \ketbra{\hatzero}{0} + \ketbra{\hatone}{1} )^{\otimes n}$.
Therefore, for any $x \in \binary^n$ it holds that
\begin{equation}
	\calE \ket{x} = \ket{\hatx}.
\end{equation}
Let $\mathfrak{T} \coloneqq \spn(\ket{\chi_0},\ket{\chi_1}) = \spn(\ket{\hatzero},\ket{\hatone})$ be the low-energy subspace of $\hatX$, and let $\frakS \coloneqq \mathfrak{T}^{\tensor n}$ be the space spanned by $\{ \ket{\hatx} \}_{x \in \binary^n}$.

\paragraph{Restricted inner product.}
We use the following notation for the inner product restricted to the region $\calI$:
\begin{equation}
	\braket{\psi}{\phi}_{\calI} \coloneqq \int_{\calI} \psi(\xi)\phi(\xi) \dee \xi.
\end{equation}

We present a key claim about $\hatX$ based on \cite[Theorem 4.22]{zheng2024computational} and its proof.\footnote{Note that in \cite{zheng2024computational}, the parameters are given as $\lambda = 1/h$, $S_0 = 1/6$, and $a = \pm 1/2$. Here, $S_0$ denotes the \emph{Agmon distance} between $-a$ and $a$ with respect to the potential function $\dwfunc$ at energy $E = 0$.}

\begin{claim}[Properties of $\hatX$] \label{clm:properties-of-X}
	The following properties hold for $\hatX$:
	\begin{enumerate}[label=(\roman*)]
		\item\label{itm:clm:properties-of-X_1} $\eigen_0(\hatX) = \Theta(\lambda)$, $\eigen_1(\hatX) - \eigen_0(\hatX) = 2/\Lambda$, and $\eigen_2(\hatX) - \eigen_1(\hatX) = \Theta(\lambda)$, where $\Lambda = \Lambda(\lambda)$ satisfies \Cref{eq:lem:tosdg_H}.
		\item\label{itm:clm:properties-of-X_2} $\ket{\hatzero}$ and $\ket{\hatone}$ are symmetric under reflection, i.e., $\Psi_{\hatzero}(\xi) = \Psi_{\hatone}(-\xi)$.
		Consequently,
		\begin{equation} \label{eq:key-sym}
			\braket{\hatzero}{\hatone}_{[0,1]} = \braket{\hatzero}{\hatone}_{[-1,0]} = \frac12 \braket{\hatzero}{\hatone}=  0.
		\end{equation}
		\item\label{itm:clm:properties-of-X_3} There exists $w \in (0,1/2)$ such that the wave function $\Psi_{\hatzero}(\cdot)$ is concentrated within the interval $(w,1]$ in the sense that
		\begin{equation}
			\braket{\hatzero}{\hatzero}_{[-1,w]} \leq O(\Lambda^{-2/3}),
		\end{equation}
		Similarly, $\braket{\hatone}{\hatone}_{[-w,1]} \leq O(\Lambda^{-2/3})$ due to the symmetry from \Cref{itm:clm:properties-of-X_2}.
	\end{enumerate}
\end{claim}

Using \Cref{clm:properties-of-X} and the definition of $\hatD$ from \Cref{eq:hatd-def}, we establish the following claims about $\hatX$ and $\hatD$. Their proofs, which involve nontrivial technical arguments for $\hatD$, are deferred to \Cref{sec:missing-proof-qhd}.

\begin{claim}[$\hatX$ effectively implements $X$] \label{clm:estimates-of-X}
	Denote $\mineng \coloneqq \eigen_0(\hatX)$. \begin{enumerate}[label=(\roman*)]
		\item\label{itm:clm:estimates-of-X_1}
		$\hatX_i$ acts on subspace $\frakS$ the same way $- \Lambda^{-1} X_i$ acts on the $n$-qubit Hilbert space, modulo an overall energy shift $\mineng +\Lambda^{-1}$:
		\begin{equation}
			P_{\mathfrak{S}} \left( \hatX_i - \calE \left( - \Lambda^{-1} X_i + \left(\mineng + \Lambda^{-1}\right)I \right) \calE^\dagger \right)P_{\mathfrak{S}}= 0.
		\end{equation}
		\item\label{itm:clm:estimates-of-X_leak} $P_{\frakS^\perp} \hatX_i P_{\frakS} = 0$.
		\item\label{itm:clm:estimates-of-X_2} The minimum eigenvalue of $\hatX_i$ restricted on subspace $\frakS^\perp$ is $\mineng + \Theta(\lambda)$:
		\begin{equation}
			\eigen_0 \left(  P_{\frakS^\perp} \hatX_i P_{\frakS^\perp} \right) = \mineng + \Theta(\lambda).
		\end{equation}
	\end{enumerate}
\end{claim}

\begin{claim}[$\hatD$ effectively implements $D$] \label{clm:estimates-of-D}
	\ \begin{enumerate}[label=(\roman*)]
		\item\label{itm:clm:estimates-of-D_1} $\norm{P_{\frakS} \left( \hatD - \calE D \calE^\dagger \right)P_{\frakS}} \leq O(\sqrt{n} \Lambda^{-1/3} \norm{D} )$.
		\item\label{itm:clm:estimates-of-D_2} $\norm{P_{\frakS^\perp} \hatD P_{\frakS}} \leq O(\sqrt{n}\Lambda^{-1/6}\norm{D})$.
		\item\label{itm:clm:estimates-of-D_3} $\eigen_0 \left(  P_{\frakS^\perp} \hatD P_{\frakS^\perp} \right) \geq \eigen_0(\hatD) = \eigen_0(D)$.
	\end{enumerate}
\end{claim}

Now, express $\hatH$ in block form with respect to the decomposition $\frakS \oplus \frakS^\perp$ of the Hilbert space $\calW([-1,1]^n)$:
\begin{equation} \label{eq:hatH-block}
	\hatH = \left[
		\begin{array}{cc}
			P_{\frakS} \hatH P_{\frakS}       & P_{\frakS} \hatH P_{\frakS^\perp}       \\
			P_{\frakS^\perp} \hatH P_{\frakS} & P_{\frakS^\perp} \hatH P_{\frakS^\perp}
		\end{array}
		\right] \eqqcolon \left[
		\begin{array}{cc}
			\hatH_{\frakS} & R^\dagger            \\
			R              & \hatH_{\frakS^\perp}
		\end{array}
		\right].
\end{equation}
We are able to analyze $\hatH_{\frakS}$, $\hatH_{\frakS^\perp}$, and $R$, as follows.
\begin{itemize}
	\item {\bf Analysis of $\hatH_{\frakS}$.}
	      Define
	      \begin{equation}
		      \baseeng \coloneqq \sum_{i \in [n]} a_i \left(\eigen_0(\hatX) + \Lambda^{-1} \right) = \sum_{i \in [n]} a_i \cdot \Theta(\lambda),
	      \end{equation}
	      where the last equality follows from \Crefitem{clm:properties-of-X}{itm:clm:properties-of-X_1}.
	      Then, the following holds when restricted to the subspace $\frakS$:
	      \begin{align}
		       & \hspace{1.35em} \hatH_{\frakS} - \calE (H + \baseeng I) \calE^\dagger                                                                                                                                                                                                                                  \\
		       & =  \hatH - \calE ( H + \baseeng I ) \calE^\dagger                                                                                                                                                                                                                                                      \\
		       & =  \underbrace{\sum_{i \in [n]} a_i \hatX_i + \hatD}_{\hatH}  -\, \calE \left( \underbrace{-\sum_{i \in [n]} \frac{a_i}{\Lambda}X_i + D}_{H} +\, \baseeng I \right) \calE^\dagger                                                                                                                      \\
		       & =  \sum_{i \in [n]} a_i \hatX_i - \calE \left( -\sum_{i \in [n]} \frac{a_i}{\Lambda}X_i + \underbrace{\sum_{i \in [n]} a_i \left(\eigen_0(\hatX) +  \Lambda^{-1} \right)}_{\baseeng} I \right) \calE^\dagger + \left( \hatD - \calE D \calE^\dagger \right)                                            \\
		       & =  \sum_{i \in [n]}  a_i \left( \underbrace{\hatX_i - \calE \left( - \frac{1}{\Lambda} X_i +  \left( \eigen_0(\hatX) + \Lambda^{-1} \right) I \right) \calE^\dagger}_{=\,0 \text{ by \Crefitem{clm:estimates-of-X}{itm:clm:estimates-of-X_1}} } \right) + \left( \hatD - \calE D \calE^\dagger \right) \\
		       & = \hatD - \calE D \calE^\dagger.
	      \end{align}
	      Consequently, by \Crefitem{clm:estimates-of-D}{itm:clm:estimates-of-D_1},
	      \begin{equation} \label{eq:hatHS-est}
		      \norm{ \hatH_{\frakS} - \calE (H + \baseeng I) \calE^\dagger} = \norm{P_{\frakS} \left( \hatD - \calE D \calE^\dagger \right) P_{\frakS} } \leq O(\sqrt{n} \Lambda^{-1/3} \norm{D}).
	      \end{equation}

	\item {\bf Analysis of $\hatH_{\frakS^\perp}$.}
	      We now analyze the smallest eigenvalue of $\hatH_{\frakS^\perp}$. Recall that
	      $
		      \hatH = \sum_i a_i \hatX_i + \hatD.
	      $
	      By the property of the Kronecker sum\footnote{This is \Cref{fct:ksum} generalized to (unbounded) self-adjoint operators; see \cite[Corollary 7.25]{schmudgen2012unbounded}.} and \Crefitem{clm:estimates-of-X}{itm:clm:estimates-of-X_2}, the spectrum\footnote{We use the fact that the spectrum of the \sdg\ operator $\hatX$ is purely discrete; see \cite[Remark 18]{leng2025qhd}.} of $\sum_{i} a_i \hatX_i$ is given by
	      \begin{equation} \label{eq:decompose-sumaihatxi}
		      \sigma \Bigl(\sum_{i \in [n]} a_i \hatX_i\Bigr)
		      = \sigma(a_1\hatX) + \cdots + \sigma(a_n\hatX)
		      = \left\{ \sum_{i \in [n]} a_i \eigen_{k_i}(\hatX) \colon k_1, \dots, k_n \in \mathbb{N} \right\}.
	      \end{equation}
	      Since $\frakS = \mathfrak{T}^{\otimes n}$, where $\mathfrak{T}$ is spanned by the ground state $\ket{\chi_0}$ and the first excited state $\ket{\chi_1}$ of $\hatX$, any eigenstate of $\sum_{i} a_i \hatX_i$ not in $\frakS$ must correspond to at least one index $k_i\notin\{0,1\}$ in \Cref{eq:decompose-sumaihatxi}. Consequently, the minimum eigenvalue among those eigenstates is
	      \begin{equation}
		      \sum_{i \in [n]} a_i \eigen_0(\hatX)
		      + \min_{i \in [n]}\{a_i\} \Bigl( \eigen_2(\hatX) - \eigen_0(\hatX) \Bigr)
		      \ge \baseeng + \min_{i \in [n]}\{a_i\} \cdot \Theta(\lambda).
	      \end{equation}
	      By Weyl's inequality and \Crefitem{clm:estimates-of-D}{itm:clm:estimates-of-D_3}, the minimum eigenvalue of $\hatH$ restricted to $\frakS^\perp$ is at least
	      \begin{equation} \label{eq:hatHS2-est}
		      \eigen_0 \Bigl( P_{\frakS} \sum_{i \in [n]} a_i \hatX_i P_{\frakS} \Bigr) + \eigen_0(P_{\frakS}\hatD P_{\frakS}) \geq \baseeng + \min_{i \in [n]}\{a_i\} \cdot \Theta(\lambda) + \eigen_0(D).
	      \end{equation}

	\item {\bf Analysis of $R$.}
	      It is straightforward from \Crefitem{clm:estimates-of-X}{itm:clm:estimates-of-X_leak} and \Crefitem{clm:estimates-of-D}{itm:clm:estimates-of-D_2} that
	      \begin{equation} \label{eq:hatHS3-est}
		      \norm{R} = \norm{P_{\frakS^\perp} \left( \sum_{i \in [n]} a_i \hatX_i + \hatD \right) P_{\frakS}} = \norm{{\sum_{i \in [n]} a_i \underbrace{P_{\frakS^\perp}  \hatX_i P_\frakS}_{=\,0} + P_{\frakS^\perp} \hatD P_{\frakS}}} \leq O(\sqrt{n} \Lambda^{-1/6} \norm{D}).
	      \end{equation}
\end{itemize}

Summarizing \Cref{eq:hatHS-est}, \Cref{eq:hatHS2-est}, and \Cref{eq:hatHS3-est} we have the subsequent claim.

\begin{claim}[Bounds on $\hatH_{\frakS},\hatH_{\frakS^\perp},R$] \label{clm:tosdg-bound}
	There exists $\epsilon \leq O(\sqrt{n} \Lambda^{-1/6} \norm{D})$ such that the following holds:
	\begin{enumerate}[label=(\roman*)]
		\item\label{itm:clm:tosdg-bound_1} $\hatH_{\frakS}$ is $(\epsilon/2)$-close to $\calE (H + \baseeng I) \calE^\dagger$ in operator norm distance.
		\item\label{itm:clm:tosdg-bound_2} $\eigen_0(\hatH_{\frakS^\perp}) \geq \baseeng + \min_{i}\{a_i\} \cdot \Theta(\lambda) + \eigen_0(D)$.
		\item\label{itm:clm:tosdg-bound_3} $\norm{R} \leq \epsilon/4$.
	\end{enumerate}
\end{claim}

We are now in position to prove \Cref{lem:tosdg}.

\begin{proof}[Proof of {\Cref{lem:tosdg}}]
	Recall from \Cref{eq:hatH-block} that we write
	$
		\hatH =
		\begin{bmatrix}
			\hatH_{\frakS} & R^\dagger            \\
			R              & \hatH_{\frakS^\perp}
		\end{bmatrix}
	$.
	The ground energy of $\hatH$ is thus at most the smallest eigenvalue of $\hatH_{\frakS}$, which, by \Crefitem{clm:tosdg-bound}{itm:clm:tosdg-bound_1}, is in turn no more than
	\begin{equation}
		\eigen_0\left(\calE (H + \baseeng I) \calE^\dagger\right) + \epsilon = \eigen_0(H) + \baseeng + \epsilon = \eigen_0(H) + \sum_{i \in [n]}a_i \cdot \Theta(\lambda) + \epsilon,
	\end{equation}
	as stated in \Cref{itm:lem:tosdg_1}.

	For \Cref{itm:lem:tosdg_2}, observe that the spectral gap of $\hatH$ is at least
	\begin{align}
		 & \hspace{1.15em} \min \left\{ \eigen_1(\hatH_{\frakS}), \eigen_0(\hatH_{\frakS^\perp}) \right\} - \eigen_0(\hatH_{\frakS})                                                                                                                                                                                                                                                                                                                              \\
		 & \geq  \min \left\{ \underbrace{\eigen_1(H) + \baseeng - \epsilon/2}_{\text{by \Crefitem{clm:tosdg-bound}{itm:clm:tosdg-bound_1}}}, \underbrace{\baseeng + \min_{i \in [n]}\{a_i\} \cdot \Theta(\lambda) + \eigen_0(D)}_{\text{by \Crefitem{clm:tosdg-bound}{itm:clm:tosdg-bound_2}}} \right\} - \left( \underbrace{\eigen_0(H) + \baseeng + \epsilon/2}_{\text{by \Crefitem{clm:tosdg-bound}{itm:clm:tosdg-bound_1}}} \right)  \label{eq:weyl-is-used} \\
		 & \geq \min\left\{\eigen_1(H), \min_{i \in [n]} \{ a_i\} \cdot \Theta(\lambda) + \eigen_0(D) \right\} - \eigen_0(H) - \epsilon \eqqcolon \delta, \label{eq:tosdg-delta-bound}
	\end{align}
	where we used Weyl's inequality in \Cref{eq:weyl-is-used}.

	Regarding \Cref{itm:lem:tosdg_3}, denote by $\ket{\psi}$ the ground state of $\hatH$.
	We apply \Cref{fct:davis-kahan}\footnote{The Davis--Kahan $\sin \theta$ theorem works for unbounded self-adjoint operators; see~\cite[Section 6]{davis1970}.} with $d=1$ and
	\begin{equation}
		A =
		\begin{bmatrix}
			\hatH_{\frakS} & 0                    \\
			0              & \hatH_{\frakS^\perp}
		\end{bmatrix}, \quad B =
		\begin{bmatrix}
			0 & R^\dagger \\
			R & 0
		\end{bmatrix}
	\end{equation}
	to obtain
	\begin{equation}
		\norm{\calE \ket{g} - \ket{\psi}} \leq \frac{2^{3/2}\norm{B}}{\delta} = \frac{2^{3/2}\norm{R}}{\delta} \leq \epsilon/\delta, \tag{$\norm{R} \leq \epsilon$ by \Crefitem{clm:tosdg-bound}{itm:clm:tosdg-bound_3}}
	\end{equation}
	by noting that the ground state of $A$ coincides with that of $\hatH_{\frakS}$ when $\delta > 0$ since $\eigen_0(\hatH_{\frakS^\perp}) - \eigen_0(\hatH_{\frakS}) \geq \delta$ from \Cref{eq:tosdg-delta-bound}.
\end{proof}

\subsection{Deferred Proofs} \label{sec:missing-proof-qhd}

\begin{proof}[Proof of {\Cref{clm:estimates-of-X}}]
	Note that $\hatX_i$ is diagonal on the basis $\{ \ket{\chi_0},\ket{\chi_1} \}^{\otimes n} $ when restricted to subspace $\frakS$, since $\ket{\chi_0}$ and $\ket{\chi_1}$ are eigenstates of $\hatX$:
	\begin{equation}
		[\hatX ]_{ \{ \ket{\chi_0},\ket{\chi_1} \} } =
		\begin{bmatrix}
			\eigen_0(\hatX) & 0               \\
			0               & \eigen_1(\hatX)
		\end{bmatrix}
		\xlongequal{\text{\Crefitem{clm:properties-of-X}{itm:clm:properties-of-X_1}}}
		\begin{bmatrix}
			\mineng & 0                   \\
			0       & \mineng + 2/\Lambda
		\end{bmatrix}.
	\end{equation}
	Comparing it to
	\begin{equation}
		[-\Lambda^{-1}X]_{\{\ket{+},\ket{-}\}} =
		\begin{bmatrix}
			-\Lambda^{-1} & 0            \\
			0             & \Lambda^{-1}
		\end{bmatrix},
	\end{equation}
	we conclude the proof of \Cref{itm:clm:estimates-of-X_1}.
	The correctness of \Cref{itm:clm:estimates-of-X_leak} is obvious by diagonalizing $\hatX$, and \Cref{itm:clm:estimates-of-X_2} follows from \Crefitem{clm:properties-of-X}{itm:clm:properties-of-X_1} directly.
\end{proof}

To prepare for the proof of \Cref{clm:estimates-of-D}, we introduce some necessary notation and concepts.

\paragraph{Wave function truncation.}
In light of \Crefitem{clm:properties-of-X}{itm:clm:properties-of-X_3}, we define truncated states $\ket{\tildezero}$ and $\ket{\tildeone}$, supported only on the ``right'' and ``left'' halves of the domain $[-1,1]$, aligned with where $\ket{\hatzero}$ and $\ket{\hatone}$ are concentrated, respectively:
\begin{equation} \label{eq:tildezero-def}
	\Psi_{\tildezero}(\xi) = \left\{
	\begin{array}{ll}
		\Psi_{\hatzero}(\xi), & \xi \in (0,1],  \\
		0,                    & \xi \in [-1,0],
	\end{array}
	\right.
	\quad
	\Psi_{\tildeone}(\xi) = \left\{
	\begin{array}{ll}
		\Psi_{\hatone}(\xi), & \xi \in [-1,0), \\
		0,                   & \xi \in [0,1].
	\end{array}
	\right.
\end{equation}
Define $\ket{\remzero} \coloneqq \ket{\hatzero} - \ket{\tildezero}$ and $\ket{\remone} \coloneqq \ket{\hatone} - \ket{\tildeone}$.
\Crefitem{clm:properties-of-X}{itm:clm:properties-of-X_3} can then be relaxed and  translated to
\begin{equation} \label{eq:remzero-est}
	\braket{\remzero}{\remzero} = \braket{\remone}{\remone} \leq O(\Lambda^{-2/3}).
\end{equation}
Moreover, for any $x = x_1 x_2 \dots x_n \in \binary^n$, we define
\begin{equation} \label{eq:tildex}
	\ket{\tildex} \coloneqq \ket{\tildex_1} \tensor \ket{\tildex_2} \tensor \cdots \tensor \ket{\tildex_n}, \quad \ket{\remx} \coloneqq \ket{\hatx} - \ket{\tildex}.
\end{equation}

Note that $\ket{\hatx_i} = \ket{\tildex_i} + \ket{\remx_i}$ is an orthogonal decomposition.
Hence we have the following orthogonal decomposition:
\begin{align} \label{eq:hatx-orth}
	\underbrace{\ket{\hatx_1}\ket{\hatx_2}\cdots \ket{\hatx_n}}_{\ket{\hatx}} & = \underbrace{\sum_{i \in [n]} \ket{\tildex_1} \cdots \ket{\tildex_{i-1}} \ket{\remx_i} \ket{\hatx_{i+1}} \cdots \ket{\hatx_n}}_{\ket{\remx}} + \underbrace{\ket{\tildex_1} \ket{\tildex_2}\cdots \ket{\tildex_n}}_{\ket{\tildex}}.
\end{align}
Therefore, by the Pythagorean theorem,

\begin{align} \label{eq:remx-est}
	\braket{\remx}{\remx} \leq \sum_{i \in [n]}  \braket{\remx_i}{\remx_i} \leq O(n \Lambda^{-2/3}), \quad \braket{\tildex}{\tildex} = 1 - \braket{\remx}{\remx} \geq 1 - O(n \Lambda^{-2/3}).
\end{align}

\paragraph{Decomposition of \(\hatD\).}
Recall the definition of \(\theta(\cdot)\) from \Cref{eq:theta-def}.
To analyze \(\hatD\), we decompose it into a principal part and a remainder:
\begin{equation}
	\hatD = \hatDp + \hatDr,
\end{equation}
where the principal part \(\hatDp\) is defined by
\begin{equation} \label{eq:hatdp-def}
	\hatDp(\xi) = \hatDp(\xi_1,\dots,\xi_n) \coloneqq \mel{x}{D}{x}, \quad x \coloneqq \theta(\xi).
\end{equation}

We now examine \(\hatDp\).

\begin{claim}[Properties of $\hatDp$] \label{clm:properties-of-hatdp}
	The following properties hold for $\hatDp$.
	\begin{enumerate}[label=(\roman*)]
		\item\label{itm:clm:properties-of-hatdp_1} $\norm{P_{\frakS} \left( \hatDp - \calE D \calE^\dagger \right)P_{\frakS}} \leq O(\sqrt{n} \Lambda^{-1/3} \norm{D} )$.
		\item\label{itm:clm:properties-of-hatdp_2} $\norm{P_{\frakS^\perp} \hatDp P_{\frakS}} \leq O(\sqrt{n}\Lambda^{-1/6}\norm{D})$.
	\end{enumerate}
\end{claim}
\begin{proof}
	Define intervals $\calI_0 \coloneqq (0,1]$ and $\calI_1 \coloneqq [-1,0)$.
	Note that the supports of wave functions $\Psi_{\tildezero}(\cdot)$ and $\Psi_{\tildeone}(\cdot)$ are $\calI_0$ and $\calI_1$, respectively.
	We further define
	\begin{equation}
		\calI(x) \coloneqq \calI_{x_1} \times \calI_{x_2} \times \cdots \times \calI_{x_n}.
	\end{equation}
	It is immediate that the support of wave function $\Psi_{\tildex}(\cdot)$ is $\calI(x)$.
	Also, by definition \Cref{eq:hatdp-def}, we have
	\begin{equation} \label{eq:hatdpbeingc}
		\hatDp(\xi) = \mel{x}{D}{x} \quad  \text{for} \quad  \xi \in \calI(x).
	\end{equation}
	Therefore,
	\begin{equation} \label{eq:txhatdptx}
		\mel{\tildex}{\hatDp}{\tildex} = \mel{x}{D}{x} \braket{\tildex}{\tildex}.
	\end{equation}
	For any $x \in \binary^n$, we have
	\begin{align}
		\mel{\hatx}{\hatDp}{\hatx} & = \underbrace{\mel{\tildex}{\hatDp}{\tildex} + \mel{\tildex}{\hatDp}{\remx}}_{\mel{\tildex}{\hatDp}{\hatx}} + \mel{\remx}{\hatDp}{\hatx} \tag{by \Cref{eq:tildex}} \\
		                           & = \mel{x}{D}{x} \braket{\tildex}{\tildex} + \mel{\tildex}{\hatDp}{\remx}+ \mel{\remx}{\hatDp}{\hatx} \tag{by \Cref{eq:txhatdptx}}                                  \\
		                           & \leq \mel{x}{D}{x} \braket{\tildex}{\tildex} - 2 \cdot \norm{\ket{\remx}} \cdot \norm{D} \tag{$\|\hatDp\| = \norm{D}$ by \Cref{eq:hatdp-def}}                      \\
		                           & \leq \mel{x}{D}{x} (1 - O(n \Lambda^{-2/3})) - \norm{D} \cdot O(\sqrt{n} \Lambda^{-1/3}) \tag{by \Cref{eq:remx-est}}                                               \\
		                           & = \mel{x}{D}{x} - \norm{D} \cdot O(\sqrt{n} \Lambda^{-1/3}) \tag{$\mel{x}{D}{x} \leq \norm{D}$}.
	\end{align}
	On the other hand, for any $x \neq y$, assume that $x_j \neq y_j$ for some $j$.
	We have for any $z \in \binary^n$ that
	\begin{equation} \label{eq:orthant-orth}
		\braket{\hatx}{\haty}_{\calI(z)} = \prod_{i \in [n]} \braket{\hatx_i}{\haty_i}_{\calI_{z_i}} = 0,
	\end{equation}
	since $\braket{\hatx_j}{\haty_j}_{\calI_{z_j}}=0$ by \Crefitem{clm:properties-of-X}{itm:clm:properties-of-X_2}.
	Therefore,
	\begin{align}
		\mel{\hatx}{\hatDp}{\haty} & = \sum_{z \in \binary^n} \mel{\hatx}{\hatDp}{\haty}_{\calI(z)}                                         \\
		                           & = \sum_{z \in \binary^n} \mel{z}{D}{z} \braket{\hatx}{\haty}_{\calI(z)} \tag{by \Cref{eq:hatdpbeingc}} \\
		                           & = \sum_{z \in \binary^n} \mel{z}{D}{z} \cdot 0 = 0. \tag{by \Cref{eq:orthant-orth}}
	\end{align}
	\Cref{itm:clm:properties-of-hatdp_1} thus follows from our estimates of $\mel{\hatx}{\hatDp}{\hatx}$ and $\mel{\hatx}{\hatDp}{\haty}$ above.

	For \Cref{itm:clm:properties-of-hatdp_2}, let $\ket{\psi} = \sum_{x \in \binary^n} \alpha_x \ket{\hatx}$ be an arbitrary state in $\frakS$.
	We first estimate $\norm{\hatDp \ket{\psi}}$:
	\begin{align}
		\norm{\hatDp \ket{\psi}}^2 & = \sum_{z \in \binary^n} \mel{z}{D}{z}^2 \braket{\psi}{\psi}_{\calI(z)} \tag{by \Cref{eq:hatdpbeingc}}                                                             \\
		                           & = \sum_z \mel{z}{D}{z}^2 \sum_x \alpha_x^2 \braket{\hatx}{\hatx}_{\calI(z)} \tag{by \Cref{eq:orthant-orth}}                                                        \\
		                           & = \sum_z \mel{z}{D}{z}^2 \alpha_z^2 \braket{\hatz}{\hatz}_{\calI(z)} + \sum_z \mel{z}{D}{z}^2 \sum_{x \colon x \neq z} \alpha_x^2 \braket{\hatx}{\hatx}_{\calI(z)} \\
		                           & \leq \sum_z \mel{z}{D}{z}^2 \alpha_z^2 + \norm{D}^2 \sum_x \alpha_x^2 \sum_{z \colon z \neq x} \braket{\hatx}{\hatx}_{\calI(z)}                                    \\
		                           & = \sum_z \mel{z}{D}{z}^2 \alpha_z^2 + \norm{D}^2  \sum_x \alpha_x^2  \braket{\remx}{\remx} \label{eq:remx-trick}                                                   \\
		                           & = \norm{D \calE^\dagger \ket{\psi}}^2 + \norm{D}^2 \cdot O(n \Lambda^{-2/3}), \label{eq:hatdppsi-est}
	\end{align}
	where the last (in)equality follows from \Cref{eq:remx-est}, and the equality $\sum_{z \colon z \neq x} \braket{\hatx}{\hatx}_{\calI(z)} \leq \braket{\remx}{\remx}$ used in \Cref{eq:remx-trick} comes from the fact that the disjoint union of $\calI(z)$ is exactly the support of $\Psi_{\remx}(\cdot)$, which is in turn simply $\Psi_{\hatx}(\cdot)$ truncated.

	For notational convenience, define
	\begin{equation} \label{eq:rs-def}
		r \coloneqq \norm{D \calE^\dagger \ket{\psi}} = \norm{P_{\frakS} \calE D \calE^\dagger P_{\frakS} \ket{\psi}} \leq \norm{D} \quad \text{and} \quad s \coloneqq \sqrt{n} \Lambda^{-1/3} \norm{D}.
	\end{equation}
	Finally, by the Pythagorean theorem,
	\begin{align}
		\norm{P_{\frakS^\perp} \hatDp \ket{\psi}}^2 & =\norm{\hatDp \ket{\psi}}^2 - \underbrace{\norm{P_{\frakS} \hatDp \ket{\psi}}^2}_{= \norm{P_{\frakS}  \hatDp P_{\frakS} \ket{\psi}}^2}                                                             \\
		                                            & = \norm{\hatDp \ket{\psi}}^2 - \left( \norm{P_{\frakS} \calE D \calE^\dagger P_{\frakS} \ket{\psi}} - \norm{P_{\frakS} \left( \hatDp - \calE D \calE^\dagger \right)P_{\frakS} \ket{\psi}} \right) \\
		                                            & \leq \left( r^2 + O(s^2) \right) - \left( r - O(s) \right)^2 \tag{by \Cref{eq:hatdppsi-est}, \Cref{eq:rs-def}, and \Cref{itm:clm:properties-of-hatdp_1}}                                           \\
		                                            & = O(rs + s^2)                                                                                                                                                                                      \\
		                                            & \leq O(\norm{D}\cdot \sqrt{n}\Lambda^{-1/3}\norm{D} + n\Lambda^{-2/3}\norm{D}^2) \tag{by \Cref{eq:rs-def}}                                                                                         \\
		                                            & \leq O(n \Lambda^{-1/3} \norm{D}^2),
	\end{align}
	as claimed in \Cref{itm:clm:properties-of-hatdp_2}.
\end{proof}

Next, we examine \(\hatDr\).
Since $\hatDr = \hatD - \hatDp$, by definitions \Cref{eq:hatd-def} and \Cref{eq:hatdp-def} we have
\begin{equation} \label{eq:hatdr-def}
	\hatDr(\xi) = \hatDr(\xi_1,\dots,\xi_n) \coloneqq \min\{0, \xi_{\rm min}/w-1\} \cdot  \mel{x}{D}{x},\quad \xi_{\rm min} \coloneqq \min_{i \in [n]} \{ \abs{\xi_i} \}, \quad x \coloneqq \theta(\xi).
\end{equation}

\begin{claim}[Property of $\hatDr$] \label{clm:properties-of-hatdr}
	$\norm{\hatDr P_{\frakS}} \leq O(\sqrt{n}\Lambda^{-1/3} \norm{D})$.
\end{claim}
\begin{proof}
	Define $\calJ \coloneqq \calJ_1 \cup \calJ_2 \cup \cdots \cup \calJ_n$, where
	\begin{equation}
		\calJ_i \coloneqq [-1,1]^{i-1} \times [-w,w] \times [-1,1]^{n-i}.
	\end{equation}
	It is immediate from \Cref{eq:hatdr-def} that $\|\hatDr\| = \norm{D}$ and $\hatDr(\xi)=0$ for $\xi \in [-1,1]^n \setminus \calJ$.
	Therefore, for \emph{any} $\ket{\psi} \in \calW([-1,1]^n)$ we have
	\begin{equation} \label{eq:property-of-hatdr-1}
		\norm{\hatDr \ket{\psi}}^2 \leq \norm{\hatDr^2} \braket{\psi}{\psi}_{\calJ} = \norm{D}^2 \braket{\psi}{\psi}_{\calJ} \leq \norm{D}^2 \sum_{i \in [n]} \braket{\psi}{\psi}_{\calJ_i}.
	\end{equation}

	Now, assume $\ket{\psi} \in \frakS$ so that we can write $\ket{\psi} = \alpha \ket{\hatzero}\ket{\psi_0} + \beta \ket{\hatone}\ket{\psi_1}$ where $\alpha^2 + \beta^2 = 1$.
	We have
	\begin{align}
		\frac{1}{2} \braket{\psi}{\psi}_{\calJ_1} & \leq \alpha^2 \braket{\hatzero,\psi_0}{\hatzero,\psi_0}_{\calJ_1} + \beta^2 \braket{\hatone,\psi_1}{\hatone,\psi_1}_{\calJ_1} \tag{triangle inequality with $\braket{\cdot}{\cdot}_{\calJ_1}$} \\
		                                          & = \alpha^2 \braket{\hatzero}{\hatzero}_{[-w,w]} \braket{\psi_0}{\psi_0} + \beta^2 \braket{\hatone}{\hatone}_{[-w,w]} \braket{\psi_1}{\psi_1}                                                   \\
		                                          & \leq \alpha^2 \braket{\hatzero}{\hatzero}_{[-1,w]} \cdot 1 + \beta^2 \braket{\hatone}{\hatone}_{[-w,1]} \cdot 1                                                                                \\
		                                          & \leq O(\Lambda^{-2/3}). \tag{by \Crefitem{clm:properties-of-X}{itm:clm:properties-of-X_3}}
	\end{align}
	Similar estimates work for every $\braket{\psi}{\psi}_{\calJ_i}$ in \Cref{eq:property-of-hatdr-1}.
	Therefore, we have
	\begin{equation}
		\mel{\psi}{\hatDr^2}{\psi} \leq \norm{D}^2 \cdot n \cdot O(\Lambda^{-2/3}) = O(n \Lambda^{-2/3} \norm{D}^2),
	\end{equation}
	as desired.
\end{proof}

With \Cref{clm:properties-of-hatdp} and \ref{clm:properties-of-hatdr}, we can now prove \Cref{clm:estimates-of-D}.

\begin{proof}[Proof of {\Cref{clm:estimates-of-D}}]
	For \Cref{itm:clm:estimates-of-D_1}, when restricted on subspace $\frakS$, we have
	\begin{equation}
		\norm{ \hatD - \calE D \calE^\dagger } = \norm{\hatDp + \hatDr - \calE D \calE^\dagger} \leq \norm{\hatDp - \calE D \calE^\dagger} + \norm{\hatDr}  \leq O(\sqrt{n} \Lambda^{-1/3} \norm{D} ),
	\end{equation}
	where the final bound uses \Crefitem{clm:properties-of-hatdp}{itm:clm:properties-of-hatdp_1} and \Cref{clm:properties-of-hatdr}.

	Similarly, for \Cref{itm:clm:estimates-of-D_2} we have
	\begin{equation}
		\norm{P_{\frakS^\perp} \hatD P_{\frakS}} \leq  \norm{P_{\frakS^\perp} \hatDp P_{\frakS}} + \underbrace{\norm{P_{\frakS^\perp} \hatDr P_{\frakS}}}_{ \leq \norm{ \hatDr P_{\frakS}} }  \leq  O(\sqrt{n} \Lambda^{-1/6} \norm{D} ),
	\end{equation}
	where the last inequality follows from \Crefitem{clm:properties-of-hatdp}{itm:clm:properties-of-hatdp_2} and \Cref{clm:properties-of-hatdr}.

	Now we consider the statement $\eigen_0 \left(  P_{\frakS^\perp} \hatD P_{\frakS^\perp} \right) \geq \eigen_0(\hatD) = \eigen_0(D)$ in \Cref{itm:clm:estimates-of-D_3}.
	The inequality follows from the variational principle: projecting to a subspace can only increase the ground energy.
	The equality uses the assumption in \Cref{lem:tosdg} that \(\eigen_0(D) \leq 0\), together with \Crefitem{fct:hatD}{itm:fct:hatD_2}.
\end{proof}

}

\printbibliography

\appendix
\section{Adiabatic Quantum Computing} \label{sec:qat}

In \Crefitem{thm:main}{itm:thm:main_2} and \Crefitem{thm:qhd}{itm:thm:qhd_2}, we used the following lemmas:
\begin{lemma} \label{lem:qatsim}
	Let \( H(t) \) be an explicitly \( s \)-sparse \apath\ for \( t \in [0,1] \). Suppose that $\|H(t)\|$, \( \|\dot{H}(t)\| \), and \( \|\ddot{H}(t)\| \) are all upper bounded by \( M \) and that the spectral gap of \( H(t) \) is at least \( \delta \) for all \( t \). Furthermore, assume that the ground state of \( H(0) \) is \(\eta\)-close to an efficiently preparable state \(\ket{\phi}\). Then, there is a quantum algorithm, given oracle access to $H(t)$, that simulates the \sdg\  dynamics
	\begin{equation} \label{eq:a1-sdg}
		i \frac{\dee}{\dee t} \ket{\psi(t)} = T H(t) \ket{\psi(t)}, \quad \ket{\psi(0)} = \ket{\phi}, \quad T = \Theta\left(\frac{M^2}{\delta^3 \eta}\right),
	\end{equation}
	producing a final state \(\ket{\psi(1)}\) that is \(2\eta\)-close to the ground state of \( H(1) \) with query and gate complexity
	$
		(n s T M)^{1+o(1)}
	$.
\end{lemma}

\begin{lemma} \label{lem:qatsim2}
	Let \( H(t) = -\Delta + V(t,x) \) be a \apath\ for \( t \in [0,\tend] \). Suppose that $ |V(t,x)|$, \( \|\dot{H}(t)\| \leq \max_x | \dot{V}(t,x) | \), and \( \|\ddot{H}(t)\| \leq \max_x | \ddot{V}(t,x) | \) are all upper bounded by \( M \) and that the spectral gap of \( H(t) \) is at least \( \delta \) for all \( t \). Furthermore, assume that the ground state of \( H(0) \) is \(\eta\)-close to an efficiently preparable state \(\ket{\phi}\). Then, there is a quantum algorithm, given oracle access to $V(t,x)$, that simulates the \sdg\  dynamics
	\begin{equation} \label{eq:a2-sdg}
		i \frac{\dee}{\dee t} \ket{\psi(t)} = T H(t) \ket{\psi(t)}, \quad \ket{\psi(0)} = \ket{\phi}, \quad T = \Theta\left(\frac{M^2}{\delta^3 \eta}\right),
	\end{equation}
	producing a final state \(\ket{\psi(\tend)}\) that is \(2\eta\)-close to the ground state of \( H(\tend) \) with query and gate complexity
	$
		(n T M)^{1+o(1)}
	$.
\end{lemma}

The correctness of \Cref{lem:qatsim} directly follows from combining the \emph{quantum adiabatic theorem} (\Cref{thm:adiabatic}) and the rescaled Dyson-series algorithm~\cite{berry2020time} that simulates \Cref{eq:a1-sdg} with query and gate complexity \( (nsTM)^{1+o(1)}\).

Similarly, The correctness of \Cref{lem:qatsim2} follows from \Cref{thm:adiabatic} and an algorithm~\cite[Theorem 3]{Childs2022} that simulates \Cref{eq:a2-sdg} with query and gate complexity \( (nTM)^{1+o(1)}\).

\subsection{Quantum Adiabatic Theorem}

Given a time-dependent Hamiltonian \( H(t) \) for \( t \in [0,1] \), denote by \( \ket{\phi(t)} \) the ground state of \( H(t) \).
The \emph{quantum adiabatic theorems} provide conditions under which the evolution remains close to \( \ket{\phi(t)} \).
Specifically, consider the following Schr\"odinger equation:
\begin{equation} \label{eq:adiabatic}
	i \frac{\dee}{\dee t} \ket{\psi(t)} = T H(t) \ket{\psi(t)}, \quad t \in [0,1],
\end{equation}
with an initial state \( \ket{\psi(0)} \coloneqq \ket{\phi(0)} \) and a sufficiently large evolution time \( T \).
This dynamics is considered adiabatic because it can be rewritten via the change of variable \( s \coloneqq tT \), often referred to as a \emph{time dilation} argument:
\begin{equation}
	i \frac{\dee}{\dee s} \ket{{\psi}(s/T)} = \tilde{H}(s) \ket{{\psi}(s/T)}, \quad s \in [0,T],
\end{equation}
where \( \tilde{H}(s) \coloneqq H(s/T) \). In this formulation, \( \tilde{H}(s) \) varies infinitesimally slowly as \( T \to \infty \), ensuring adiabatic evolution.
Below, we introduce a version of the theorem that establishes a sufficient bound on \( T \) based on the operator norm of \( \dot{H}(t) \) and \( \ddot{H}(t) \) and the spectral gap of \( H(t) \), following a standard operator-theoretic approach.
For an elementary proof with slightly weaker bounds but offering more intuition, we refer the reader to~\cite{ambainis2004elementary}.

\begin{theorem}[\cite{Teufel2003,jansen2007bounds,childs2025lecture}] \label{thm:adiabatic}
	Suppose \( H(t) \) is piecewise twice differentiable and has a nondegenerate ground state for all \( t \in [0,1] \).
	Furthermore, \( \|\dot{H}(t)\| \) and \( \|\ddot{H}(t)\| \) are upper bounded by $M$ and the spectral gap of $H(t)$ is at least $\delta$ for all $t$.
	Then \Cref{eq:adiabatic} produces a final state \( \ket{\psi(1)} \) satisfying
	\begin{equation}
		\|\ket{\psi(1)} - \ket{\phi(1)}\| \leq \epsilon,
	\end{equation}
	provided that
	\begin{equation}
		T \geq \frac{1}{\epsilon} \cdot  O\left( \frac{M^2}{\delta^3} \right).
	\end{equation}
\end{theorem}
\begin{remark}
	Note that $H(t)$ itself can be unbounded in \Cref{thm:adiabatic}.
\end{remark}

\section{Explicit Technical Building Blocks}\label{sec:intro_merit}

We present clean statements of the main technical building blocks employed impiclitly in the proofs of \Cref{thm:main}.
These results have been modified to meet the specific requirements of our approach and are therefore (re)stated here for clarity and potential reuse.

\medskip

The following proposition reduces any stoquastic sparse Hamiltonian to a TFD Hamiltonian.
Its proof is implicit in those of \Cref{lem:pathA} and \ref{lem:pathB}.

\begin{proposition} \label{prop:merit_sparse}
	Let \( n, s, M \) be sufficiently large, and let \( \eta, \epsilon > 0 \) be sufficiently small.
	Then, there exist some $a>0$ and a mapping that transforms any \( n \)-qubit $s$-sparse stoquastic Hamiltonian \( H \) satisfying \( \|H\| \leq M \) into a \((6n+4s)\)-qubit stoquastic TFD Hamiltonian
	\begin{equation}
		K = -a \sum_{i\in[6n+4s]} X_i + D
	\end{equation}
	with $D$ diagonal, such that the following holds.

	\begin{enumerate}[label=(\roman*)]
		\item $(K,\calE)$ simulates $H$ with error $(\eta,\epsilon)$, where $\calE \colon \ket{x} \mapsto \ket{x,x,0^s,0^s,0^n,x,x,0^s,0^s,0^n}$ is the encoding isometry.
		\item $\norm{K} \leq \poly(1/\eta,1/\epsilon,n,s,M)$.
		\item An oracle for \( K \) can be efficiently implemented: a query to this oracle can be simulated using \( O(ns + s^2) \) queries to an oracle for \( H \), with gate complexity \( \poly(n, s) \).
	\end{enumerate}
\end{proposition}

The following proposition linearizes a Hamiltonian $K(t)$ to a Hamiltonian $H(t)$ that is linear in $t$. Note that the reduction \emph{preserves stoquasticity}. Its proof is implicit in those of \Cref{lem:pathB2}, \ref{lem:linearTFI}, and \ref{lem:pathC}.

\begin{proposition} \label{prop:merit_linearize}
	Let $K(t)$ be a $n$-qubit \apath\ that is explicitly $s$-sparse.
	Assume that for any $t\in[0,1]$, $K(t)$ has norm $\norm{K(t)}\le M$ and spectral gap at least $\delta$.
	Assume in addition that $K(t)$ is $L$-Lipschitz in $t$.

	Let $\lambda\le O(\delta/M)$ be arbitrary.
	For some $\ell=\Theta(L/\delta)$ and $\Delta\ge\poly(\ell,1/\lambda)$, define
	\begin{align}
		H(t)=\frac{\Delta}{\lambda}\left(\sum_{i=0}^\ell i(i+1-2t)\ketbra{i}{i}-\lambda\sum_{i=0}^{\ell-1}\left(\ketbra{i}{i+1}+\ketbra{i+1}{i}\right)\right) + \sum_{i=0}^\ell K(i/\ell)\tensor\ketbra{i}{i}.
	\end{align}
	Then the following holds.
	\begin{enumerate}[label=(\roman*)]
		\item For $t\in[0,1]$, $H(t)$ has spectral gap $\delta/4$, norm $\poly(1/\lambda,\ell)$, and sparsity at most $s+2$.
		\item Let $\ket{\psi}$ be a ground state of $K(0)$. Then $\ket{\psi} \tensor \ket{0}$ is $O(\lambda)$-close to a ground state of $H(0)$.
		\item Let $\ket{\psi'}$ be a ground state of $K(1)$. Then $\ket{\psi'}\tensor\ket{\ell}$ is $O(\lambda)$-close to a ground state of $H(1)$.
		\item An oracle for $H(t)$ can be efficiently implemented: a query to this oracle can be simulated using $O(1)$ queries to an oracle for $K(t)$, with gate complexity $\poly(n,s)$.
	\end{enumerate}
\end{proposition}
\begin{remark}
	We note that the statement can be strengthened such that $H(t)$ simulates $K(\tau)$ for some monotonically increasing function $\tau = \tau(t)$ satisfying $\tau(0)=0$ and $\tau(1)=1$, although certain parameter dependencies need to be adjusted.
\end{remark}

The following proposition transforms any hard instances exhibiting quantum advantage via adiabatic quantum optimization to hard instances exhibiting quantum advantage via Quantum Hamiltonian Descent (with the help of a quick-vanishing term $\nu(t)g(x)$).
Its proof is implicit in that of \Cref{thm:qhd}.

\begin{proposition} \label{prop:merit_dtc}
	Let $n,M,q$ be sufficiently large and $\epsilon,\delta>0$ be sufficiently small.
	Let $H_{\rm TFD}$ be an $n$-qubit TFD Hamiltonian such that $\norm{H_{\rm TFD}} \leq M$, and that all coefficients of $X_i$ in $H_{\rm TFD}$ are no more than $-1$.
	Let $D$ be an $n$-qubit diagonal Hamiltonian such that its eigenvalues lie in $[-1,0)$.
	Suppose that the \apath\ $H(t)$ of the form
	\begin{equation}
		H_{\rm TFD} \pathto D
	\end{equation}
	has quantum advantage given oracle access to $D$:
	\begin{enumerate}[label=(\roman*)]
		\item Any classical algorithm requires at least $q$ queries to find $u \in \binary^n$ that minimizes $\mel{u}{D}{u}$, with success probability greater than $1/q$.
		\item $H(t)$ has spectral gap $\delta$, norm $M$, and the ground state of $H(0) = H_{\rm TFD}$ is $\epsilon$-close to a product state.
	\end{enumerate}
	Then, there exist functions $g(x)$ and $\nu(t)$, and a family of functions $f \colon \mathcal{X} \to \mathbb{R}$ with a fixed box domain $\mathcal{X} \subset \mathbb{R}^n$ where the diameter of $\mathcal{X}$ is at most $\poly(M)$, such that the \apath\ $\hatH(t)$ of the form
	\begin{equation}
		\hatH(t) \coloneqq -\Delta + V(t,x),\quad V(t,x) \coloneqq \nu(t)g(x) + t f(x), \quad \forall t \geq 0, x \in \mathcal{X}
	\end{equation}
	has quantum advantage given oracle access to $f(\cdot)$:
	\begin{enumerate}[label=(\roman*)]
		\item Any classical algorithm requires at least $q$ queries to find an $x$ such that $f(x) - f(x^\ast) \leq \delta$, with success probability greater than $1/q$.
		\item $\hatH(t)$ has spectral gap at least $\frac{1}{\poly(n,1/\epsilon,1/\delta)}$, and the ground state of $\hatH(0)$ is $2\epsilon$-close to a product state.
		      Moreover, for any $\eta > 0$, there exists $\tend = \poly(M,1/\eta)$ such that $\langle \phi(\tend) |f|\phi(\tend)\rangle - f(x^\ast) \leq \eta$, where $\ket{\phi(t)}$ is the ground state of $\hatH(t)$.
		\item $f(x)$, $g(x)$, and $\nu(t)$ satisfy the following:
		      \begin{itemize}
			      \item $f(x)$ is bounded within $\pm \poly(n,M)$ and it is $\poly(n,M)$-Lipschitz. Moreover, $f(x)$ has $2^n$ local minima, and the function values at these minima correspond to the diagonal entries of the diagonal matrix $D$.
			      \item $g(x)$ is bounded within $\pm \poly(n,M, \log \frac{1}{\epsilon},\log \frac{1}{\delta})$ and it is $\poly(n,M, \log \frac{1}{\epsilon},\log \frac{1}{\delta})$-Lipschitz.
			      \item $\dfrac{\dee^k \nu}{\dee t^k} = e^{-\Theta(\sqrt{t})}$ for any $k$.
		      \end{itemize}
	\end{enumerate}

\end{proposition}
\begin{remark}
	The functions $f(x)$ and $g(x)$ can be made smooth (i.e., infinitely differentiable) by a more sophisticated construction of $\hatD$ in \Cref{def:hatD}, for example by convolving a mollifier with a more carefully designed base function.
	We chose not to present this approach formally, as the benefits are minimal and proofs of certain claims, e.g., \Crefitem{fct:hatD}{itm:fct:hatD_1}, would be more complicated.
\end{remark}
\end{document}